\documentclass[conference]{IEEEtran}
\IEEEoverridecommandlockouts
\usepackage{tikz}

\usepackage{amsmath}
\usepackage{amssymb}
\usepackage{amsfonts}
\usepackage{mathtools}         % adds \mathclap, etc.

\allowdisplaybreaks

\usepackage{hyphenat}          % provides \hyp for breakable hyphens

\newcommand{\runinhead}[1]{\noindent\textbf{#1}\nobreak\hspace{0.5em}\ignorespaces}

\usepackage[most]{tcolorbox}

\usepackage{graphicx}
\usepackage{adjustbox}
\usepackage{booktabs}
\usepackage{xcolor}
\usepackage{colortbl}
\definecolor{TableHeaderYellow}{HTML}{FFF8D6}
\colorlet{tabheader}{TableHeaderYellow}
\definecolor{TableOursGreen}{HTML}{F0FFF4}
\colorlet{oursrow}{TableOursGreen}
\definecolor{PalSystem}{HTML}{5B91C2}
\definecolor{PalSystemDark}{HTML}{3E6E9C}
\definecolor{PalExact}{HTML}{D27F74}
\graphicspath{{../figures/}{figures/}}
\usetikzlibrary{positioning}

\usepackage{multirow}
\usepackage{makecell}
\usepackage{array}
\usepackage{tabularx}
\usepackage{siunitx}
\usepackage{algorithm}
\usepackage{algpseudocode}     % do NOT also load algorithmic

\usepackage{stmaryrd}
\SetSymbolFont{stmry}{bold}{U}{stmry}{m}{n}
\usepackage{fontawesome5}
\usepackage{capt-of}
\usepackage{placeins}          %  to force floats to flush at section boundaries

\usepackage{cite}
\usepackage[colorlinks=true,
            linkcolor={blue!60!black},
            citecolor={blue!60!black},
            urlcolor={blue!60!black}]{hyperref}
\usepackage{xurl}

\makeatletter
\def\BibTeX{{\rm B\kern-.05em{\sc i\kern-.025em b}\kern-.08em
  T\kern-.1667em\lower.7ex\hbox{E}\kern-.125emX}}
\makeatother

\usepackage{amsthm}
\theoremstyle{plain}
\newtheorem{theorem}{Theorem}

\theoremstyle{definition}

\newcommand{\system}{FESC}

\newcommand{\mambabase}{Mamba-base}

\usepackage{tikz}
\newsavebox{\rwYesBox}
\newsavebox{\rwHalfBox}
\newsavebox{\rwNoBox}
\AtBeginDocument{%
  \savebox{\rwYesBox}{\begin{tikzpicture}[baseline=-0.55ex]\fill[black] (0,0) circle (0.48ex);\end{tikzpicture}}%
  \savebox{\rwHalfBox}{\begin{tikzpicture}[baseline=-0.55ex]\draw[black, line width=0.18pt] (0,0) circle (0.48ex);\fill[black] (0,0) -- (0,0.48ex) arc(90:-90:0.48ex) -- cycle;\end{tikzpicture}}%
  \savebox{\rwNoBox}{\begin{tikzpicture}[baseline=-0.55ex]\draw[black, line width=0.25pt] (0,0) circle (0.48ex);\end{tikzpicture}}%
}
\newcommand{\rwYes}{\usebox{\rwYesBox}}
\newcommand{\rwHalf}{\usebox{\rwHalfBox}}
\newcommand{\rwNo}{\usebox{\rwNoBox}}

\graphicspath{{figures/}{./}}
\begin{document}

\title{\system{}: Remodeling Long-Context Private Inference with Encrypted State-Space Models}

% --- arXiv: real authors, no NDSS review banner ---
\author{%
  \IEEEauthorblockN{Yufan Zhu\IEEEauthorrefmark{1},\quad
                    Chao Jin\IEEEauthorrefmark{2},\quad
                    Khin Mi Mi Aung\IEEEauthorrefmark{2},\quad
                    Xiaokui Xiao\IEEEauthorrefmark{1}}\\[0.6ex]
  \IEEEauthorblockA{%
    \IEEEauthorrefmark{1}National University of Singapore, Singapore\\
    \IEEEauthorrefmark{2}Agency for Science, Technology and Research, Singapore}
}

\maketitle
\pagestyle{plain}

% --- content: auto-synced from ../tex/sections (single source of truth) ---
\begin{abstract}
Processing long, sensitive documents with machine-learning models requires efficient, privacy-preserving long-context inference.
Prior private inference systems optimize or distribute encrypted Transformer
attention, but its quadratic token-pair work remains the bottleneck as
sequence length grows.
Selective state-space models (SSMs) offer linear-time recurrence, yet direct
encrypted implementation incurs linear multiplicative depth, sequence-wide state
residency, or dense FHE--MPC conversion.
We present Factorized Encrypted Scan-Contract (\system{}), a hybrid FHE--MPC system for
private long-context selective SSM inference.
Its factorized scan-contract keeps input-dependent transitions compact across
conversion boundaries, composes them without dense expansion, streams state
chunks on demand, and contracts outputs before conversion.
We demonstrate interface compatibility of the scan-contract implementation across invariant and selective SSM architectures.
For our Mamba-2 instantiation, we design GPU-optimized CKKS kernels for linear computations, MPC protocols for SiLU, softplus,
exponential, and RMSNorm, with approximation-aware fine-tuning.
To our knowledge, \system{} is the first private long-document inference
system to complete native end-to-end execution at $L{\ge}1{,}024$ on a single GPU\@.
At $L{=}2{,}048$, a $12$-layer \mambabase{} model completes inference in
$77.3$ minutes on one A100 GPU with a peak memory footprint of
$32.7$ GB, while maintaining near-plaintext accuracy on the evaluated
long-document tasks.
\end{abstract}

% !TEX root = ../main.tex
\section{Introduction}

The growing deployment of Machine Learning as a Service (MLaaS) enables users to outsource inference to the cloud without hosting models locally.
However, inputs such as legal contracts, medical records, and financial reports are sensitive and should remain hidden from the model provider~\cite{lee2019biobert,chalkidis2020legal,yang2023finchainbert}.
Private inference addresses this problem by evaluating the model without revealing the user's input, typically using cryptographic methods such as fully homomorphic encryption (FHE) and secure multi-party computation (MPC)~\cite{gilad-bachrach2016cryptonets,mohassel2017secureml,juvekar2018gazelle}.
Recent hybrid FHE--MPC systems make Transformer inference practical for short-context tasks~\cite{huang2022cheetah,pang2024bolt,lu2025bumblebee,xu2025blb,zhu2026encformer}, but their encrypted attention scales quadratically with sequence length in both operations and memory.
For workloads spanning hundreds or thousands of tokens, this quadratic cost dominates latency and can exceed available hardware memory~\cite{huang2026beyondlatency}.
Existing mitigations reduce this cost through token pruning or multi-GPU placement~\cite{zhang2025cipherprune,gong2026aegis}, improving practical efficiency while retaining quadratic attention over the executed tokens.
As Figure~\ref{fig:observations} illustrates, long-context private inference therefore needs a different encrypted sequence primitive.

\begin{figure}[t]
\centering
\includegraphics[width=\linewidth,page=1,trim=488pt 323pt 495pt 237.04pt,clip]{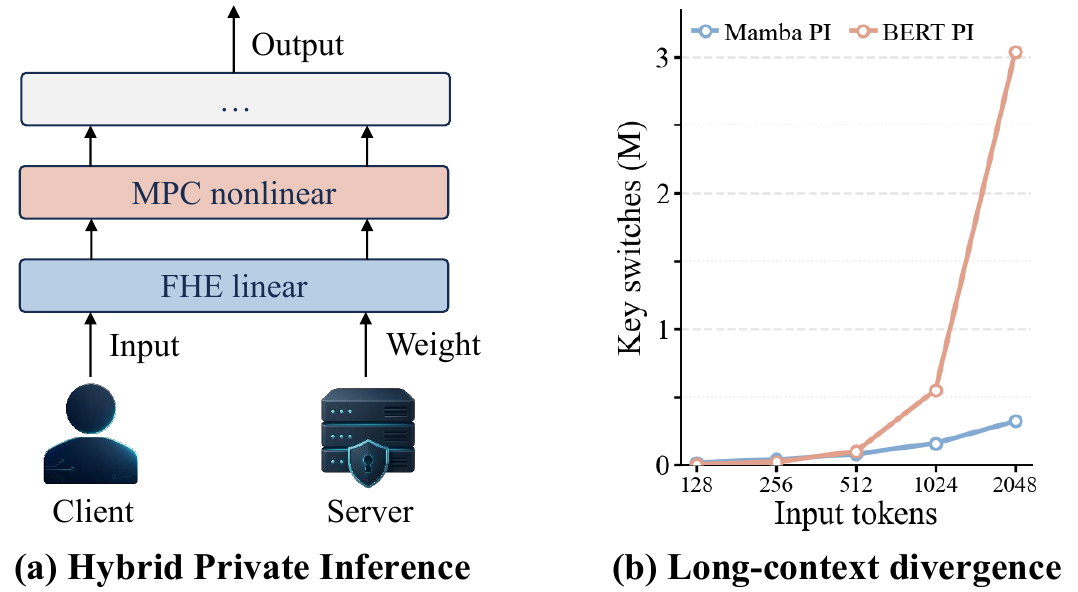}
\caption{\system{} motivation. (a) Hybrid FHE--MPC inference. (b) Encrypted Mamba key-switching scales linearly with sequence length, while encrypted BERT attention is quadratic.}
\label{fig:observations}
\end{figure}

Selective SSMs~\cite{gu2024mamba,dao2024mamba2} replace token-pair attention with an input-selective recurrence that scales linearly in sequence length and performs strongly on long-sequence tasks~\cite{gu2022efficiently}.
However, direct encrypted execution breaks this linear-scaling advantage in three ways, as summarized in Figure~\ref{fig:naive_failure}.
Sequentially evaluating the recurrence creates a multiplicative chain of length $\Theta(L)$ and exceeds the leveled-FHE multiplicative depth budget.
A parallel prefix scan~\cite{blelloch1990prefix} reduces the depth to $\Theta(\log L)$, but keeping all token-state ciphertexts live across the scan exceeds hardware memory at long context.
Also, a naive CKKS--MPC boundary materializes the affine transition ciphertexts before conversion, requiring full-state communication and high boundary payload.
These naive implementations remain impractical, so selective SSMs require redesigned FHE kernels and MPC protocols under encryption.

\begin{figure}[t]
\centering
\includegraphics[width=\linewidth,page=2,trim=620pt 238pt 382pt 111.04pt,clip]{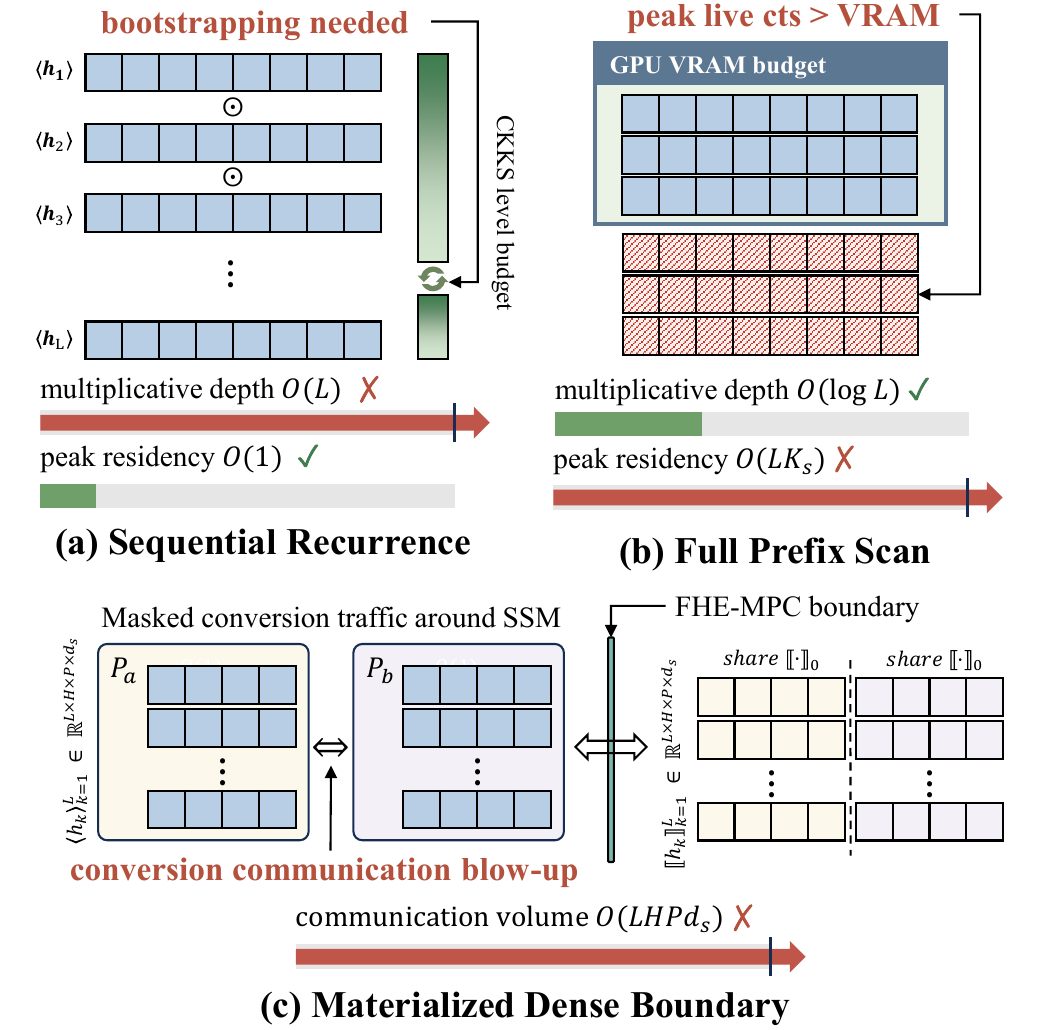}
\caption{Naive encrypted selective SSM failure modes. (a) Sequential recurrence exceeds the multiplicative depth budget. (b) A full prefix scan exceeds GPU memory. (c) A materialized dense boundary communicates the full recurrent state volume.}
\label{fig:naive_failure}
\end{figure}

We present Factorized Encrypted Scan-Contract (\system{}) as a hybrid FHE--MPC system for private long-context selective SSM inference with a logarithmic-depth encrypted dependency chain, resident encrypted state, and low online communication.
Its scan-contract primitive keeps compact private factors, composes the same affine recurrence under CKKS, and exports contracted outputs.
We deploy \system{} on Mamba-2 as the headline end-to-end model, evaluate it against prior encrypted inference systems on matched tasks, and validate the same scan-contract implementation across representative invariant and selective SSM recurrences.
This establishes the long-context cost advantage of encrypted selective SSM inference while retaining high task accuracy, and shows that the same encrypted primitive applies beyond Mamba-2.
\system{} is the first system to realize encrypted input-selective SSM inference at long context, breaking the quadratic complexity bottleneck that has confined prior private systems to short sequences.

We make the following contributions.

\begin{enumerate}
\item \textbf{FHE Scan-Contract Algorithm.}
\system{} introduces a factorized FHE scan-contract algorithm for blocked private input-selective affine recurrences, formalized in \S\ref{sec:ks-he} and evaluated in \S\ref{sec:performance}, that resolves all three naive implementation failures. A parallel-summary scan has logarithmic dependency depth and linear encrypted work, the resident schedule cuts the $L{=}2048$ live-state footprint from $59.2$\,GB to $12.8$\,GB, and compact boundary factors compose selective decay without dense state expansion.

\item \textbf{Specialized MPC Nonlinear Protocols.}
\system{} designs MoSiLU, MoSoft, NeRMS, and MMExp, four MPC protocols formalized in \S\ref{sec:mpc} that replace generic elementary-function circuits with structure-specific polynomial, comparison, mux, and Newton components. Their fixed data-independent schedules reduce MPC online operations and rounds, yielding up to a $7.7\times$ online nonlinear-protocol cost reduction in \S\ref{sec:performance}.

\item \textbf{End-to-End GPU Implementation and Evaluation.}
\system{} builds an end-to-end private inference system covering $L\in\{128,256,512,1024,2048,4096\}$ with GPU-optimized kernels, detailed in \S\ref{sec:experiments}, running on a single GPU where prior hybrid systems~\cite{juvekar2018gazelle,huang2022cheetah,pang2024bolt,lu2025bumblebee,xu2025blb,zhu2026encformer} require multi-GPU execution or remain limited to short sequences, while also scaling to four GPUs. On a single A100, \system{} runs $L{=}2048$ in $77.3$ minutes, faster than the reported four-GPU AEGIS~\cite{gong2026aegis} result, and a four-A100 path reaches $36.0$ minutes. \S\ref{sec:performance} reports latency, memory, communication, accuracy, approximation, ablations, model-shape sweeps, and cross-GPU portability, separate from the cross-family SSM implementations in Appendix~\ref{app:cross-ssm-family}.
\end{enumerate}

% !TEX root = ../main.tex
\section{Background and Threat Model}
\label{sec:preliminaries}

\subsection{Fully Homomorphic Encryption and CKKS}
\label{sec:prelim_crypto}

A public-key FHE scheme provides $\mathsf{KeyGen}$, $\mathsf{Enc}$, $\mathsf{Dec}$, and $\mathsf{Eval}$ such that an evaluator can compute on ciphertexts without seeing the plaintext
\[
  \mathsf{Dec}_{sk}
  \bigl(\mathsf{Eval}_{evk}(f,\mathsf{Enc}_{pk}(x))\bigr)
  \approx f(x).
\]
For exact schemes such as BGV~\cite{brakerski2012leveled}, the output is exact modulo a plaintext ring. The CKKS scheme~\cite{cheon2017ckks} instead provides approximate arithmetic over complex values, with error from encoding, rounding, and noise growth. A leveled CKKS instantiation fixes a modulus chain $Q=\prod_i q_i$ for a target multiplicative depth, consumed by evaluation until the level budget is exhausted, with bootstrapping available to refresh it at the cost of substantial latency and additional approximation error. Since selective SSM inference is finite-precision and dominated by real-valued linear algebra, \system{} follows prior hybrid private-inference systems~\cite{juvekar2018gazelle,huang2022cheetah,pang2024bolt,lu2025bumblebee,xu2025blb} and uses leveled CKKS without bootstrapping, as detailed in Appendix~\ref{app:he_choice}.

CKKS encrypts vectors over the cyclotomic ring $R_Q=\mathbb{Z}_Q[X]/(X^N+1)$ and packs $n_{\mathrm{ckks}}=N/2$ complex slots into one ciphertext for SIMD-style evaluation. A key-switching rewrites a ciphertext after an automorphism or multiplication so it remains usable under the evaluation keys. Slot rotations, conjugation, and relinearization require auxiliary special-prime moduli and key-switching through public evaluation keys, which dominate packed GPU CKKS kernels. Separately, plaintext--ciphertext and ciphertext--ciphertext multiplications raise the scale and require rescaling that drops one modulus prime and one level, while additions, rotations, conjugation, and relinearization do not by themselves consume levels. Packing conventions and GPU primitive latencies are given in Appendix~\ref{app:he-shifts} and Table~\ref{tab:he_primitive_bench}.

\subsection{Secure Multiparty Computation}
\label{sec:prelim_mpc}

Secure multiparty computation allows parties to jointly evaluate a function over private inputs while revealing only the intended output.
\system{} uses two-party arithmetic secret sharing to evaluate the nonlinear components of selective SSM inference without exposing intermediate activations.
A real value $r$ is encoded in fixed point over $\mathbb{Z}_{2^\ell}$ and split into an additive share pair $[[r]]=(r_0,r_1)$ with $r_0+r_1=r\bmod 2^\ell$, so neither party holds $r$ in the clear.
Additions and public-scalar operations are local and free.
Multiplying two secret-shared values uses a precomputed Beaver triple~\cite{beaver1992efficient} $([[a]],[[b]],[[c]])$ with $c=ab$. Given $[[x]]$ and $[[y]]$, the parties open $\alpha=x-a$ and $\beta=y-b$, then compute
\[
  [[xy]] = [[c]]+\alpha[[b]]+\beta[[a]]+\alpha\beta,
\]
where $\alpha\beta$ is a public correction absorbed into one share.
Fixed-point multiplication is followed by secure truncation to restore scale $2^F$.
Secret-share multiplication and secure comparison introduce interaction, so \S\ref{sec:mpc-approx} reports both critical-path rounds and total interactive operation counts.
We evaluate SiLU, softplus, the discretization exponential, and RMSNorm under these primitives in the semi-honest real--ideal model~\cite{yao1986how,goldreich1987how}, with full protocol specifications in \S\ref{sec:mpc}.

\subsection{Threat Model and Security Goal}
\label{sec:prelim_model}

We follow the standard two-party semi-honest private-inference setting for
MLaaS~\cite{juvekar2018gazelle,huang2022cheetah,pang2024bolt,lu2025bumblebee,xu2025blb,zhu2026encformer}.
The client $P_0$ holds the private input $x$ and the CKKS secret key. The
server $P_1$ holds the private model parameters $\theta$ and the CKKS public
and evaluation keys. We assume one static semi-honest corruption without
collusion, where the corrupted party follows the protocol but may inspect its
full execution view.

Security is defined for the fixed deployed encrypted computation
$f^{\mathrm{dep}}_\theta$, including its approximation, fixed-point encoding, packing,
and FHE--MPC schedule, with $\mu$ covering the architecture, tensor shapes,
sequence lengths, CKKS parameters, approximation degrees, bucket counts, loop
bounds, and fixed execution schedule, treating numerical approximation error
separately from protocol privacy. The protocol achieves two-sided privacy, so
a corrupted server learns nothing about $x$ beyond $\mu$, message sizes, timing
induced by the fixed public schedule, and abort status, while a corrupted client
learns nothing about $\theta$ beyond its output $y$ and the same public
quantities.

We formalize the ideal functionality as
\[
  \mathcal{F}_{\mathrm{SelSSM}}\langle x,\theta,\mu\rangle
  \mapsto
  \langle f^{\mathrm{dep}}_\theta(x),\bot\rangle ,
\]
where $x$ is supplied by $P_0$, $\theta$ is supplied by $P_1$, and $\mu$ is
public. The first output is delivered to $P_0$, while $P_1$ receives $\bot$.

\begin{theorem}[End-to-end semi-honest security]
\label{thm:main-security}
Assume semantic security of CKKS, semi-honest simulation security of the MPC
primitives, semi-honest simulation security of the CKKS--MPC conversion
protocols, and valid input-independent preprocessing correlations. Then
$\Pi_{\mathsf{Infer}}$ securely realizes $\mathcal{F}_{\mathrm{SelSSM}}$ for
the deployed function $f^{\mathrm{dep}}_\theta$ against one static semi-honest
corruption, leaking only $\mu$, message sizes and timing induced by the fixed
public schedule, and abort status.
\end{theorem}

The proof appears in Appendix~\ref{app:pipeline-security}. Appendix~\ref{app:offline-budget}
quantifies the required input-independent preprocessing.

\subsection{State Space Models}
\label{sec:ssm_model}

Structured state space models map an input sequence through a latent state governed by a linear dynamical system,
\begin{align*}
  h'(t) &= \mathbf A h(t)+\mathbf B u(t),\\
  y(t)  &= \mathbf C h(t)+D_{\mathrm{skip}}u(t),
\end{align*}
where $\mathbf A$ is the state transition, $\mathbf B$ and $\mathbf C$ are input and output maps, and $D_{\mathrm{skip}}$ is a skip parameter.
After discretization with step size $\delta$, the model becomes an affine recurrence
\[
  h_k=\bar{\mathbf A}h_{k-1}+\bar{\mathbf B}u_k,
  \qquad
  y_k=\mathbf C h_k+D_{\mathrm{skip}}u_k,
\]
with $\bar{\mathbf A}=\exp(\delta\mathbf A)$ and the corresponding zero-order-hold input map $\bar{\mathbf B}$.
This affine form exposes the decay and input-update factors that the scan-contract primitive composes under encryption.
Appendix~\ref{app:mamba-details} gives the closed-form discretization details used by Mamba-style selective SSMs.

\subsection{Mamba-2 Recurrence}
\label{sec:mamba_model}

Mamba-1~\cite{gu2024mamba} makes the SSM parameters input-dependent, allowing the recurrence to select information as it scans the sequence, and Mamba-2~\cite{dao2024mamba2} exposes a grouped scan form that is especially useful for encrypted execution.
\system{} uses Mamba-2 as the deployed model, while Appendix~\ref{app:cross-ssm-family} instantiates the same scan-packet view for representative invariant and selective SSM implementations, including Mamba-3.
For each token $k$, the Mamba-2 mixer produces compact factors, illustrated in Figure~\ref{fig:mamba-arch}. The scan input is
\[
  x_k[h,p]=\Delta_k[h]x^{\mathrm{raw}}_k[h,p],
\]
with per-head decay $a_k[h]$, grouped factors $B_k,C_k\in\mathbb{R}^{G\times d_s}$, raw skip branch $x^{\mathrm{raw}}_k$, and output gate $g^{\mathrm{out}}_k$, whose full construction appears in Appendix~\ref{app:mamba-details}.

Let $g_h=g(h)\in\{1,\ldots,G\}$ be the group assigned to head $h$, with $H$ heads, per-head width $P$, and state dimension $d_s$.
The grouped Mamba-2 recurrence is
\begin{align*}
  h_k[h,p,i] &= a_k[h]\,h_{k-1}[h,p,i]+x_k[h,p]\,B_k[g_h,i],\\
  y_k[h,p] &= \sum_i h_k[h,p,i]\,C_k[g_h,i]
  +D_{\mathrm{skip}}[h,p]\,x^{\mathrm{raw}}_k[h,p],\\
  \tilde y_k[h,p] &= y_k[h,p]\,g^{\mathrm{out}}_k[h,p].
\end{align*}

Thus each $(h,p)$ channel maintains a $d_s$-dimensional recurrent state, while heads assigned to the same group share the dynamic $B_k$ and $C_k$ factors.
\system{} preserves this model-side factorization, where the scan input $x_k$ is compact over $HP$ channels, the decay $a_k$ is compact over heads, the additive update factors as $x_k\otimes_g B_k$, the state is contracted with $C_k$ before the skip and gate path, and the block output is $o_k=\hat y_k W_{\mathrm{out}}$ with $\hat y_k=\operatorname{RMSNorm}_{\boldsymbol\gamma_{\mathrm{rms}}}(\tilde y_k)$ flattened to length $HP$.

\begin{figure}[t]
\centering
\includegraphics[width=0.8\columnwidth,page=4,trim=530pt 277pt 495pt 227.04pt,clip]{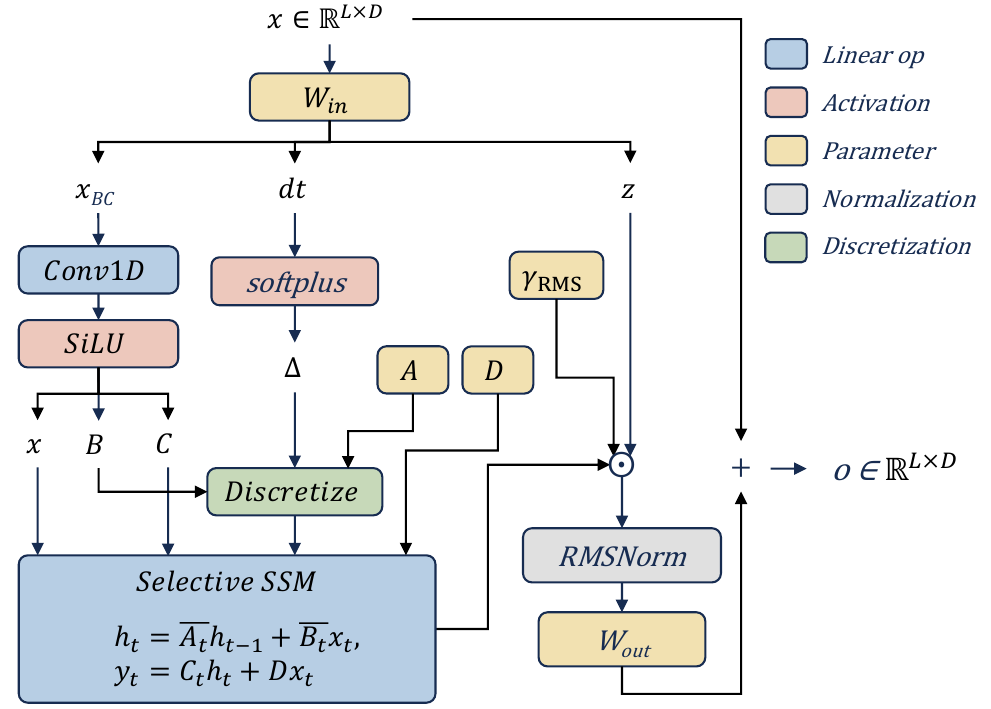}
\caption{Illustration of the Mamba-2 SSM mixer, producing the scan factors $(x,a,B,C)$ and post-SSM factors $x^{\mathrm{raw}},g^{\mathrm{out}}$.}
\label{fig:mamba-arch}
\end{figure}

\subsection{Parallel Prefix Scans}
\label{sec:prelim_prefix_scan}

The Mamba-2 recurrence can be evaluated as a parallel prefix scan, where each token update defines an affine map $F_k(z)=a_k z+s_k$ with $s_k=x_k\otimes_g B_k$, where $\otimes_g$ denotes the grouped outer product induced by $g(h)$.
Composing two such maps gives
\[
  (a_L,s_L)\bullet(a_R,s_R)
  =
  (a_{R}a_{L},\;a_{R}s_{L}+s_{R}),
\]
which is associative because it is function composition, so an inclusive scan
\[
  T_k=t_1\bullet t_2\bullet\cdots\bullet t_k,
  \qquad k=1,\ldots,L,
\]
over the token-local maps produces the same states as sequential evaluation in $\Theta(\log L)$ dependency depth.
Among classic scan networks~\cite{kogge1973parallel,sklansky1960conditional,brent1982regular}, \system{} uses the Brent--Kung scan, shown in Figure~\ref{fig:scan-tree}, because it achieves the fastest encrypted runtime by minimizing total composition nodes at logarithmic depth, since each composition is expensive under CKKS\@.
Appendix~\ref{app:scan-comparison} compares the scan topologies.

\begin{figure}[t]
\centering
\includegraphics[width=0.8\columnwidth,page=3,trim=380pt 304pt 478pt 260.04pt,clip]{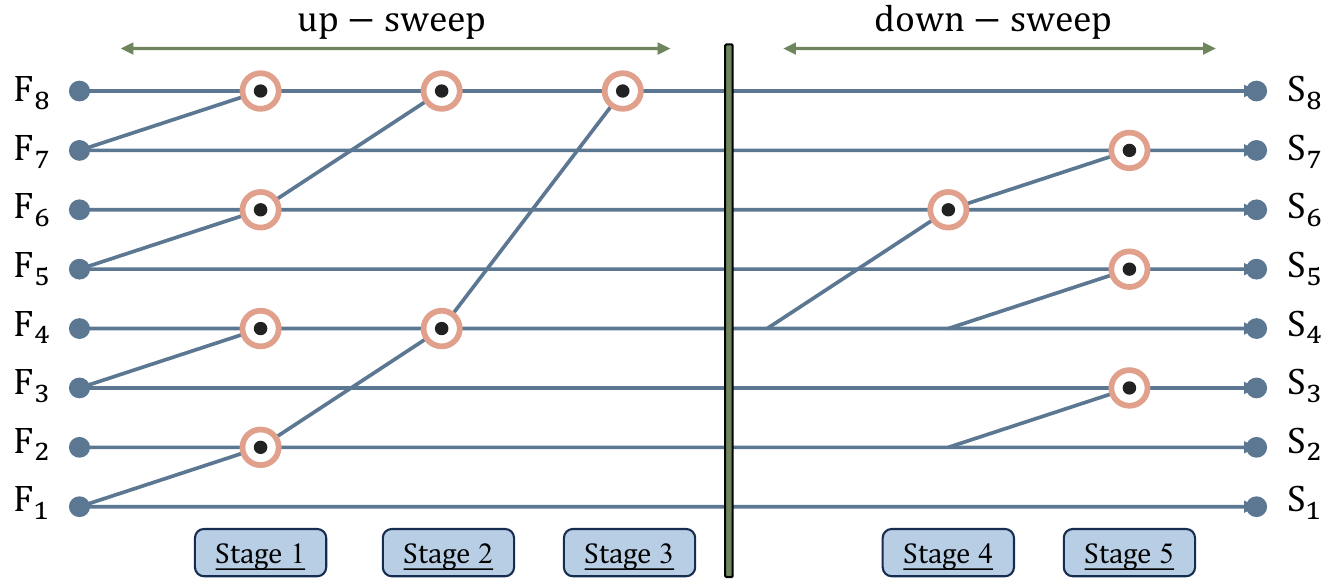}
\caption{Brent--Kung prefix network for $L{=}8$. Each node applies the associative affine-composition operator.}
\label{fig:scan-tree}
\end{figure}

% !TEX root = ../main.tex
\section{Related Work}
\label{sec:related}

\newcommand{\fescbackend}[1]{%
  \rotatebox[origin=c]{90}{\bfseries #1}}
\newcommand{\fescrot}[1]{%
  \rlap{\hspace*{-1.6em}\rotatebox[origin=lB]{45}{#1}}}

\begin{table*}[t]
\centering
\caption{Comparison of private-inference systems.}
\label{tab:private_sequence_landscape}
\scriptsize
\setlength{\tabcolsep}{5pt}
\renewcommand{\arraystretch}{1.03}

\begin{adjustbox}{max width=\textwidth}
\begin{tabular}{@{}c|l|ccc@{\hspace{18pt}}ccc@{\hspace{18pt}}ccccc@{\hspace{18pt}}cccccc@{\hspace{18pt}}cccc@{\hspace{18pt}}cccccc@{}}
\toprule
\rowcolor{tabheader}
\multicolumn{2}{c|}{}
& \multicolumn{3}{c}{\textbf{Properties}}
& \multicolumn{3}{c}{\textbf{Workloads}}
& \multicolumn{5}{c}{\textbf{Sequence Design}}
& \multicolumn{6}{c}{\textbf{Nonlinear Methods}}
& \multicolumn{4}{c}{\textbf{Evaluation}}
& \multicolumn{6}{c}{\textbf{Systems}} \\
\midrule
\rule{0pt}{6ex}\textbf{Family}
& \multicolumn{1}{c|}{\textbf{System}}
& \fescrot{Setting}
& \fescrot{Sec.\ model}
& \fescrot{Artifact}
& \fescrot{Text}
& \fescrot{Vision}
& \fescrot{Long-doc}
& \fescrot{Attention}
& \fescrot{SSM}
& \fescrot{$O(L)$}
& \fescrot{C-Cvt.}
& \fescrot{NoBS}
& \fescrot{Poly.}
& \fescrot{Piece.}
& \fescrot{Iter.}
& \fescrot{LUT/FSS}
& \fescrot{FS/ODE}
& \fescrot{Adapt.}
& \fescrot{Latency}
& \fescrot{Memory}
& \fescrot{Enc.\ acc.}
& \fescrot{Ablations}
& \fescrot{GPU Opts.}
& \fescrot{Multi-GPU}
& \fescrot{Network}
& \fescrot{Breadth}
& \fescrot{$L_{\mathrm{run}}$}
& \fescrot{$L_{\mathrm{acc}}$} \\
\midrule
\multirow{7}{*}{\fescbackend{FHE-only}} & THE-X~\cite{chen2022thex}
& 2P/CA & \rwNo & \rwHalf & \rwYes & \rwNo & \rwNo & \rwYes & \rwNo & \rwNo & -- & \rwYes & \rwNo & \rwNo & \rwNo & \rwNo & \rwNo & \rwYes & \rwNo & \rwNo & \rwYes & \rwYes& \rwNo & \rwNo & -- & \rwNo & NR & NR  \\

 & NEXUS~\cite{zhang2025nexus}
& 2P/NI & \rwYes & \rwYes & \rwYes & \rwNo & \rwNo & \rwYes & \rwNo & \rwNo & -- & \rwNo & \rwYes & \rwYes & \rwYes & \rwNo & \rwNo & \rwNo & \rwHalf & \rwNo & \rwYes & \rwYes& \rwYes & \rwNo & -- & \rwYes & 128t & 128t  \\

 & THOR~\cite{moon2025thor}
& 2P/NI & \rwHalf & \rwYes & \rwYes & \rwNo & \rwNo & \rwYes & \rwNo & \rwNo & -- & \rwNo & \rwYes & \rwNo & \rwYes & \rwNo & \rwNo & \rwNo & \rwYes & \rwNo & \rwYes & \rwHalf& \rwYes & \rwNo & -- & \rwNo & 128t & 128t  \\

 & Powerformer~\cite{park2025powerformer}
& 2P/NI & \rwHalf & \rwYes & \rwYes & \rwNo & \rwNo & \rwYes & \rwNo & \rwNo & -- & \rwNo & \rwYes & \rwNo & \rwNo & \rwNo & \rwNo & \rwYes & \rwYes & \rwNo & \rwYes & \rwYes& \rwYes & \rwNo & -- & \rwNo & 128t & 128t  \\

 & EncryptedLLM~\cite{decastro2025encryptedllm}
& 2P/NI & \rwHalf & \rwYes & \rwYes & \rwNo & \rwNo & \rwYes & \rwNo & \rwNo & -- & \rwNo & \rwYes & \rwYes & \rwYes & \rwNo & \rwNo & \rwNo & \rwHalf & \rwNo & \rwHalf & \rwYes& \rwYes & \rwNo & -- & \rwNo & 128t & NR  \\

 & AEGIS~\cite{gong2026aegis}
& 2P/NI & \rwHalf & \rwNo & \rwYes & \rwNo & \rwNo & \rwYes & \rwNo & \rwNo & -- & \rwNo & \rwYes & \rwNo & \rwNo & \rwNo & \rwNo & \rwNo & \rwYes & \rwYes & \rwYes & \rwYes& \rwYes & \rwYes & -- & \rwNo & 2048t & NR  \\

 & MOAI~\cite{zhang2026moai}
& 2P/NI & \rwNo & \rwYes & \rwYes & \rwNo & \rwNo & \rwYes & \rwNo & \rwNo & -- & \rwNo & \rwYes & \rwNo & \rwYes & \rwNo & \rwNo & \rwNo & \rwYes & \rwNo & \rwYes & \rwHalf & \rwYes & \rwNo & -- & \rwHalf & 128t & 128t  \\
\midrule
\multirow{6}{*}{\fescbackend{MPC-only}} & MPCFormer~\cite{li2023mpcformer}
& 2P+P & \rwHalf & \rwYes & \rwYes & \rwNo & \rwYes & \rwYes & \rwNo & \rwNo & -- & -- & \rwYes & \rwNo & \rwNo & \rwNo & \rwNo & \rwYes & \rwYes & \rwNo & \rwHalf & \rwYes& \rwHalf & \rwNo & \rwHalf & \rwYes & 512t & 512t  \\

 & SIGMA~\cite{gupta2024sigma}
& 2P+P & \rwYes & \rwYes & \rwYes & \rwNo & \rwNo & \rwYes & \rwNo & \rwNo & -- & -- & \rwNo & \rwYes & \rwNo & \rwYes & \rwNo & \rwNo & \rwYes & \rwHalf & \rwYes & \rwYes& \rwYes & \rwNo & \rwYes & \rwYes & 1024t & 128t  \\

 & SecFormer~\cite{luo2024secformer}
& 2P+P & \rwYes & \rwYes & \rwYes & \rwNo & \rwNo & \rwYes & \rwNo & \rwNo & -- & -- & \rwYes & \rwYes & \rwYes & \rwNo & \rwYes & \rwYes & \rwYes & \rwNo & \rwHalf & \rwYes& \rwHalf & \rwNo & \rwHalf & \rwYes & 512t & NR  \\

 & PUMA~\cite{dong2025puma}
& 3P & \rwYes & \rwYes & \rwYes & \rwNo & \rwNo & \rwYes & \rwNo & \rwNo & -- & -- & \rwYes & \rwYes & \rwYes & \rwNo & \rwNo & \rwNo & \rwYes & \rwNo & \rwYes & \rwYes& \rwNo & \rwNo & \rwHalf & \rwYes & 256t & 128t  \\

 & SHAFT~\cite{kei2025shaft}
& 2P+P & \rwYes & \rwYes & \rwYes & \rwYes & \rwNo & \rwYes & \rwNo & \rwNo & -- & -- & \rwHalf & \rwYes & \rwYes & \rwNo & \rwYes & \rwNo & \rwYes & \rwNo & \rwYes & \rwYes& \rwHalf & \rwNo & \rwYes & \rwYes & 128t & 128t  \\

 & MPCMamba~\cite{yu2025mpcmamba}
& 2P+P & \rwHalf & \rwNo & \rwNo & \rwYes & \rwNo & \rwNo & \rwYes & \rwYes & -- & -- & \rwYes & \rwNo & \rwYes & \rwNo & \rwNo & \rwNo & \rwYes & \rwNo & \rwHalf & \rwNo& \rwHalf & \rwHalf & \rwNo & \rwHalf & NR & NR  \\
\midrule
\multirow{7}{*}{\fescbackend{FHE--MPC}} & Iron~\cite{hao2022iron}
& 2P & \rwYes & \rwNo & \rwYes & \rwNo & \rwNo & \rwYes & \rwNo & \rwNo & \rwNo & \rwYes & \rwNo & \rwNo & \rwHalf & \rwYes & \rwNo & \rwNo & \rwYes & \rwNo & \rwYes & \rwYes& \rwNo & \rwNo & \rwHalf & \rwYes & 128t & 128t  \\

 & BOLT~\cite{pang2024bolt}
& 2P & \rwHalf & \rwYes & \rwYes & \rwNo & \rwNo & \rwYes & \rwNo & \rwNo & \rwNo & \rwYes & \rwYes & \rwYes & \rwHalf & \rwNo & \rwNo & \rwYes & \rwYes & \rwNo & \rwHalf & \rwYes& \rwNo & \rwNo & \rwYes & \rwNo & 128t & 128t  \\

 & BumbleBee~\cite{lu2025bumblebee}
& 2P & \rwYes & \rwYes & \rwYes & \rwYes & \rwNo & \rwYes & \rwNo & \rwNo & \rwNo & \rwYes & \rwYes & \rwYes & \rwYes & \rwNo & \rwNo & \rwNo & \rwYes & \rwNo & \rwYes & \rwYes& \rwNo & \rwNo & \rwYes & \rwYes & 128t & 128t  \\

 & BLB~\cite{xu2025blb}
& 2P & \rwYes & \rwYes & \rwYes & \rwNo & \rwNo & \rwYes & \rwNo & \rwNo & \rwHalf & \rwYes & \rwYes & \rwYes & \rwYes & \rwNo & \rwNo & \rwNo & \rwYes & \rwNo & \rwYes & \rwYes& \rwYes & \rwNo & \rwYes & \rwYes & 128t & 128t  \\

 & CipherPrune~\cite{zhang2025cipherprune}
& 2P & \rwYes & \rwYes & \rwYes & \rwNo & \rwNo & \rwYes & \rwNo & \rwNo & \rwNo & \rwYes & \rwYes & \rwYes & \rwYes & \rwNo & \rwNo & \rwYes & \rwYes & \rwNo & \rwHalf & \rwYes& \rwNo & \rwNo & \rwYes & \rwYes & 512t & 128t  \\

 & EncFormer~\cite{zhu2026encformer}
& 2P & \rwYes & \rwYes & \rwYes & \rwNo & \rwNo & \rwYes & \rwNo & \rwNo & \rwYes & \rwYes & \rwYes & \rwYes & \rwNo & \rwNo & \rwNo & \rwYes & \rwYes & \rwHalf & \rwYes & \rwYes& \rwYes & \rwNo & \rwYes & \rwYes & 128t & 128t  \\

\rowcolor{oursrow}
\cellcolor{white} & \system{}
& 2P & \rwYes & \rwYes & \rwYes & \rwNo & \rwYes & \rwNo & \rwYes & \rwYes & \rwYes & \rwYes & \rwYes & \rwYes & \rwYes & \rwNo & \rwNo & \rwYes & \rwYes & \rwYes & \rwYes & \rwYes& \rwYes & \rwYes & \rwYes & \rwYes & \textbf{8192t}$^{\S}$ & \textbf{2048t}  \\
\bottomrule
\end{tabular}
\end{adjustbox}

\vspace{0.1ex}
\begin{minipage}{0.995\textwidth}
\tiny
\rwYes: present; \rwHalf: partial; \rwNo: absent or unreported; ``--'': not applicable; NR: not reported or released.
Sec. model: explicit threat model and security proof; CA: client-assisted; NI: non-interactive; C-Cvt.: complex and compact FHE--MPC conversion;
NoBS: no bootstrapping; LUT/FSS: lookup-table or function-secret-sharing; FS/ODE: Fourier-series or ODE-based.
Multi-GPU: one inference across GPUs.
Network: multiple network profiles; Breadth: multiple models, sizes, or architectures.
$L_{\mathrm{run}}$/$L_{\mathrm{acc}}$: largest paper-reported encrypted-runtime and encrypted-accuracy lengths; $t$, $p$: tokens, patches;
$^{\S}$: $L{=}8192$ with a secure carry refresh;  baseline scope: Appendix~\ref{app:baseline_sources}.
\end{minipage}
\end{table*}

\runinhead{Private SSM Inference.} Structured state-space models such as S4~\cite{gu2022efficiently},
Mamba~\cite{gu2024mamba}, and Mamba-2~\cite{dao2024mamba2}
provide linear-time alternatives to attention for long-context sequence
processing.
MPCMamba~\cite{yu2025mpcmamba} establishes secure Mamba inference for short
image-patch vision sequences under an MPC-only system. \system{} is the first
hybrid FHE--MPC system to realize private selective SSM inference on
long-document text.

\runinhead{Private Transformer Systems.}
Private Transformer systems fall into three backend families.
THE-X~\cite{chen2022thex} is client-assisted, while non-interactive
FHE-only systems including NEXUS~\cite{zhang2025nexus},
THOR~\cite{moon2025thor}, Powerformer~\cite{park2025powerformer},
EncryptedLLM~\cite{decastro2025encryptedllm},
MOAI~\cite{zhang2026moai}, and AEGIS~\cite{gong2026aegis}
evaluate the model homomorphically,
replacing nonlinearities with HE-compatible approximations and managing
depth, packing, key-switching, and, where used, bootstrapping.
MPC-only systems, including MPCFormer~\cite{li2023mpcformer},
PUMA~\cite{dong2025puma}, SecFormer~\cite{luo2024secformer},
SIGMA~\cite{gupta2024sigma}, and
SHAFT~\cite{kei2025shaft}, evaluate the model through secret sharing
together with arithmetic, comparison, and, where used, function-secret-sharing
protocols for attention, normalization, activations,
quantization, and fixed-point arithmetic, avoiding homomorphic depth
but incurring interaction, communication, and preprocessing.
Hybrid systems, including Iron~\cite{hao2022iron},
BOLT~\cite{pang2024bolt}, BumbleBee~\cite{lu2025bumblebee},
BLB~\cite{xu2025blb}, and EncFormer~\cite{zhu2026encformer},
assign wide linear algebra to FHE and nonlinear operators to MPC,
reducing deep homomorphic nonlinear evaluation and communication-heavy
MPC matrix multiplication, but adding conversion and packing boundaries.
These systems establish practical private Transformer inference for short
contexts, but encrypted attention keeps their rotation, key-switch,
communication, or materialization costs quadratic in sequence length.

\runinhead{Long-Context Strategies.}
Recent systems extend private Transformer inference without replacing
attention.
CipherPrune~\cite{zhang2025cipherprune} adaptively prunes encrypted
tokens layer by layer and assigns lower-degree polynomials to less
important retained tokens, reporting GPT-2 runtime scaling through
$L{=}512$.
AEGIS~\cite{gong2026aegis} retains full BERT-base self-attention and
scales FHE execution through $L{=}2048$ on four GPUs using
application- and encryption-aware placement together with
polynomial-operator reordering.
CipherPrune shortens the sequence processed by later layers, whereas
AEGIS partitions and overlaps the dense-attention workload.
Neither replaces attention with a full-token linear-work sequence
operator.
\system{} instead preserves every input token and replaces token-pair
attention with an associative selective recurrence, as positioned
against each family in Table~\ref{tab:private_sequence_landscape}.

% !TEX root = ../main.tex
\section{System Overview}
\label{sec:overview}

\begin{figure}[t]
\centering
\includegraphics[width=\columnwidth,keepaspectratio,page=5,trim=508pt 115pt 520pt 87.04pt,clip]{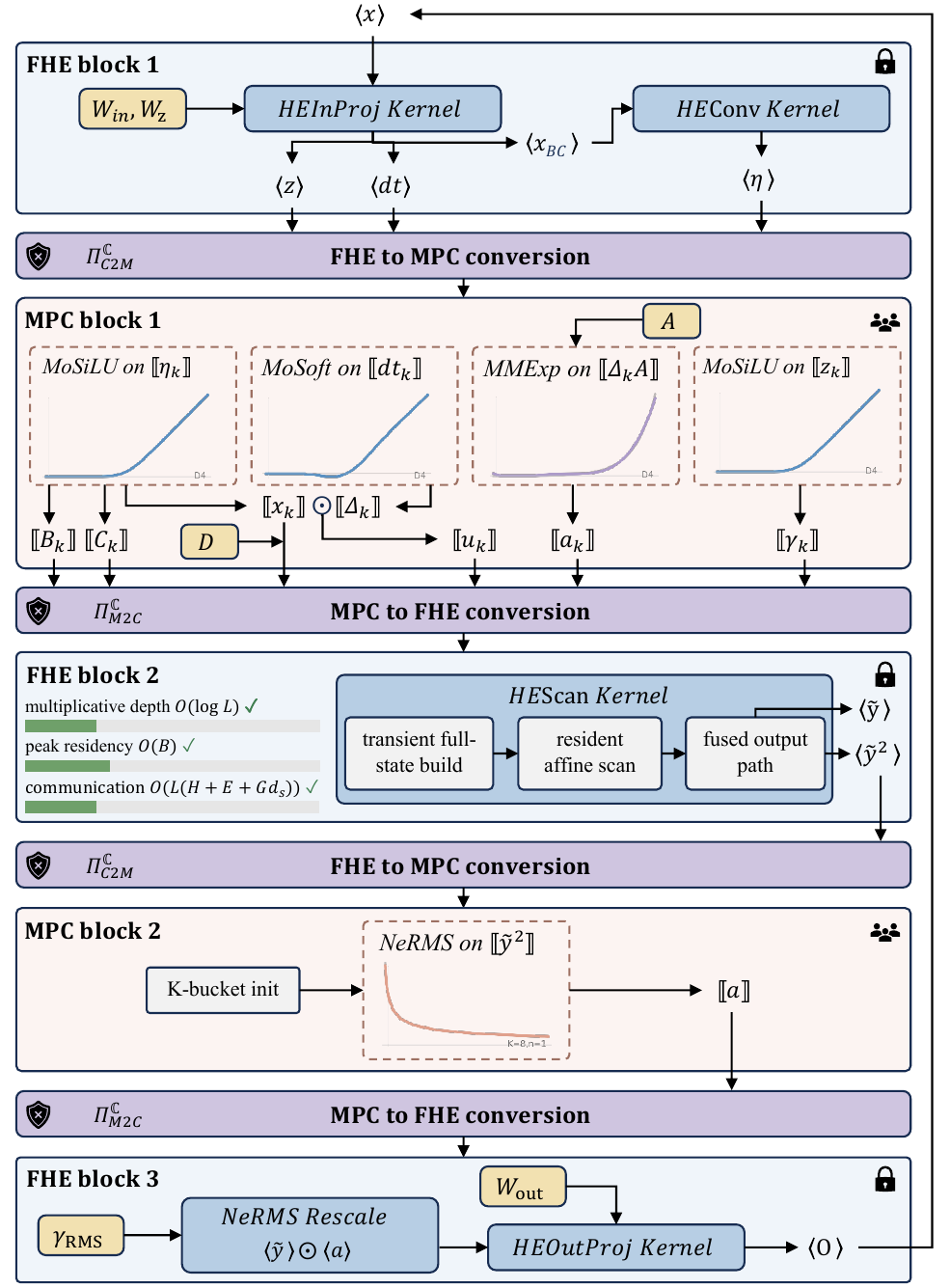}
\caption{\system{} per-block pipeline alternating CKKS linear kernels and MPC nonlinear computations.}
\label{fig:pipeline}
\end{figure}

The \system{} block follows the alternating FHE--MPC pipeline in Figure~\ref{fig:pipeline}, with encrypted kernels developed in \S\ref{sec:ks-he}, nonlinear protocols in \S\ref{sec:mpc}, and the shared SSM interface in Appendix~\ref{app:cross-ssm-family}.
It runs HEInProj and HEConv in CKKS, uses MoSiLU, MoSoft, and MMExp in MPC to build the compact HEScan packet, evaluates the selective recurrence in CKKS, then applies the post-scan skip, gate, and NeRMS path before HEOutProj emits the encrypted block output.
HEOutProj feeds the next \system{} block as a CKKS ciphertext, so the client decrypts only the final prediction.
The packet contains the private token input, decay, and $\mathbf B,\mathbf C$ factors, while skip and output-gate streams remain in MPC until the contracted scan output returns.
This split keeps wide linear maps and recurrent processing in CKKS, places nonlinear factor construction in MPC, and makes each backend crossing proportional to compact packets rather than expanded recurrent state.
Each crossing uses $\Pi_{\mathrm{C2M}}^{\mathbb{C}}$ and $\Pi_{\mathrm{M2C}}^{\mathbb{C}}$ from \S\ref{sec:complex-conversion}, and the full protocol details appear in Appendices~\ref{app:resident-hescan-protocol} and~\ref{app:fesc-block-protocol}.

% !TEX root = ../main.tex
\section{Encrypted Kernel Design}
\label{sec:ks-he}

This section first defines the CKKS primitives and complex projection and convolution kernels, then develops HEScan as the encrypted recurrence primitive that makes the scan contract executable under FHE.

% ============================================================================
\subsection{CKKS Primitives and Slot Maps}
\label{sec:ks-he:packing}

We write $\langle\mathbf{v}\rangle$ for a CKKS ciphertext with $n_{\mathrm{ckks}}$ slots, and $[[\mathbf{v}]]$ for additive MPC shares.
Unbracketed CKKS operands are public plaintext weights, masks, or constants. Table~\ref{tab:he_ops} lists the CKKS primitives used below. The cyclic left shift by $r$ slots is
\[
  \rho(\mathbf{v};r)_j
  =
  v_{(j+r)\bmod n_{\mathrm{ckks}}}.
\]
The symbol $\odot$ denotes elementwise multiplication. On CKKS operands, it is implemented by \textsf{ptmul} with one ciphertext operand and by \textsf{ctmul} with two.

% Table~\ref{tab:notation} disambiguates overloaded symbols such as $N$, $n_{\mathrm{ckks}}$, and $\Delta$.

\begin{table}[t]
\centering
\caption{CKKS primitives used by the encrypted kernels.}
\label{tab:he_ops}
\scriptsize
\setlength{\tabcolsep}{3.5pt}
\renewcommand{\arraystretch}{1.15}
\begin{tabular}{@{}l|ccccc@{}}
\toprule
\rowcolor{tabheader}
 & \textsf{add} & \textsf{ptmul} & \textsf{ctmul} & \textsf{rot} & \textsf{conj} \\
\midrule
in  & $\langle\mathbf{v}\rangle,\, \mathbf{u}\ \text{or}\ \langle\mathbf{u}\rangle$ & $\langle\mathbf{v}\rangle,\, \mathbf{w}$ & $\langle\mathbf{v}\rangle,\, \langle\mathbf{w}\rangle$ & $\langle\mathbf{v}\rangle,\, r$ & $\langle\mathbf{v}\rangle$ \\
out & $\langle\mathbf{v}{+}\mathbf{u}\rangle$ & $\langle\mathbf{v}{\odot}\mathbf{w}\rangle$ & $\langle\mathbf{v}{\odot}\mathbf{w}\rangle$ & $\langle\rho(\mathbf{v};r)\rangle$ & $\langle\overline{\mathbf{v}}\rangle$ \\
shorthand & $+$ & $\odot$ & $\odot$ & $\rho(\cdot;r)$ & $\overline{(\cdot)}$ \\
\bottomrule
\end{tabular}
\end{table}

\system{} uses public broadcast and reduction slot-map families.
The map $\mathrm{br}$ expands a compact factor into the state packing, $\mathrm{br}^{(\kappa)}$ restricts it to chunk $\kappa$, and $\Sigma_d$ sums an axis of length $d$ while aligning chunk-local partial sums to common output slots,
\[
  \mathrm{br}^{(\kappa)}(\mathbf{a}\odot\mathbf{a}')
  =
  \mathrm{br}^{(\kappa)}(\mathbf{a})\odot\mathrm{br}^{(\kappa)}(\mathbf{a}'),
  \quad
  \Sigma_d(\mathbf{u})[e]
  =
  \sum_{j=0}^{d-1}u[e,j].
\]
On ciphertext inputs, $\rho$, $\mathrm{br}$, $\mathrm{br}^{(\kappa)}$, and $\Sigma_d$ denote their homomorphic realizations, with exact packings and reduction accounting in Appendices~\ref{app:hescan-packing} and~\ref{app:complex-reduction-accounting}.

% ============================================================================
\subsection{Plaintext--Ciphertext Projections}
\label{sec:ks-he:ctpt}

\system{} uses plaintext-weight linear projections before and after the selective recurrence, evaluated with BSGS plaintext--ciphertext matrix multiplication~\cite{pang2024bolt,park2025powerformer} under segment-column packing and complex pair compression~\cite{zhu2026encformer}.
The HEInProj kernel computes packed slices of $\langle\mathbf{z}\rangle$, $\langle\mathbf{xBC}\rangle$, and $\langle\mathbf{dt}\rangle$, the HEOutProj kernel applies $\mathbf{W}_{\mathrm{out}}$ after the MPC post-SSM stage, and all weight diagonals are encoded offline.
Appendix~\ref{app:he-kernels:proj} expands this projection kernel into its slot packing, diagonal encoding, and online cost.

% ============================================================================
\subsection{Causal Depthwise Convolution}
\label{sec:ks-he:conv}

The causal depthwise convolution sits between the HEInProj kernel and the nonlinear factor builder.
At the plaintext level, HEInProj produces $\mathbf{q}_k=\mathbf{xBC}_k\in\mathbb{R}^{HP+2Gd_s}$ alongside $\mathbf{z}_k$ and $\mathbf{dt}_k$.
For channel $c$, the HEConv kernel evaluates
\begin{equation}
\begin{aligned}
  \eta_k[c]
  &=
  b^{\mathrm{conv}}_c
  +
  \sum_{r=0}^{d_{\mathrm{conv}}-1}
    w^{\mathrm{conv}}_{c,r}\,q_{k-r}[c],\\
  q_t[c]
  &=0,
  \qquad t<0.
\end{aligned}
\label{eq:conv1d_depthwise}
\end{equation}
With the HEInProj packing, the HEConv kernel realizes Eq.~\ref{eq:conv1d_depthwise} in CKKS as
\[
  \langle\boldsymbol{\eta}\rangle
  =
  \mathbf{b}^{\mathrm{conv}}
  +
  \sum_{r=0}^{d_{\mathrm{conv}}-1}
    \mathbf{M}^{\mathrm{conv}}_{r}
    \odot
    \rho\!\left(\langle\mathbf{q}\rangle;\delta_r\right).
\]
The public plaintext vector $\mathbf{M}^{\mathrm{conv}}_r$ fuses tap weights with the causal boundary mask, and $\delta_r$ is a public rotation offset.
At the CKKS--MPC boundary, \system{} converts $\langle\boldsymbol{\eta}\rangle$, $\langle\mathbf{z}\rangle$, and $\langle\mathbf{dt}\rangle$ in one boundary crossing, after which the nonlinear factor builder forms the compact private factors from the resulting shares.

% ============================================================================
\subsection{Factorized HEScan Kernel}
\label{sec:ks-he:scan}

HEInProj and HEConv produce the encrypted pre-activations from which the MPC factor builder derives one packet per token. For the deployed Mamba-2 block and token $k$, the packet contains
\[
\begin{aligned}
  \mathbf{x}_k
  &\in
  \mathbb{R}^{H\times P},
  &
  \mathbf{a}_k
  &\in
  \mathbb{R}^{H},
  &
  \mathbf{B}_k,\mathbf{C}_k
  &\in
  \mathbb{R}^{G\times d_s}.
\end{aligned}
\]
With the grouped outer product $\otimes_g$ defined in \S\ref{sec:mamba_model}, HEScan realizes the scan contract
\begin{align}
  \mathbf{h}_k
  &=
  \mathrm{br}(\mathbf{a}_k)\odot\mathbf{h}_{k-1}
  +
  \mathbf{x}_k\otimes_g\mathbf{B}_k,
  \qquad
  \mathbf{h}_0=\mathbf{0},
  \label{eq:fesc_recurrence}\\
  \mathbf{m}_k
  &=
  \Sigma_{d_s}\!\left(\mathbf{h}_k\odot\mathrm{br}(\mathbf{C}_k)\right).
  \label{eq:fesc_contract}
\end{align}
Equations~\ref{eq:fesc_recurrence} and~\ref{eq:fesc_contract} define the logical recurrence and output contraction.
The mechanisms below realize this contract without representing $\{\mathbf{h}_k\}_{k=1}^{L}$ as either a hybrid-boundary payload or a sequence-resident ciphertext tensor.
Figure~\ref{fig:hescan-factorized-dataflow} traces one factor packet through the encrypted prefix scan and output contraction on a toy packing, using compact labels without boldface or share brackets.

\begin{figure*}[t]
\centering
\includegraphics[width=\textwidth,keepaspectratio,page=7,trim=69pt 177pt 157pt 137pt,clip]{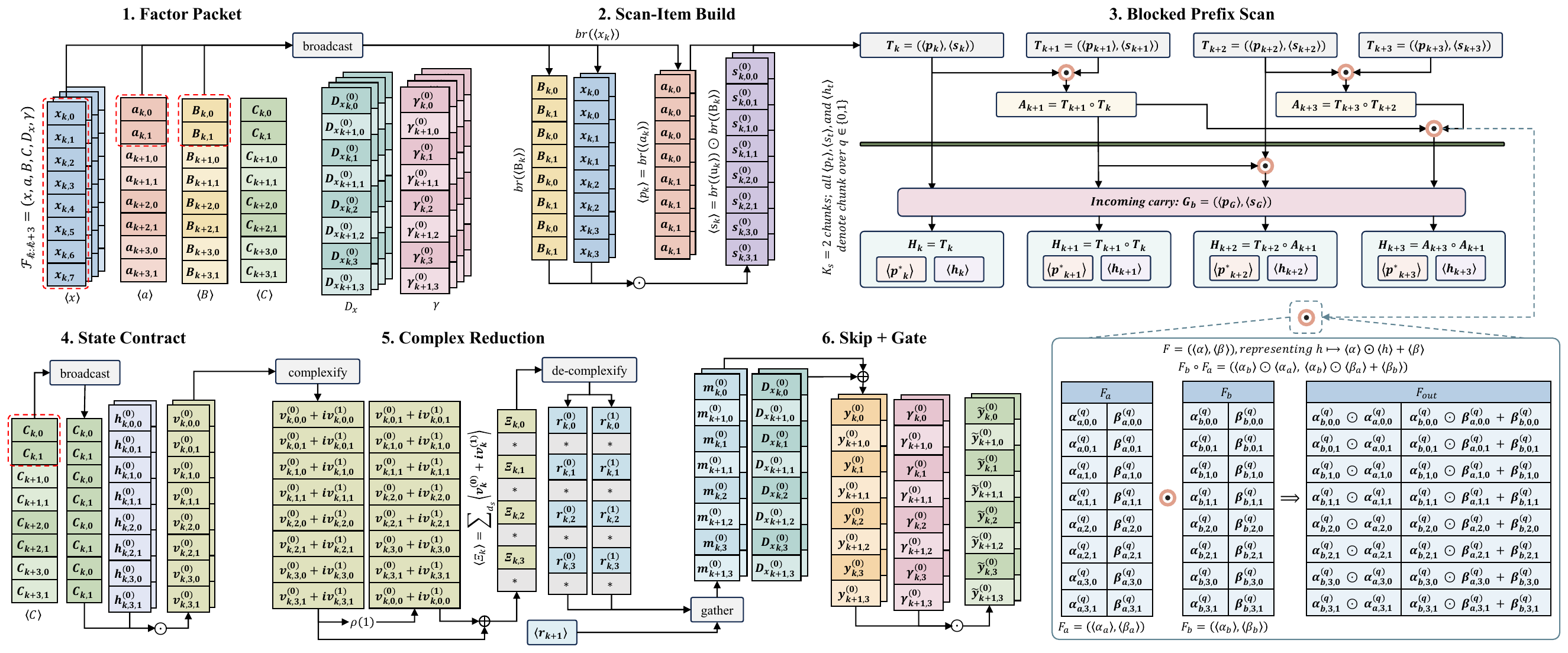}
\caption{Factorized HEScan dataflow on a toy packing with $L{=}4$, $H{=}2$, $P{=}4$, $G{=}1$, $d_s{=}2$, and $K_s{=}2$, so $E{=}HP{=}8$ expanded channels are split into two state chunks. The numbered panels follow one factor packet from broadcast through the blocked prefix scan to HEScan contraction and MPC gating, and the inset gives the affine composition at each scan node.}
\label{fig:hescan-factorized-dataflow}
\end{figure*}

\phantomsection\label{sec:ks-he:factorized-boundary}
\runinhead{Factorized Boundary.} A direct implementation expands the compact decay and input factors into state-shaped tensors,
\[
  \mathrm{br}(\mathbf{a}_k),
  \qquad
  \mathbf{x}_k\otimes_g\mathbf{B}_k,
  \qquad
  \mathbf{h}_k
  \in
  \mathbb{R}^{H\times P\times d_s}.
\]
Transferring these representations across the hybrid boundary over the sequence incurs $\Theta(LHPd_s)$ logical real values.
\system{} instead converts only the factorized packet defined above,
\[
  \bigl([[\mathbf{x}_k]],[[\mathbf{a}_k]],[[\mathbf{B}_k]],[[\mathbf{C}_k]]\bigr),
  \quad
  x_k[h,p]=\Delta_k[h]x_k^{\mathrm{raw}}[h,p].
\]
The skip and gate streams $D_{\mathrm{skip}}\mathbf{x}^{\mathrm{raw}}_k$ and $\mathbf{g}^{\mathrm{out}}_k$, shown as $Dx$ and $\gamma$ in Figure~\ref{fig:hescan-factorized-dataflow}, remain in MPC, while HEScan returns only the contracted value $\mathbf{m}_k\in\mathbb{R}^{H\times P}$.
Consequently, the logical inbound and outbound boundary payloads over $L$ tokens are
\[
  N_{\mathrm{in}}
  =
  L(HP+H+2Gd_s),
  \qquad
  N_{\mathrm{out}}
  =
  LHP,
\]
rather than quantities proportional to $LHPd_s$.
These are logical-real counts, and complex pairing carries two independent real entries in each CKKS slot.
Table~\ref{tab:ablation_boundary} quantifies how this factorized contract removes dense-state conversion from the pipeline.

\phantomsection\label{sec:ks-he:hescan-packing}
\runinhead{State-Chunk Packing.} Let $E=HP$ and choose an active state capacity $s_{\mathrm{state}}\le n_{\mathrm{ckks}}$ divisible by $d_s$, so that each ciphertext stores $c_{\mathrm{state}}=s_{\mathrm{state}}/d_s$ expanded channels and
\[
  K_s
  =
  \left\lceil
    \frac{Ed_s}{s_{\mathrm{state}}}
  \right\rceil
\]
chunks cover the recurrent state.
Indexing an expanded channel by $e=hP+p$ and a state coordinate by $(e,i)$ for $0\le i<d_s$, chunk $\kappa$ owns the channel set $\mathcal{E}_{\kappa}$ and stores $(e,i)$ at slot $\lambda_{\kappa}(e,i)$, where
\begin{equation}
\begin{aligned}
  \mathcal{E}_{\kappa}
  &=
  \left\{
    e\in\mathbb{Z}
    \,\middle|\,
    \substack{
      0\le e<E\\
      \kappa c_{\mathrm{state}}\le e<(\kappa+1)c_{\mathrm{state}}
    }
  \right\},
  \\
  \lambda_{\kappa}(e,i)
  &=
  (e-\kappa c_{\mathrm{state}})d_s+i.
\end{aligned}
\label{eq:hescan-chunk-map}
\end{equation}
Appendix~\ref{app:hescan-packing} gives the exact source indices with the padded broadcast and reduction maps.

\phantomsection\label{sec:ks-he:affine-prefix}
\runinhead{Affine Prefix Composition.} For token $k$ and chunk $\kappa$, the induced affine update is represented as $\mathcal{T}_k^{(\kappa)}=(\langle\mathbf{a}_k\rangle,\allowbreak \langle\mathbf{s}_k^{(\kappa)}\rangle)$.
Its ideal plaintext semantics are the chunk-local affine map
\[
  F_k^{(\kappa)}(\mathbf{z})
  =
  \mathrm{br}^{(\kappa)}(\mathbf{a}_k)
  \odot
  \mathbf{z}
  +
  \mathbf{s}_k^{(\kappa)}.
\]
When map $R$ follows map $L$, their composition is
\begin{align}
  \mathcal{T}_R^{(\kappa)}
  \circ
  \mathcal{T}_L^{(\kappa)}
  =
  \Bigl(
  &
    \langle\mathbf{a}_R\rangle
    \odot
    \langle\mathbf{a}_L\rangle,
  \notag\\
  &
    \mathrm{br}^{(\kappa)}
    \!\left(\langle\mathbf{a}_R\rangle\right)
    \odot
    \langle\mathbf{s}_L^{(\kappa)}\rangle
    +
    \langle\mathbf{s}_R^{(\kappa)}\rangle
  \Bigr).
  \label{eq:compact_affine_compose}
\end{align}
Because $\mathrm{br}^{(\kappa)}$ preserves elementwise products, the composed head-wise decay remains in compact packing, while the additive term remains local to chunk $\kappa$.
An encrypted Brent--Kung network evaluates $O(L)$ such composition nodes with $O(\log L)$ composition depth, and applying each resulting prefix map to the zero initial state yields $\langle\mathbf{h}_k^{(\kappa)}\rangle$.

\phantomsection\label{sec:ks-he:construct-contract}
\runinhead{On-Demand Construction and Contraction.} HEScan constructs each additive update only while processing its chunk and contracts each prefix before releasing the chunk workspace,
\begin{align}
  \langle\mathbf{s}_k^{(\kappa)}\rangle
  &=
  \mathrm{br}^{(\kappa)}
  \!\left(\langle\mathbf{x}_k\rangle\right)
  \odot
  \mathrm{br}^{(\kappa)}
  \!\left(\langle\mathbf{B}_k\rangle\right),
  \notag\\
  \langle\mathbf{m}_k\rangle
  &=
  \sum_{\kappa=0}^{K_s-1}
  \Sigma_{d_s}
  \!\left(
    \langle\mathbf{h}_k^{(\kappa)}\rangle
    \odot
    \mathrm{br}^{(\kappa)}
    \!\left(\langle\mathbf{C}_k\rangle\right)
  \right).
  \label{eq:hescan-contract-output}
\end{align}
Because the sets $\{\mathcal{E}_{\kappa}\}_{\kappa=0}^{K_s-1}$ partition the expanded channels, and hence the recurrent coordinates, while every $\Sigma_{d_s}$ reduction targets the same output slots, summing the chunk contributions exactly recovers the full contraction in Eq.~\ref{eq:fesc_contract}.
The output therefore remains in its contracted packing throughout the chunk sweep.
For the state-axis reduction, \system{} uses complex reduction, placing two chunks in the real and imaginary channels of one rotate--add tree before one conjugation separates their partial sums.
Relative to reducing the two chunks independently, this pairing halves the state-axis rotation count.
Appendix~\ref{app:complex-reduction-accounting} gives the unpacking identities and exact rotation and conjugation counts.

% ============================================================================
\subsection{Resident Block Schedule}
\label{sec:ks-he:resident-schedule}

For $L>B_{\mathrm{blk}}$, HEScan splits the sequence into $K_{\mathrm{blk}}=\lceil L/B_{\mathrm{blk}}\rceil$ blocks and evaluates each state chunk with a two-pass resident schedule, so blocks partition tokens while chunks partition recurrent coordinates.
With $\mathcal{T}_{j,t}^{(\kappa)}$ denoting the compact affine map for global token $k=jB_{\mathrm{blk}}+t$ in block $j$, the schedule uses
\[
\begin{aligned}
  B_j
  &=
  \min(B_{\mathrm{blk}},L-jB_{\mathrm{blk}}),
  \\
  \mathcal{I}^{(\kappa)}
  &=
  \left(\langle\mathbf{1}_H\rangle,\langle\mathbf{0}^{(\kappa)}\rangle\right),
  \qquad
  \mathcal{G}_{0}^{(\kappa)}
  =
  \mathcal{I}^{(\kappa)},
\end{aligned}
\]
where $B_j$ counts the valid positions in block $j$ and the identity map $\mathcal{I}^{(\kappa)}$ pads the final partial block to the prefix-network capacity.
The local prefixes, block summaries, incoming carries, and corrected prefixes are
\begin{equation}
\begin{alignedat}{2}
  \mathcal{P}_{j,t}^{(\kappa)}
  &=
  \mathcal{T}_{j,t}^{(\kappa)}\circ\cdots\circ\mathcal{T}_{j,0}^{(\kappa)},
  \qquad&
  \mathcal{S}_{j}^{(\kappa)}
  &=
  \mathcal{P}_{j,B_j-1}^{(\kappa)},
  \\
  \mathcal{G}_{j}^{(\kappa)}
  &=
  \mathcal{S}_{j-1}^{(\kappa)}\circ\cdots\circ\mathcal{S}_{0}^{(\kappa)},
  \quad j\ge1,
  \qquad&
  \widehat{\mathcal{P}}_{j,t}^{(\kappa)}
  &=
  \mathcal{P}_{j,t}^{(\kappa)}\circ\mathcal{G}_{j}^{(\kappa)}.
\end{alignedat}
\label{eq:resident-block-schedule}
\end{equation}
The first pass computes block summaries, an exclusive scan over them gives incoming carries, and the second pass replays one block at a time to rebuild the corrected prefixes.
After carry injection, each rebuilt prefix is applied and contracted into the output, allowing the chunk workspace to be reused before the next state chunk.
The replay therefore increases total work only by a constant factor, as illustrated in Figure~\ref{fig:hescan-blocked-carry}.

\begin{figure}[t]
\centering
\includegraphics[width=0.85\linewidth,keepaspectratio,page=8,trim=508pt 243pt 556pt 264pt,clip]{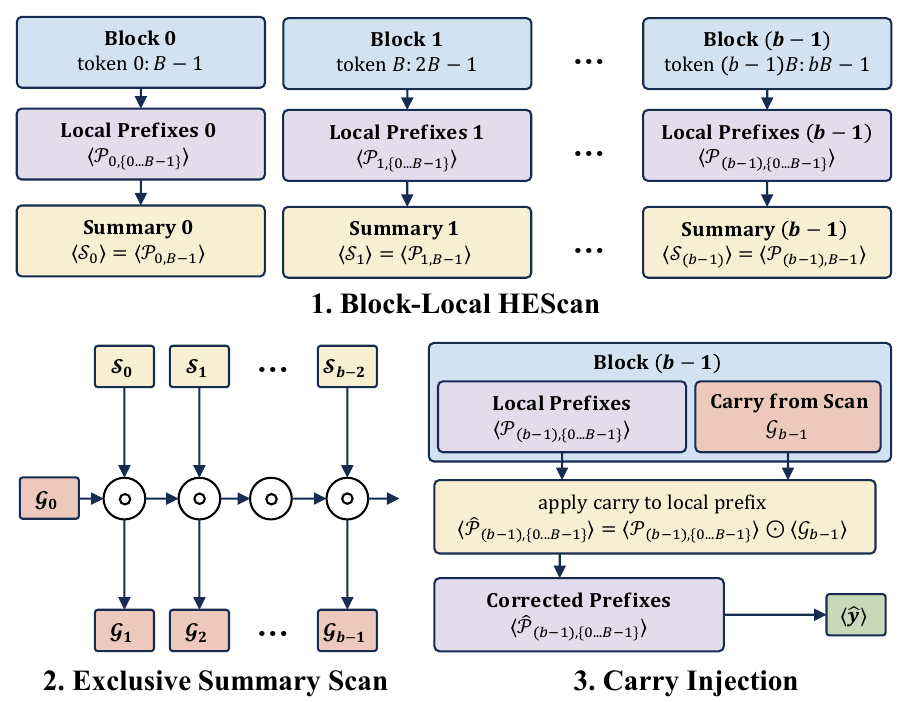}
\caption{Resident HEScan kernel schedule for one state chunk. Only block summaries persist between passes, and corrected prefixes are rebuilt on demand through carry injection.}
\label{fig:hescan-blocked-carry}
\end{figure}

% ============================================================================
\subsection{Correctness and Complexity}
\label{sec:ks-he:correctness-cost}

\begin{theorem}[Correctness of the HEScan Kernel]
\label{thm:hescan-correctness}
For private factor packets and broadcast maps satisfying the product-preservation property above, HEScan with $\mathbf{h}_0=\mathbf{0}$ has ideal plaintext semantics equal to the packet-induced recurrence and contraction in Eqs.~\ref{eq:fesc_recurrence} and~\ref{eq:fesc_contract}.
The decoded deployed outputs differ from these ideal values only because of CKKS arithmetic error and CKKS--MPC conversion quantization.
\end{theorem}

Appendix~\ref{app:hescan-correctness} proves the theorem from Eqs.~\ref{eq:compact_affine_compose},~\ref{eq:resident-block-schedule}, and~\ref{eq:hescan-contract-output}, while Appendix~\ref{app:hescan-cost} gives the primitive-count ledger.
For fixed model dimensions and public slot maps, HEScan satisfies
\[
\begin{aligned}
  N_{\mathsf{ctmul}},N_{\mathsf{ks}}
  &=
  O(LK_s),
  \\
  D_{\mathsf{mult}}
  &=
  O\!\left(
    \log B_{\mathrm{blk}}
    +
    \log K_{\mathrm{blk}}
  \right)
  =
  O(\log L),
  \\
  N_{\mathsf{live}}
  &=
  O\!\left(
    B_{\mathrm{blk}}
    +
    K_{\mathrm{blk}}
    +
    K_y
    +
    K_{\mathrm{fac}}
  \right).
\end{aligned}
\]
The live set comprises one block of prefixes, one summary per block, the contracted-output inventory $K_y=\lceil LHP/s_y\rceil$ for active output capacity $s_y$, and the resident compact-factor inventory $K_{\mathrm{fac}}$.
Because state chunks are processed serially, $K_s$ scales total work but does not enter the peak-residency bound.
Appendix~\ref{app:longseq} audits the level budget.

% !TEX root = ../main.tex
\section{MPC Nonlinear Protocols}
\label{sec:mpc}

This section presents the MPC nonlinear protocols used by the
\system{} block, where MoSiLU evaluates the convolved input and output-gate
branches, MoSoft produces positive timesteps, MMExp derives selective decays,
and NeRMS computes inverse-RMS factors.
Across SiLU, softplus, exponential, and RMSNorm, these protocols use fixed
data-independent schedules and replace generic elementary-function circuits
with structure-specific polynomial, comparison, mux, and Newton components.

% ============================================================================
\subsection{MPC Primitives}
\label{sec:mpc-primitives}

All MPC real values are additively shared over $\mathbb{Z}_{2^\ell}$, with
$r$ encoded as $\lfloor r2^F\rceil$ and $[[\mathbf{x}]]$ denoting shares at
scale $2^F$.
For public $c$, $c\odot_{\mathsf{pub}}[[x]]$ is local.
Additions, public-scalar multiplications, and reductions over shares also require no
online communication.
Table~\ref{tab:mpc_ops} therefore lists only the communicating primitives.
Protocol schedules count one communicating layer per primitive invocation.
Comparisons return unscaled shared bits $[[\mathbf{b}]]_{\mathsf{bit}}$ under
the centered signed interpretation of ring elements.
Correctness assumes that no centered intermediate wraps, while security follows
by sequential composition of $\Pi_\times$, $\Pi_{\mathrm{cmp}}$, and
$\Pi_{\mathrm{mux}}$ in the semi-honest model.
Appendix~\ref{app:mpc-nonlinear-details} gives the preprocessing, range audit,
deployed counts, and approximation coefficients.

\begin{table}[t]
\centering
\caption{Online communicating MPC primitives over $\mathbb{Z}_{2^\ell}$.}
\label{tab:mpc_ops}
\scriptsize
\setlength{\tabcolsep}{4pt}
\renewcommand{\arraystretch}{1.2}
\begin{tabular}{@{}l|l|cc@{}}
\toprule
\rowcolor{tabheader}
Primitive & Functionality & Comm.\ layers & Online comm. \\
\midrule
$\Pi_{\times}$~\cite{beaver1992efficient} &
$([[a]],[[b]])\!\to\![[\lfloor ab/2^F\rceil]]$ &
$1$ & $2\ell$ \\

$\Pi_{\mathrm{cmp}}$~\cite{rathee2021sirnn} &
$([[x]],\tau)\!\to\![[\mathbf{1}[x<\tau]]]_{\mathsf{bit}}$ &
$1$ & $\lambda\ell$ \\

$\Pi_{\mathrm{mux}}$~\cite{rathee2021sirnn} &
$([[u]],[[b]]_{\mathsf{bit}})\!\to\![[bu]]$ &
$1$ & $2\ell$ \\
\bottomrule
\end{tabular}
\end{table}

% ============================================================================
\subsection{Even-Residual SiLU and Softplus}
\label{sec:mpc-mosilu-mosoft}
\label{sec:mpc-nonlinear}

MoSiLU and MoSoft exploit the same symmetry in SiLU and softplus.
Both
$f_{\mathsf{silu}}(x)=x\sigma(x)$ and
$f_{\mathsf{soft}}(x)=\log(1+e^x)$ satisfy
\[
  f(-x)=f(x)-x.
\]
Hence the residual is even and factors through $x^2$,
\[
  g(x)
  =
  f(x)-\frac{x}{2}
  =
  h(x^2).
\]
For an even degree $d=2r$, \system{} fits $h$ on $[0,\tau^2]$ by minimax
approximation and evaluates
\[
  \widehat f(x)
  =
  \frac{x}{2}
  +
  \sum_{i=0}^{r} a_i u^i,
  \qquad
  u=x^2.
\]
The two functions therefore use the same secure graph and differ only in their
public coefficients.
Two parallel comparisons select the left-tail approximation $0$, the polynomial
inside $[-\tau,\tau]$, or the right-tail approximation $x$.
The shared even-residual form gives MoSiLU and MoSoft one fixed MPC execution
path, reducing each nonlinearity to $r$ secure products and constant
comparison and mux overhead instead of generic elementary-function chains.

\begin{algorithm}[t]
\caption{Even-residual protocol for MoSiLU and MoSoft}
\label{alg:mosilu-mosoft}
\label{alg:silu}
\begin{algorithmic}[1]
\Require Shares $[[x]]$, threshold $\tau$, coefficients $a_0,\ldots,a_r$, $r\ge1$
\Ensure Shares $[[\widehat f(x)]]$ approximating $f(x)$,
        for $f\in\{\operatorname{SiLU},\operatorname{softplus}\}$
\State $[[u]]\gets\Pi_\times([[x]],[[x]])$
\If{$r=1$}
  \State $[[p]]\gets
         \frac12\odot_{\mathsf{pub}}[[x]]
         +a_0
         +a_1\odot_{\mathsf{pub}}[[u]]$
\Else
  \State $[[q]]\gets a_r\odot_{\mathsf{pub}}[[u]]+a_{r-1}$
  \For{$i=r{-}2,\ldots,1$}
    \State $[[q]]\gets\Pi_\times([[q]],[[u]])+a_i$
  \EndFor
  \State $[[p]]\gets
         \frac12\odot_{\mathsf{pub}}[[x]]
         +a_0
         +\Pi_\times([[q]],[[u]])$
\EndIf
\State $[[b_-]]_{\mathsf{bit}}\gets \Pi_{\mathrm{cmp}}([[x]],-\tau)$
\State $[[b_{<}]]_{\mathsf{bit}}\gets \Pi_{\mathrm{cmp}}([[x]],\tau)$
\State $[[b_+]]_{\mathsf{bit}}\gets 1-[[b_{<}]]_{\mathsf{bit}}$
\State $[[b_{\mathrm{in}}]]_{\mathsf{bit}}
       \gets [[b_{<}]]_{\mathsf{bit}}-[[b_-]]_{\mathsf{bit}}$
\State $[[y]]\gets \Pi_{\mathrm{mux}}([[p]],[[b_{\mathrm{in}}]]_{\mathsf{bit}})
       +\Pi_{\mathrm{mux}}([[x]],[[b_+]]_{\mathsf{bit}})$
\State \Return $[[y]]$
\end{algorithmic}
\end{algorithm}

% ============================================================================
\subsection{Split Inverse-RMS Normalization}
\label{sec:mpc-nerms}

NeRMS splits RMSNorm at the elementwise square, forming
$\langle u_i\rangle=\langle x_i\rangle\odot\langle x_i\rangle$ in the
surrounding CKKS path and reducing the resulting shares in MPC to
\[
  [[v]]
  =
  \frac{1}{D}
  \sum_{i=1}^{D}[[u_i]]
  +\epsilon.
\]
After this reduction, the remaining nonlinear step is the inverse root, where
parallel comparisons against $K-1$ public bucket boundaries select an
initializer together with its precomputed square from public tables.
Both table values feed the first Newton update, avoiding an additional secure
square,
\[
  y
  \gets
  y(1.5-0.5vy^2).
\]
Algorithm~\ref{alg:nerms} gives the resulting protocol.

\begin{algorithm}[t]
\caption{Secure protocol for NeRMS}
\label{alg:nerms}
\small
\begin{algorithmic}[1]
\Require Shares $[[u_i]]$ of CKKS-formed squares $u_i=x_i^2$, $D$, $\epsilon$,
         ordered boundaries $b_1<\cdots<b_{K-1}$,
         initializers $\{y_j^{(0)}\}_{j=0}^{K-1}$,
         steps $t_{\mathrm{NR}}$
\Ensure Shares $[[\widehat s]]$ approximating
        $1/\sqrt{D^{-1}\sum_i u_i+\epsilon}$
\State $[[v]]\gets
       \frac1D\odot_{\mathsf{pub}}\sum_{i=1}^{D}[[u_i]]+\epsilon$
\State $[[y]]\gets y_0^{(0)}$
\State $[[y^2]]\gets (y_0^{(0)})^2$
\For{$j=1,\ldots,K-1$ \textbf{in parallel}}
  \State $[[e_j]]_{\mathsf{bit}}
         \gets 1-\Pi_{\mathrm{cmp}}([[v]],b_j)$
\EndFor
\For{$j=1,\ldots,K-1$}
  \State $[[y]]\gets [[y]]
         +\bigl(y_j^{(0)}-y_{j-1}^{(0)}\bigr)
          \odot_{\mathsf{pub}}[[e_j]]_{\mathsf{bit}}$
  \State $[[y^2]]\gets [[y^2]]
         +\bigl((y_j^{(0)})^2-(y_{j-1}^{(0)})^2\bigr)
          \odot_{\mathsf{pub}}[[e_j]]_{\mathsf{bit}}$
\EndFor
\For{$t=1,\ldots,t_{\mathrm{NR}}$}
  \State $[[y]]\gets
         \Pi_\times\!\left(
           [[y]],
           \frac32-\frac12\odot_{\mathsf{pub}}
           \Pi_\times([[v]],[[y^2]])
         \right)$
  \If{$t<t_{\mathrm{NR}}$}
    \State $[[y^2]] \gets \Pi_\times([[y]],[[y]])$
  \EndIf
\EndFor
\State \Return $[[y]]$
\end{algorithmic}
\end{algorithm}

The bucket indicators stay secret-shared, while table selection uses local
linear arithmetic.
With $K$ buckets and $t_{\mathrm{NR}}$ Newton steps, NeRMS requires $K-1$
parallel comparisons and $3t_{\mathrm{NR}}-1$ secure products, with interaction
depth independent of $D$.
The surrounding path applies the inverse-RMS factor and RMSNorm weights
after the subprotocol returns.

% ============================================================================
\subsection{Private Selective-Decay Exponential}
\label{sec:mpc-mmexp}

MMExp exploits the one-sided selective-decay domain induced by
$A_h=-\exp(A_{\log,h})<0$ and $\Delta_h\ge0$,
\[
  z_h=\Delta_h A_h\le0,
  \qquad
  a_h=\exp(z_h).
\]
Here $a_h$ is the decay factor consumed by the HEScan packet, staying near one
at small $|z_h|$ and vanishing in the negative tail, so the protocol only needs
the negative half-line rather than a general signed-input circuit.
MMExp approximates $\exp(z)$ over $[-\tau,0]$ using the public-coefficient
minimax polynomial
\[
  p_d(z)=\sum_{i=0}^{d}c_i z^i,
\]
and muxes the negligible tail $z<-\tau$ to zero with one secret comparison,
avoiding repeated-squaring and range-reduction machinery.
All coefficients, thresholds, and the per-head schedule are public, while the
timestep, decay argument, tail bit, and returned factor stay secret-shared, so
the executed schedule is data-independent.
Algorithm~\ref{alg:mmexp} gives the protocol, and its deployed D4 instantiation
is the five-product, one-comparison, one-mux row in
Appendix~\ref{app:mpc-nonlinear-details}.

\begin{algorithm}[t]
\caption{Secure protocol for MMExp}
\label{alg:mmexp}
\begin{algorithmic}[1]
\Require Shares $[[\boldsymbol{\Delta}]]$, shares $[[A_h]]$, threshold $\tau$, coefficients $c_0,\ldots,c_d$
\Ensure Shares $[[\widehat a_h]]$ approximating $\exp(\Delta_hA_h)$
        for each head $h$
\ForAll{heads $h$ \textbf{in parallel}}
  \State $[[z_h]] \gets \Pi_\times([[\Delta_h]],[[A_h]])$
  \State $[[p_h]] \gets c_d\odot_{\mathsf{pub}}[[z_h]]+c_{d-1}$
  \For{$i=d{-}2,\ldots,0$}
    \State $[[p_h]] \gets \Pi_\times([[p_h]],[[z_h]])+c_i$
  \EndFor
  \State $[[b_h]]_{\mathsf{bit}} \gets \Pi_{\mathrm{cmp}}([[z_h]],-\tau)$
  \State $[[a_h]] \gets \Pi_{\mathrm{mux}}([[p_h]],1-[[b_h]]_{\mathsf{bit}})$
\EndFor
\State \Return $[[\mathbf{a}]]$
\end{algorithmic}
\end{algorithm}

% !TEX root = ../main.tex
\section{CKKS--MPC Conversion}
\label{sec:complex-conversion}

At each CKKS--MPC boundary, \system{} converts between CKKS ciphertexts and MPC additive shares over $\mathbb{Z}_{2^\ell}$.
\system{} instantiates the complex-lane CKKS--MPC conversion primitives of EncFormer~\cite{zhu2026encformer}, whose masking and share-modulus conversion follow BLB and SiRNN~\cite{xu2025blb,rathee2021sirnn}.
We use these inherited primitives at the factorized selective-SSM boundary, where the \system{} contribution is the compact boundary packet and its placement in the pipeline.
We write $\Pi_{\mathrm{C2M}}^{\mathbb{C}}$ for complex CKKS-to-MPC conversion, which maps a CKKS ciphertext to additive shares, and $\Pi_{\mathrm{M2C}}^{\mathbb{C}}$ for the reverse MPC-to-CKKS conversion, which maps additive shares back to a CKKS ciphertext.
Following EncFormer~\cite{zhu2026encformer}, \system{} applies modulus trimming to reduce ciphertext payload before conversion by dropping unused RNS primes.
Appendix~\ref{app:conversion} records the exact interfaces, parameterization, and security argument used by the protocol.

% !TEX root = ../main.tex
\section{Experimental Setup}
\label{sec:experiments}

\noindent Evaluation follows the \mambabase{} pipeline under the same semi-honest, single-query document classification setting used by private-inference baselines.
The main latency measurements use NVIDIA A100 GPUs, with one-GPU and multi-GPU results in Table~\ref{tab:e2e_comparison}, while Appendix~\ref{app:gpu-memory-envelope} extends the feasibility checks to additional GPU memory classes.

\subsection{System Setup}
\label{sec:exp-setup}

\runinhead{FHE Settings.} The encrypted kernels are implemented with CKKS and executed on PhantomFHE's GPU backend~\cite{yang2024phantom}.
Table~\ref{tab:ckks_config} summarizes the parameter set for each stage, including the polynomial modulus degree, usable slots, multiplicative depth, modulus chain, and scale.
The projection stages and the NeRMS--HEOutProj path use the $N{=}32768$ setting, while HEScan moves to $N{=}65536$ for $L{\ge}512$ to support the longer scan with the matching slot and level budget.
For \mambabase{}, scan packing keeps $s_{\mathrm{state}}{=}16{,}384$ active state slots per ciphertext, which yields $K_s=12$ state chunks.
Appendix~\ref{app:bootstrap-boundary} gives the level-budget audit.

\begin{table}[t]
\centering
\caption{CKKS parameters for each encrypted stage.}
\label{tab:ckks_config}
\scriptsize
\setlength{\tabcolsep}{3pt}
\renewcommand{\arraystretch}{1.12}
\resizebox{\columnwidth}{!}{%
\begin{tabular}{@{}l|ccccc@{}}
\toprule
\rowcolor{tabheader}
Stage & $N$ & $n_{\mathrm{ckks}}$ & Depth & Chain bits & CKKS scale \\
\midrule
HEInProj kernel  & 32768 & 16384 & 2 & $[60,40,40,60]$ & $2^{30}$ \\
HEScan kernel, $L{\le}256$ & 32768 & 16384 & $D_{\mathrm{scan}}(L)$\textsuperscript{$\dagger$} & $[60,40^{\times D_{\mathrm{scan}}(L)},60]$ & $2^{40}$ \\
HEScan kernel, $L{\ge}512$ & 65536 & 32768 & $D_{\mathrm{scan}}(L)$\textsuperscript{$\dagger$} & $[60,40^{\times D_{\mathrm{scan}}(L)},60]$ & $2^{40}$ \\
NeRMS + HEOutProj kernel & 32768 & 16384 & 3 & $[60,40,40,40,60]$ & $2^{30}$ \\
\bottomrule
\end{tabular}%
}
\vspace{2pt}
\parbox{\linewidth}{\scriptsize $^{\dagger}$Adapts to $L$, for example $22$ at $L{=}2048$, with values in Appendix~\ref{app:longseq}.}
\end{table}

\runinhead{MPC Settings.} The pipeline evaluates all nonlinear stages with a SCI/EzPC-style two-party arithmetic-sharing backend~\cite{chandran2019ezpc,rathee2020cryptflow2}.
MPC values are signed fixed-point shares over $\mathbb{Z}_{2^\ell}$ with fractional precision $F{=}19$ and arithmetic bit width $\ell{=}44$ in the full encrypted pipeline.
Secure products, comparisons, and muxes use the interactive backend, while additions and public-coefficient products are local share operations.
Standalone protocol sweeps keep the deployed protocol graphs and polynomial degrees, and report their fixed-point setting when it differs.

\runinhead{Network Profiles.} Following encrypted-inference evaluations in prior works~\cite{lu2025bumblebee,xu2025blb,zhu2026encformer}, latency combines online computation with aggregate online traffic under LAN at 1\,Gbps and 0.3\,ms, WAN1 at 400\,Mbps and 4\,ms, WAN2 at 100\,Mbps and 4\,ms, and WAN3 at 100\,Mbps and 80\,ms.

\runinhead{OOM Threshold.} For feasibility reporting, OOM denotes out-of-memory execution, meaning a run that fails due to memory exhaustion or exceeds $40$\,GB of live GPU allocation or $256$\,GB of peak host allocation per inference, with threshold provenance in Appendix~\ref{app:oom-threshold}.

\subsection{Model and Workload Setup}
\label{sec:workload-setup}

\runinhead{Deployed Model and Reference.} The main implementation, \mambabase{}, is a $12$-layer student constructed through task-specific distillation from Mamba-2-130M~\cite{huggingface} to match the depth and width of BERT-base-uncased~\cite{devlin2019bert} used in prior private-inference evaluations.
For each dataset and target length, we adapt the public $24$-layer Mamba-2 checkpoint for one epoch and distill the $12$-layer student for one epoch, while BERT-base is fine-tuned directly for two epochs.
Both recipes use the same dataset splits, target sequence length, document-prefix policy, seed set, and validation-loss checkpoint selection, so the comparison controls task accuracy under fixed training recipes.
\mambabase{} keeps model width $d_{\mathrm{model}}{=}768$, expanded width $d_{\mathrm{inner}}{=}E{=}1536$, and state dimension $d_s{=}128$ in the encrypted scan interface.
\mambabase{} is the end-to-end model used for accuracy, latency, communication, residency, and multi-GPU scaling results, while Appendix~\ref{app:cross-ssm-family} validates the same scan-contract implementation across different model shapes and representative invariant and selective SSM implementations, including Mamba-3.

\runinhead{Workloads and Datasets.} Following the single-query inference setting used by prior private inference systems~\cite{pang2024bolt,lu2025bumblebee,xu2025blb,zhu2026encformer}, \system{} targets \emph{single-pass} document classification and scoring.
Table~\ref{tab:datasets} summarizes the three long-document benchmarks, each evaluated at $L \in \{128,256,512,1024,2048,4096\}$ tokens to test long-context encrypted inference beyond prior systems' windows.
Fine-tuning and evaluation hyperparameters are summarized in Appendix~\ref{app:hyperparams}.

\begin{table}[t]
\centering
\caption{Long-document classification datasets.}
\label{tab:datasets}
\scriptsize
\setlength{\tabcolsep}{4pt}
\renewcommand{\arraystretch}{1.12}
\begin{tabular}{@{}l|llcc@{}}
\toprule
\rowcolor{tabheader}
Dataset & Domain & Task & Classes & \# Docs \\
\midrule
SCOTUS~\cite{chalkidis2022lexglue}   & Legal      & Issue area classif. & 13 & $1{,}400$ \\
arXiv~\cite{clement2019arxiv}        & Scientific & Topic classif.      & 11 & $2{,}500$ \\
Patent~\cite{sharma2019bigpatent}    & Technical  & CPC classif.        & 9  & $5{,}000$ \\
\bottomrule
\end{tabular}%
\end{table}

\runinhead{Approximation Methodology.} We evaluate two approximation regimes.
\emph{Post-swap} substitutes polynomial approximations only at inference on a matched exact \mambabase{} checkpoint.
\emph{Approximation-aware fine-tuning} (AAT) fine-tunes \mambabase{} with the polynomial approximations installed.
The exact \mambabase{} control and AAT initialize from the same distilled student and each receive one additional task epoch, isolating the effect of the secure-friendly approximations.
MoSiLU and MoSoft sweep degree $d_{\mathrm{poly}} \in \{2,4,6\}$ at threshold $\tau{=}4$, NeRMS sweeps $t_{\mathrm{NR}}\in\{1,2,3\}$ Newton iterations at $K{=}8$ buckets, and MMExp sweeps degree $d_{\mathrm{poly}} \in \{2,4,6,8\}$.
Because a full joint sweep spans thousands of AAT runs, we use a progressive approximation search that replaces nonlinear operators stage by stage with fine-tuning, following prior FHE methodology~\cite{tong2024smartpaf} and the HE-friendly practice of training around fixed approximations~\cite{park2025powerformer}.

\subsection{Baselines and Measurement Methodology}
\label{sec:exp-baselines}

Direct end-to-end latency compares measured \system{} runs with executable or published semi-honest two-party private-inference rows under common network profiles.
The baseline set follows the FHE-only, MPC-only, then hybrid order, covering AEGIS~\cite{gong2026aegis}, MPCFormer~\cite{li2023mpcformer}, SIGMA~\cite{gupta2024sigma}, SHAFT~\cite{kei2025shaft}, BOLT~\cite{pang2024bolt}, BumbleBee~\cite{lu2025bumblebee}, BLB~\cite{xu2025blb}, and EncFormer~\cite{zhu2026encformer}.
Mamba-family ablations use MPCMamba~\cite{yu2025mpcmamba} as the prior encrypted-SSM protocol reference for the MPC nonlinear path.
The direct comparison retains AEGIS because it is the only prior FHE-only system that attempts $L{\ge}1024$, reporting full-sequence encrypted execution through $L{=}2048$ on four A100s.
Other FHE-only systems are positioned in related work because their workloads are less directly aligned, and CipherPrune~\cite{zhang2025cipherprune} is excluded because it evaluates pruned-token inference rather than the full-sequence workload.
Section~\ref{sec:related} positions all systems by backend and sequence primitive, and Appendix~\ref{app:baseline_sources} gives artifact provenance and counting rules. 

% !TEX root = ../main.tex
\section{Performance}
\label{sec:performance}

\noindent We evaluate \system{} across kernel cost, accuracy, approximation accuracy, encrypted correctness, end-to-end latency, memory, communication, GPU scaling, and portability.
We separate reach into accuracy-audited secure inference through $L{=}2{,}048$, native one-GPU latency through $L{=}4{,}096$, and backend reach through $L{=}8{,}192$ with secure refresh.
The results show near-linear scaling, preserved task accuracy, OOM-free execution beyond the audited range, portability across six GPU classes, and kernel compatibility with SSM families including Mamba-3.

\subsection{Microbenchmarks}
\label{sec:gpu-bench}

\runinhead{FHE Computation Cost.} Figure~\ref{fig:ks_scaling} compares one \system{} Mamba-base layer with one BERT-base layer from EncFormer~\cite{zhu2026encformer} and BLB~\cite{xu2025blb}, showing the key-switch advantage of a linear encrypted sequence primitive over quadratic encrypted attention.
Transformer primitives retain token-pair cost, so the crossover occurs by $L{=}512$ and widens at long context, with equations and one-layer KSw totals in Appendices~\ref{app:ks-scaling-model} and~\ref{app:baseline-ks-accounting}.

\begin{figure}[t]
  \centering
  \includegraphics[width=0.92\linewidth]{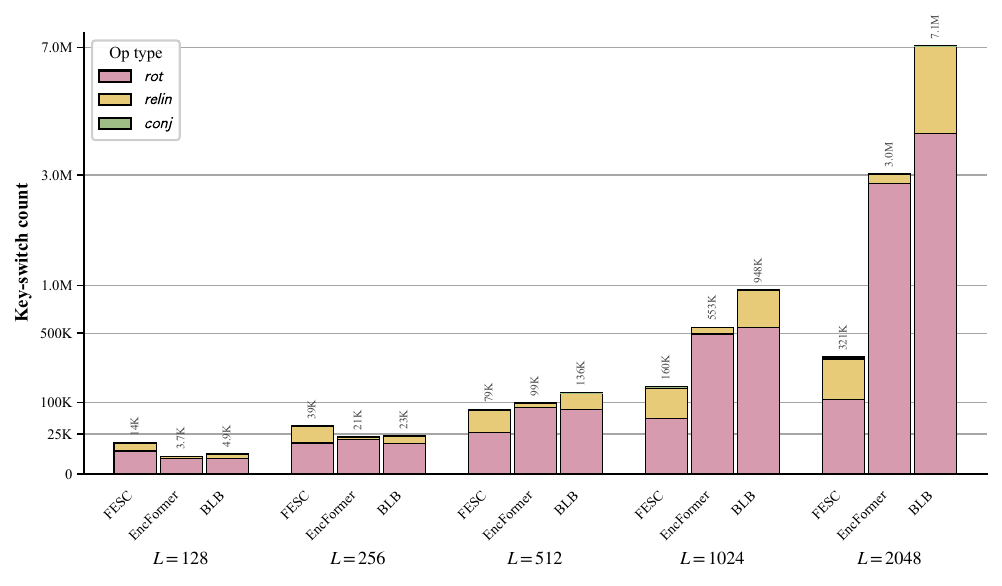}
  \caption{Key-switch cost by operation type for one layer of \system{}, EncFormer~\cite{zhu2026encformer}, and BLB~\cite{xu2025blb}.}
  \label{fig:ks_scaling}
\end{figure}

\runinhead{Scan Topology Selection.} As introduced in \S\ref{sec:prelim_prefix_scan}, SSM recurrences admit several parallel-prefix scan topologies, and Table~\ref{tab:real-phantom-scan} selects the best one for FHE execution.
The comparison uses the Mamba-base scan shape $d_{\mathrm{inner}}{=}1536$, $d_s{=}128$, and active \system{} state packing $s_{\mathrm{state}}{=}16{,}384$ slots, and the times are s/layer under the one-GPU online latency accounting used in Table~\ref{tab:runtime-breakdown}.
Sequential recurrence has the lowest work but depth $L{-}1$, which exceeds the multiplicative budget.
Among parallel-prefix scans, Brent--Kung has the fewest key-switches and lowest latency at every length, so \system{} adopts Brent--Kung inside each encrypted scan block.

\begin{table}[t]
\centering
\caption{Prefix-scan topology comparison for one \system{} scan block.}
\label{tab:real-phantom-scan}
\scriptsize
\setlength{\tabcolsep}{2.4pt}
\renewcommand{\arraystretch}{1.08}
\begin{adjustbox}{max width=\columnwidth}
\begin{tabular}{@{}c|ccc|ccc|ccc|ccc@{}}
\toprule
\rowcolor{tabheader}
\multicolumn{1}{c|}{}
& \multicolumn{3}{c|}{Sequential}
& \multicolumn{3}{c|}{Kogge--Stone}
& \multicolumn{3}{c|}{Sklansky}
& \multicolumn{3}{c}{Brent--Kung} \\
\rowcolor{tabheader}
$L$ & KSw & \makecell{(s/layer)} & dep. & KSw & \makecell{(s/layer)} & dep. & KSw & \makecell{(s/layer)} & dep. & KSw & \makecell{(s/layer)} & dep. \\
\midrule
128  & $3.0\mathrm{K}$  & --- & $127$  & $18.5\mathrm{K}$  & $48.5$      & $7$  & $10.8\mathrm{K}$  & $28.3$  & $7$  & $5.9\mathrm{K}$  & $15.6$  & $13$ \\
256  & $6.1\mathrm{K}$  & --- & $255$  & $43.0\mathrm{K}$  & $115.1$     & $8$  & $24.6\mathrm{K}$  & $65.7$  & $8$  & $12.0\mathrm{K}$ & $32.2$  & $15$ \\
512  & $12.3\mathrm{K}$ & --- & $511$  & $98.3\mathrm{K}$  & $256.6$     & $9$  & $55.3\mathrm{K}$  & $144.3$ & $9$  & $24.3\mathrm{K}$ & $63.5$  & $17$ \\
1024 & $24.6\mathrm{K}$ & --- & $1023$ & $221.2\mathrm{K}$ & $631.1$     & $10$ & $122.9\mathrm{K}$ & $350.6$ & $10$ & $48.9\mathrm{K}$ & $139.4$ & $19$ \\
2048 & $49.1\mathrm{K}$ & --- & $2047$ & $491.5\mathrm{K}$ & $1{,}425.6$ & $11$ & $270.3\mathrm{K}$ & $784.0$ & $11$ & $98.0\mathrm{K}$ & $284.2$ & $21$ \\
\bottomrule
\end{tabular}
\end{adjustbox}
\end{table}

\subsection{Accuracy and Approximation Results}
\label{sec:mpc-approx}

\runinhead{Long-Sequence Model Accuracy.} Figure~\ref{fig:exact_accuracy} compares \mambabase{} with the BERT-base reference used by prior private-inference systems under the controlled setup in \S\ref{sec:workload-setup}, matching splits, target lengths, seeds, validation selection, and two task epochs.
At $L{\ge}512$, \mambabase{} matches or exceeds BERT-base in all nine evaluated cells and averages $2.66$ percentage points higher, showing that the linear selective-scan primitive improves long-sequence task quality while preserving the scaling needed for encrypted execution.
Because each length uses the checkpoint deployed at that operating point, these accuracies are paired with the corresponding encrypted latency and memory measurements.

\newcommand{\accdelta}[2]{$#1{\scriptscriptstyle\, (#2)}$}

\begin{figure}[t]
\centering
\includegraphics[width=\columnwidth]{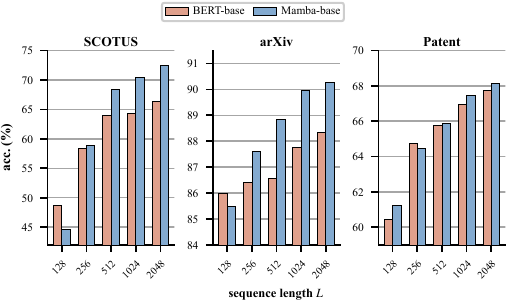}
\caption{Three-seed mean task accuracy across sequence lengths.}
\label{fig:exact_accuracy}
\end{figure}

\phantomsection\label{sec:sweep-illustration}
\runinhead{Nonlinear Protocol Sweep.} Figure~\ref{fig:nonlinear_sweep_pareto} reports the progressive sweep from \S\ref{sec:workload-setup} and selects MoSiLU/MoSoft~D4, NeRMS~N1, and MMExp~D4 as the deployed nonlinear stack. MoSiLU/MoSoft D4 gives the best low-round accuracy--cost tradeoff, NeRMS~N1 avoids extra Newton-step cost without stable accuracy loss, and MMExp~D4 is sufficient because higher degrees do not improve the selected AAT point enough to justify their cost.

\begin{figure}[t]
\centering
\includegraphics[width=\linewidth]{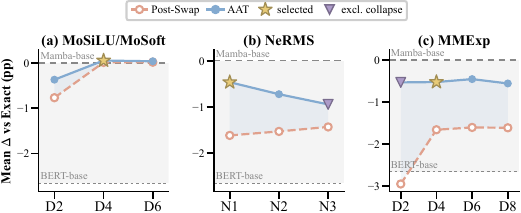}
\caption{Nonlinear protocol selection across three progressive sweep stages, averaged over all datasets and lengths. Per-cell three-seed deviations appear in Appendix~\ref{app:full-results}.}
\label{fig:nonlinear_sweep_pareto}
\end{figure}

AAT gives the highest accuracy by training around the fixed nonlinear protocols offline, and this preparation is outside online latency. Post-swap needs no retraining and remains a competitive fallback, showing that the selected D4/N1/D4 stack is stable under direct protocol substitution. We use the AAT-selected stack for deployed encrypted measurements, while Appendices~\ref{app:mpc-nonlinear-details} and~\ref{app:full-results} report the calibration curves, pointwise errors, and full grids.

\runinhead{MPCMamba Protocol Comparison.} MPCMamba~\cite{yu2025mpcmamba} is the closest MPC-only Mamba system, but uses generic elementary-function protocols for Mamba nonlinearities, incurring high communication overhead.
Figure~\ref{fig:mpc_protocol_cost} compares \system{} with the nonlinear protocols reported in that paper, using repeated-squaring $\exp$ at $n_{\exp}{=}8$ and Newton/Householder refinements at $t{=}3$, including SiLU sign comparison and RMSNorm/softplus exponential initialization.
Across the Mamba nonlinear stack, \system{} uses up to $7.7\times$ fewer online communicating operations than the MPCMamba protocols.

\begin{figure}[t]
\centering
\includegraphics[width=\linewidth]{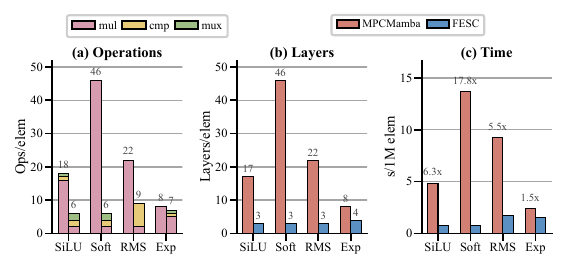}
\caption{MPC nonlinear cost per output element, reporting (a) online operation mix, (b) communicating layers, and (c) measured runtime at $L{=}2048$.}
\label{fig:mpc_protocol_cost}
\end{figure}

\phantomsection\label{sec:longseq}
\runinhead{Approximation Fidelity.} Table~\ref{tab:longseq_accuracy} separates exact-model accuracy from the deployed approximation stack.
Across the nine full-test cells, Plain \system{} differs from exact \mambabase{} by $-0.53$ signed points on average and $0.62$ absolute points on average.
The selected D4/N1/D4 stack therefore preserves the useful accuracy of the exact \mambabase{} control.

\runinhead{Encrypted Accuracy.} A full encrypted sweep of the three datasets would require an estimated $823$ GPU-days per seed on one A100, so we perform an end-to-end encrypted accuracy audit on a stratified $10\%$ subset with $2{,}670$ inferences per seed.
By running Plain and Encrypted \system{} on the same stratified examples, the audit directly isolates the encrypted-execution gap from subset sampling noise.
Across $8{,}010$ audited predictions, Encrypted \system{} stays within $0.48$ percentage points of Plain \system{} in every row and within $0.15$ percentage points in aggregate.
Encrypted \system{} therefore preserves the Plain \system{} accuracy profile across the audited settings, showing that leveled-CKKS evaluation and CKKS--MPC conversion introduce no meaningful task-level loss on the stratified subset.
Appendix~\ref{app:encrypted-audit-scope} reports the audit construction, native-runtime estimate, and residual-stream measurements.

\begin{table}[t]
\centering
\caption{Full-test accuracy and encrypted audit accuracy [\%].}
\label{tab:longseq_accuracy}
\scriptsize
\setlength{\tabcolsep}{2.0pt}
\renewcommand{\arraystretch}{1.12}
\begin{adjustbox}{max width=\columnwidth}
\begin{tabular}{@{}l|c|ccc|cc@{}}
\toprule
\rowcolor{tabheader}
\multicolumn{1}{l|}{} & \multicolumn{1}{c|}{} & \multicolumn{3}{c|}{Full-test} & \multicolumn{2}{c}{Audit} \\
\rowcolor{tabheader}
Dataset & $L$ & BERT-base & Mamba-base & Plain \system{} & Plain \system{} & Encrypted \system{} \\
\midrule
\multirow{3}{*}{SCOTUS}
& 512  & $64.07$ & $68.36$ & $67.71$ & $69.05$ & $68.81$ \\
& 1024 & $64.36$ & $70.50$ & $69.57$ & $70.48$ & $70.00$ \\
& 2048 & $66.36$ & $72.50$ & $70.93$ & $70.00$ & $69.76$ \\
\midrule
\multirow{3}{*}{arXiv}
& 512  & $86.56$ & $88.84$ & $89.08$ & $87.47$ & $87.33$ \\
& 1024 & $87.76$ & $89.96$ & $89.36$ & $88.00$ & $88.00$ \\
& 2048 & $88.36$ & $90.28$ & $89.80$ & $88.40$ & $88.00$ \\
\midrule
\multirow{3}{*}{Patent}
& 512  & $65.78$ & $65.88$ & $66.06$ & $65.20$ & $65.07$ \\
& 1024 & $66.98$ & $67.44$ & $67.00$ & $66.47$ & $66.40$ \\
& 2048 & $67.72$ & $68.12$ & $67.60$ & $66.93$ & $66.87$ \\
\bottomrule
\end{tabular}
\end{adjustbox}
\end{table}

\subsection{End-to-End Latency}
\label{sec:e2e-latency}

Table~\ref{tab:e2e_comparison} compares end-to-end encrypted document-classification latency and long-context coverage across the stated private-inference settings and GPU counts. The rows place the FHE-only reference first, the MPC-only baselines next, and the hybrid baselines together with \system{} last. Several hybrid BERT systems are faster at $L{=}128$, but the long-context columns expose the execution boundary, where Transformer baselines become unavailable or run out of memory while \system{} reports native latency measurements through $L{=}4{,}096$. Baseline provenance and the multi-GPU HEScan algorithm appear in Appendices~\ref{app:baseline_sources} and~\ref{app:multi-gpu-hescan}.

\begin{table*}[t]
\centering
\caption{Encrypted end-to-end latency and GPU-minutes.}
\label{tab:e2e_comparison}
\label{tab:gpu_scaling}
\scriptsize
\setlength{\tabcolsep}{2pt}
\renewcommand{\arraystretch}{1.15}
\begin{tabular}{@{}c|l|lcc|c|cc|cc|cc|cc|cc|cc@{}}
\toprule
\rowcolor{tabheader}
& & & & & & \multicolumn{2}{c|}{$L{=}128$} & \multicolumn{2}{c|}{$L{=}256$} & \multicolumn{2}{c|}{$L{=}512$} & \multicolumn{2}{c|}{$L{=}1{,}024$} & \multicolumn{2}{c|}{$L{=}2{,}048$} & \multicolumn{2}{c}{$L{=}4{,}096$\textsuperscript{\scriptsize \S}} \\
\rowcolor{tabheader}
Type & System & Model & Layers & Params. & GPUs & min & GPU$\cdot$min & min & GPU$\cdot$min & min & GPU$\cdot$min & min & GPU$\cdot$min & min & GPU$\cdot$min & min & GPU$\cdot$min \\
\midrule
\multirow{2}{*}{\rotatebox[origin=c]{90}{\scriptsize\itshape FHE}} & AEGIS~\cite{gong2026aegis} & BERT-base & $12$ & $110$M & $1$ & $12.4^{\ddagger}$ & $12.4^{\ddagger}$ & NR & NR & $34.9^{\ddagger}$ & $34.9^{\ddagger}$ & NR & NR & OOM & OOM & NR & NR \\
 & AEGIS & BERT-base & $12$ & $110$M & $4$ & $3.2$ & $12.8$ & NR & NR & $9.1$ & $36.4$ & NR & NR & $83.9$ & $335.6$ & NR & NR \\
\midrule
\multirow{3}{*}{\rotatebox[origin=c]{90}{\scriptsize\itshape MPC}} & MPCFormer~\cite{li2023mpcformer} & BERT-base & $12$ & $110$M & $1$ & $5.5^{\dagger}$ & $5.5$ & $10.6^{\dagger}$ & $10.6$ & $29.9^{\dagger}$ & $29.9$ & NR & NR & NR & NR & NR & NR \\
 & SIGMA~\cite{gupta2024sigma} & BERT-base & $12$ & $110$M & $1$ & $7.8^{\dagger}$ & $7.8$ & NR & NR & NR & NR & NR & NR & NR & NR & NR & NR \\
 & SHAFT~\cite{kei2025shaft} & BERT-base & $12$ & $110$M & $1$ & $2.9^{\dagger}$ & $2.9$ & NR & NR & NR & NR & NR & NR & NR & NR & NR & NR \\
\midrule
\multirow{8}{*}{\rotatebox[origin=c]{90}{\scriptsize\itshape Hybrid}} & BOLT~\cite{pang2024bolt} & BERT-base & $12$ & $110$M & $1$ & $24.0$ & $24.0$ & NR & NR & NR & NR & NR & NR & NR & NR & NR & NR \\
 & BumbleBee~\cite{lu2025bumblebee} & BERT-base & $12$ & $110$M & $1$ & $5.1$ & $5.1$ & NR & NR & NR & NR & NR & NR & NR & NR & NR & NR \\
 & BLB~\cite{xu2025blb} & BERT-base & $12$ & $110$M & $1$ & $2.6$ & $2.6$ & OOM & OOM & OOM & OOM & NR & NR & NR & NR & NR & NR \\
 & EncFormer~\cite{zhu2026encformer} & BERT-base & $12$ & $110$M & $1$ & $2.2$ & $2.2$ & NR & NR & NR & NR & NR & NR & NR & NR & NR & NR \\
 & \cellcolor{oursrow}\system{}-base & \cellcolor{oursrow}Mamba-base & \cellcolor{oursrow}$12$ & \cellcolor{oursrow}$84$M & \cellcolor{oursrow}$1$ & \cellcolor{oursrow}$4.4$ & \cellcolor{oursrow}$4.4$ & \cellcolor{oursrow}$9.0$ & \cellcolor{oursrow}$9.0$ & \cellcolor{oursrow}$17.8$ & \cellcolor{oursrow}$17.8$ & \cellcolor{oursrow}$38.1$ & \cellcolor{oursrow}$38.1$ & \cellcolor{oursrow}$77.3$ & \cellcolor{oursrow}$77.3$ & \cellcolor{oursrow}$149.8$ & \cellcolor{oursrow}$149.8$ \\
 & \system{}-base & Mamba-base & $12$ & $84$M & $2$ & $2.8$ & $5.6$ & $5.6$ & $11.2$ & $11.7$ & $23.4$ & $24.9$ & $49.8$ & $50.6$ & $101.2$ & $98.2$ & $196.4$ \\
 & \system{}-base & Mamba-base & $12$ & $84$M & $4$ & $2.2$ & $8.8$ & $4.4$ & $17.6$ & $8.6$ & $34.4$ & $17.8$ & $71.2$ & $36.0$ & $144.0$ & $68.5$ & $274.0$ \\
 & \system{}-130M & Mamba-130M & $24$ & $130$M & $1$ & $7.8$ & $7.8$ & $17.7$ & $17.7$ & $34.2$ & $34.2$ & $69.8$ & $69.8$ & $148.6$ & $148.6$ & $291.7$ & $291.7$ \\
\bottomrule
\end{tabular}
\vspace{2pt}
\parbox{\linewidth}{\scriptsize NR means no end-to-end latency is paper-reported or reproducible from a released artifact. OOM means reported or reproduced out of memory. $^{\dagger}$ cells are reproduced and normalized to the paper network profile. $^{\ddagger}$ cells are computed from the paper-reported four-A100 rows and speedup ratio. \textsuperscript{\scriptsize \S} reports $L{=}4{,}096$ latency only.}
\end{table*}

\system{}-base is the row used throughout the systems measurements, and \system{}-130M provides a parameter-size reference near 110M-parameter BERT-base systems.
At $L{=}2{,}048$, \system{} completes in $77.3$ minutes on one GPU, faster than the only reported encrypted Transformer result at that length, AEGIS on \emph{four} GPUs at $83.9$ minutes, while consuming $4.34\times$ fewer GPU-minutes, and four-GPU \system{} reaches $36.0$ minutes. From $L{=}128$ to $L{=}2{,}048$, one-GPU \system{} latency grows by $17.6\times$, close to the $16\times$ sequence growth, against $26.2\times$ for reported four-GPU AEGIS.

\subsection{Communication Cost}
\label{sec:comm-cost}

Online communication in \system{} comes from CKKS--MPC conversion and MPC nonlinear computations, scaling linearly from $4.7$\,GB at $L{=}128$ to $73.8$\,GB at $L{=}2048$ for the $12$-layer pipeline. Figure~\ref{fig:network_cost} shows measured latency across the network profiles from \S\ref{sec:exp-setup}, with full values in Appendix~\ref{app:full-results}. The main sensitivity is bandwidth rather than round-trip latency, as moving from LAN to WAN1 adds $16.0$\,min at $L{=}2048$ and the $100$\,Mbps profile makes communication the largest term, while WAN3's higher RTT adds only $1.6$\,min over WAN2 after online conversion and MPC protocol layers.

\begin{figure}[t]
\centering
\includegraphics[width=1.0\columnwidth]{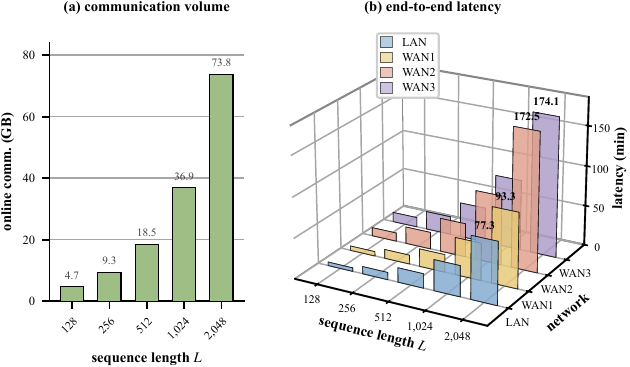}
\caption{\system{} online communication (a) and end-to-end latency across network profiles (b) versus sequence length. Exact values appear in Appendix~\ref{app:full-results}.}
\label{fig:network_cost}
\end{figure}

\runinhead{Boundary-Placement Ablation.} Table~\ref{tab:ablation_boundary} measures the gains from FHE contraction and compact factor transport over a dense boundary.
The dense bridge transfers scan maps and states, FHE contraction removes the outbound state transfer, and the factorized boundary replaces inbound maps with compact factors expanded inside HEScan.
At $L{=}2048$, FHE contraction reduces online communication by $1.30\times$, while factorized inputs add another $12.4\times$ reduction and deliver the full $16.1\times$ gain, sustained across sequence lengths.

\begin{table}[t]
\centering
\caption{Boundary-placement ablation across sequence lengths.}
\label{tab:ablation_boundary}
\scriptsize
\setlength{\tabcolsep}{2.4pt}
\renewcommand{\arraystretch}{1.08}
\begin{adjustbox}{max width=\columnwidth}
\begin{tabular}{@{}l|l|ccccc@{}}
\toprule
\rowcolor{tabheader}
Boundary variant & Metric & $128$ & $256$ & $512$ & $1024$ & $2048$ \\
\midrule
\multirow{3}{*}{\makecell[l]{Dense bridge}} & CTs/layer & $3,123$ & $6,245$ & $12,489$ & $24,978$ & $49,955$ \\
 & Online (GB) & $6.20$ & $12.39$ & $24.79$ & $49.57$ & $99.15$ \\
 & 12-layer LAN (min) & $10.6$ & $21.3$ & $42.6$ & $85.2$ & $170.3$ \\
\midrule
\multirow{3}{*}{\makecell[l]{$+$ FHE contraction}} & CTs/layer & $1,623$ & $3,245$ & $6,489$ & $12,978$ & $25,955$ \\
 & Online (GB) & $4.78$ & $9.56$ & $19.12$ & $38.23$ & $76.46$ \\
 & 12-layer LAN (min) & $8.2$ & $16.4$ & $32.8$ & $65.7$ & $131.4$ \\
\midrule
\multirow{3}{*}{\makecell[l]{$+$ Factorized inputs}} & CTs/layer & $124$ & $246$ & $490$ & $980$ & $1,958$ \\
 & Online (GB) & $0.39$ & $0.77$ & $1.54$ & $3.08$ & $6.16$ \\
 & 12-layer LAN (min) & $0.7$ & $1.3$ & $2.6$ & $5.3$ & $10.6$ \\
\bottomrule
\end{tabular}
\end{adjustbox}
\end{table}

\subsection{Long-Sequence Memory Feasibility}
\label{sec:mem-feasibility}

At $L{=}2048$, the full-length HEScan schedule requires $59.2$\,GB of live state, exceeding the $40$\,GB device capacity.
Figure~\ref{fig:scan_residency} shows the depth--residency tradeoff, where small blocks exhaust multiplicative depth, while larger blocks increase residency until the full-length schedule exceeds device memory.
The deployed $B_{\mathrm{blk}}{=}1024$ serialized-carry schedule has depth $22$, limits live-state peak to $32.7$\,GB, and runs the complete $12$-layer pipeline in $77.3$\,min end-to-end.

\begin{figure}[t]
\centering
\includegraphics[width=0.9\columnwidth]{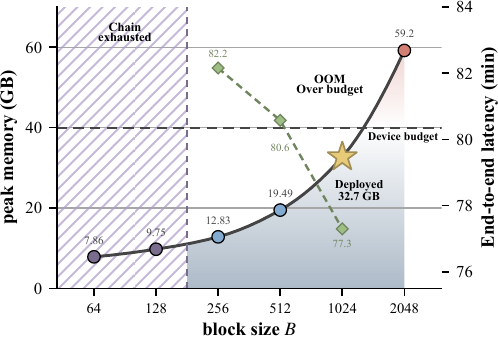}
\caption{HEScan block-size tradeoff between depth and residency at $L{=}2048$.}
\label{fig:scan_residency}
\end{figure}

\subsection{Additional Deployment Evaluations}
\label{sec:additional-deployment}

Additional evaluations establish that \system{} extends beyond the headline Mamba-base, A100-40GB, and $L{\le}4096$ operating point.
Appendix~\ref{app:dimension-robustness} sweeps hidden width, head count, state dimension, and layer depth across configurations up to \system{}-370M and confirms that latency follows the theoretical state-chunk scaling law.
Appendix~\ref{app:cross-ssm-family} implements S4D, S5, LRU, Mamba-1/S6, Mamba-2, and Mamba-3 through the same HEScan interface, validating a shared encrypted primitive across SSM recurrences.
Appendix~\ref{app:gpu-memory-envelope} runs the native backend through $L{=}2048$ on T4, Titan, A100-40GB, A100-80GB, H100, and H200 GPUs, establishing portability across different hardware platforms.
Appendix~\ref{app:longseq} reports execution at the longer sequence length $L{=}8192$ with secure carry refresh at depth $24$.

% !TEX root = ../main.tex
\section{Conclusion}

We introduced \system{}, a hybrid FHE--MPC pipeline that makes selective SSMs practical for private long-context inference. Its factorized scan-contract achieves linear encrypted work and logarithmic depth while reducing memory and communication costs. Specialized MPC protocols further reduce online nonlinear cost by up to $7.7\times$. Across sequence lengths from $L{=}128$ to $L{=}4096$, \system{} executes a native 12-layer \mambabase{} model end to end on a single GPU. At $L{=}2048$, it completes inference in $77.3$ minutes on one A100 with $32.7$\,GB peak memory. It maintains near-plaintext accuracy across three long-document tasks, reaches $36.0$ minutes on four A100s, runs across six GPU classes, and applies the same scan-contract interface across representative invariant and selective SSMs. Together, these results establish encrypted selective SSMs as a scalable foundation for private long-context inference beyond encrypted attention.

\runinhead{Limitations and Future Work.} The end-to-end evaluation targets semi-honest, single-query document classification with Mamba-2, while the same interface instantiates more SSM families including Mamba-3. Mamba-3 exposes a separate crypto-friendliness challenge because its factor construction shifts more work into the secure nonlinear path. Future work can optimize Mamba-3's secure factor construction, extend to malicious threat models, autoregressive decoding, and model-agnostic automatic FHE--MPC implementations.

% Ethics Considerations is inline in ../tex/main.tex (not a section file); kept
% in sync here. build_arxiv.sh warns if the two drift.  [SYNC: ethics]
\section{Ethics Considerations}
This work introduces \system{}, a system for privacy-preserving machine-learning inference.
It does not introduce new attack capabilities, train on sensitive data, or generate harmful content.
The underlying models are publicly available pretrained checkpoints fine-tuned on public benchmarks.
No human subjects data was collected, and no personally identifiable information was used beyond these public benchmarks.

% Bibliography: references.bib is shared from ../tex. No \nocite list, matching
% ../tex/main.tex exactly, so the reference list is identical (nothing forced
% in beyond what's actually \cite'd in the shared sections).  [SYNC: nocite]
\bibliographystyle{IEEEtran}
\bibliography{references}

% !TeX root = ../main.tex
\appendices

% Float packing for the appendix, which carries 37 floats across roughly
% twenty pages. The defaults reserve too much of each page for text and
% allow too few floats per page, so the queue backs up and LaTeX falls back
% to float-only pages. These settings are declared here, after the body has
% been typeset, so the main-paper layout is unaffected.
\setcounter{topnumber}{3}
\setcounter{bottomnumber}{2}
\setcounter{totalnumber}{5}
\setcounter{dbltopnumber}{3}
\renewcommand{\topfraction}{0.95}
\renewcommand{\bottomfraction}{0.85}
\renewcommand{\textfraction}{0.06}
\renewcommand{\floatpagefraction}{0.80}
\renewcommand{\dbltopfraction}{0.95}
\renewcommand{\dblfloatpagefraction}{0.80}

\section*{Appendix Roadmap}
This appendix opens with the released artifact and then follows the technical dependencies of \system{}.
Appendix~\ref{app:open-science} describes the released package and the checks it supports.
Appendix~\ref{app:model-details} clarifies notation and model semantics.
Appendix~\ref{app:encrypted-kernels} gives the detailed encrypted kernel construction, packing, resident schedule, correctness proof, complex reductions, and scan-network choice.
Appendix~\ref{app:hybrid-security} specifies the FHE--MPC boundary, nonlinear protocols, deployed block protocol, and security argument.
Appendix~\ref{app:cost-accounting} derives operation, communication, and preprocessing costs.
Appendix~\ref{app:evaluation-methodology} records the accuracy-evaluation setup, baseline provenance, hardware limits, key-switch accounting, and protocol-cost comparison.
Appendix~\ref{app:deployment-scaling} gives runtime breakdowns, multi-GPU HEScan execution, model-shape scaling, longer-sequence deployment, and cross-GPU portability evidence.
Appendix~\ref{app:cross-ssm-family} gives the common scan-contract interface, representative invariant and selective SSM implementations including Mamba-3, and cross-family cost results.
Appendix~\ref{app:training-accuracy} gives training details, audit scope, complete accuracy tables, and nonlinear-sweep results.

\section{Artifact and Open Science}
\label{app:open-science}

We provide an anonymous \system{} artifact at \url{https://anonymous.4open.science/r/FESC-23D8}.
The artifact contains the source code, experiment logs, and verification drivers used to check the scan-contract primitive, resident schedule, CKKS parameters, nonlinear protocols, model-shape scaling, and accuracy results.
After review, the same artifact will be de-anonymized and archived with more trained checkpoints, complete logs, and permanent public run instructions.

\section{Notation and Model Semantics}
\label{app:model-details}

\subsection{Notation Conventions}
\label{app:notation}

Table~\ref{tab:notation} lists the recurring symbols whose interpretation depends on context.
Bold symbols denote packed vectors or tensors in the kernel-packing sections,
while unbolded symbols denote abstract recurrence operands in the background
and cross-family equations. Fully indexed expressions such as $x_k[h,p]$ and
$B_k[g,i]$ denote scalar components. Mathematical token indices are one-based, while algorithmic array and block indices are zero-based.
All other symbols retain their local definitions.

\begin{table*}[!t]
\centering
\caption{Notation conventions across \system{}.}
\label{tab:notation}
\scriptsize
\setlength{\tabcolsep}{6pt}
\renewcommand{\arraystretch}{1.13}
\begin{tabularx}{\textwidth}{
@{}
>{\raggedright\arraybackslash}p{0.1\textwidth}
|
>{\raggedright\arraybackslash}X
@{}}
\toprule
\rowcolor{tabheader}
\textbf{Symbol} & \textbf{Meaning and convention} \\
\midrule

$\Delta,\boldsymbol{\Delta}$ &
\textbf{CKKS.} $\Delta$ denotes the encoding scale.
\textbf{SSM.} $\Delta_k[h]$ is the Mamba timestep and
$\boldsymbol{\Delta}$ is the timestep tensor.
\textbf{Evaluation.} $\Delta$ denotes an accuracy difference. \\

$N,n_{\mathrm{ckks}}$ &
$N$ is the CKKS ring degree and
$n_{\mathrm{ckks}}=N/2$ is the complex-slot count.
\\

\makecell[l]{$d_{\mathrm{model}},d_{\mathrm{inner}},$\\$d_{\mathrm{poly}},d_s$} &
$d_{\mathrm{model}}$ is model hidden width,
$d_{\mathrm{inner}}=E=HP$ is expanded width,
$d_{\mathrm{poly}}$ is nonlinear-approximation degree, and
$d_s$ is SSM state dimension. \\

\makecell[l]{$s_{\mathrm{state}},c_{\mathrm{state}},$\\$s_y$} &
$s_{\mathrm{state}}$ is the active state-packing capacity,
$c_{\mathrm{state}}=s_{\mathrm{state}}/d_s$ is the number of expanded
channels stored in one state chunk, and $s_y$ is the effective number of
real output entries stored in one ciphertext. \\

$E,e,\kappa,\lambda_\kappa$ &
In the HEScan kernel, $E=HP$ is the expanded channel count,
$e=hP+p$ indexes an expanded channel,
$\kappa$ indexes a state chunk, and
$\lambda_\kappa$ maps chunk-local state coordinates to CKKS slots. \\

$G,N_{\mathrm{gpu}}$ &
$G$ is the number of SSM groups, and $N_{\mathrm{gpu}}$ is the number of GPUs
used by the multi-GPU execution. \\

$t_{\mathrm{NR}}$ &
In $\Pi_{\mathsf{NeRMS}}$, $t_{\mathrm{NR}}$ denotes the number of
Newton iterations. \\

$\mathrm{i}, i$ &
Upright $\mathrm{i}$ is the imaginary unit in complex CKKS slots.
Italic $i$ indexes the SSM state coordinate. \\

$B_{\mathrm{blk}},B_j,\mathbf{B}_k$ &
$B_{\mathrm{blk}}$ is the resident HEScan block capacity, $B_j$ the valid
token count in block $j$, and $\mathbf{B}_k$ the dynamic SSM input factor
with scalar entry $B_k[g,i]$. \\

$F,\mathbf{s}_k,\mathcal{S}_j^{(\kappa)}$ &
$F$ is the number of fractional bits in the MPC fixed-point encoding.
In the HEScan kernel, $\mathbf{s}_k=\mathbf{x}_k\otimes_g\mathbf{B}_k$ is the
additive state update and $\mathbf{s}_k^{(\kappa)}$ is its chunk-local form.
$\mathcal{S}_j^{(\kappa)}$ is the affine summary of block $j$. \\

$K$ &
$K$ is the number of NeRMS initializer buckets. \\

\makecell[l]{$K_s,K_{\mathrm{blk}},$\\$K_y,K_{\mathrm{fac}}$} &
$K_s$ is the number of state chunks,
$K_{\mathrm{blk}}$ is the number of sequence blocks,
$K_y$ is the number of output ciphertexts, and
$K_{\mathrm{fac}}$ is the number of resident compact-factor ciphertexts. \\

\makecell[l]{$g(h),\mathbf{g}^{\mathrm{out}}_k,$\\$\boldsymbol{\gamma}_{\mathrm{rms}}$} &
$g(h)$ maps head $h$ to its SSM group.
$\mathbf{g}^{\mathrm{out}}_k$ is the Mamba output gate.
$\boldsymbol{\gamma}_{\mathrm{rms}}$ is the RMSNorm weight. \\

$\delta$ &
In the SSM background, $\delta$ denotes a continuous-time discretization step. \\

$\sigma,\sigma_{\mathrm{stat}}$ &
$\sigma(\cdot)$ denotes the logistic sigmoid.
$\sigma_{\mathrm{stat}}$ denotes the statistical-security parameter. \\

\bottomrule
\end{tabularx}
\end{table*}

\subsection{Mamba-2 Recurrence Details}
\label{app:mamba-details}

This appendix expands the factor construction and discretization details summarized in \S\S\ref{sec:ssm_model} and~\ref{sec:mamba_model}.

\runinhead{Factor Construction.} Let $X\in\mathbb{R}^{L\times d_{\mathrm{model}}}$ be the block input after the surrounding normalization.
The mixer has $H$ SSM heads, per-head channel dimension $P$, $G$ groups, state dimension $d_s$, and inner width $d_{\mathrm{inner}}=HP$.
Algebraically, one input projection produces the gate branch, the SSM factor branch, and the timestep branch
\begin{align*}
  [z_k,\; xBC_k,\; dt_k]
  &=
  X_k W_{\mathrm{in}} + b_{\mathrm{in}},
  \\
  W_{\mathrm{in}}
  &\in
  \mathbb{R}^{d_{\mathrm{model}}\times(2HP + 2Gd_s + H)},
  \\
  b_{\mathrm{in}}
  &\in
  \mathbb{R}^{2HP + 2Gd_s + H}.
\end{align*}
Here $z_k\in\mathbb{R}^{H\times P}$ is the gate branch, $dt_k\in\mathbb{R}^{H}$ is the timestep branch, and $xBC_k\in\mathbb{R}^{HP+2Gd_s}$ contains the SSM input stream and the dynamic $B,C$ factors.
If an implementation stores the gate projection as a separate matrix $W_z\in\mathbb{R}^{d_{\mathrm{model}}\times HP}$, it is algebraically equivalent to concatenating that projection into $W_{\mathrm{in}}$.
\system{} evaluates these plaintext-weight projections with the same FHE input-projection kernel.

Only the $xBC$ branch passes through the causal depthwise convolution and SiLU
\begin{align*}
  \xi_k
  &=
  \operatorname{SiLU}
  \bigl(
    \operatorname{Conv1D}_{\mathrm{causal}}(xBC)_k
  \bigr),
  \\
  [x^{\mathrm{raw}}_k,\; B_k,\; C_k]
  &=
  \operatorname{split}(\xi_k),
\end{align*}
where $x^{\mathrm{raw}}_k\in\mathbb{R}^{H\times P}$ and
$B_k,C_k\in\mathbb{R}^{G\times d_s}$.
We reserve $x_k$ for the step-scaled scan input.
The timestep, scan input, decay, and output gate are
\begin{align*}
  \Delta_k[h]
  &=
  \operatorname{softplus}\bigl(dt_k[h]+b_{\Delta,h}\bigr),
  \\
  A_h
  &=
  -\exp(A_{\log,h}),
  \\
  x_k[h,p]
  &=
  \Delta_k[h]\,x^{\mathrm{raw}}_k[h,p],
  \\
  a_k[h]
  &=
  \exp\bigl(\Delta_k[h]A_h\bigr),
  \\
  g^{\mathrm{out}}_k[h,p]
  &=
  \operatorname{SiLU}\bigl(z_k[h,p]\bigr).
\end{align*}
Since $A_h<0$ and $\Delta_k[h]>0$ in the source recurrence, $a_k[h]\in(0,1)$ is a stable decay.
The block output is
\begin{align*}
  \hat y_k
  &=
  \operatorname{RMSNorm}_{\boldsymbol\gamma_{\mathrm{rms}}}(\tilde y_k),
  \\
  o_k
  &=
  \hat y_k W_{\mathrm{out}},
  \qquad
  W_{\mathrm{out}}\in\mathbb{R}^{HP\times d_{\mathrm{model}}}.
\end{align*}

\runinhead{Zero-Order-Hold Discretization.} For the continuous system in \S\ref{sec:ssm_model} with diagonal $\mathbf A$ and scalar step $\delta$, zero-order hold gives
\[
  \bar{\mathbf A}
  =
  \exp(\delta\mathbf A),
  \qquad
  \bar{\mathbf B}
  =
  \left(\int_0^\delta \exp(\mathbf A\tau)\,d\tau\right)\mathbf B .
\]
When $\mathbf A$ is invertible, the input map can be written as
\[
  \bar{\mathbf B}
  =
  \mathbf A^{-1}
  \bigl(\exp(\delta\mathbf A)-\mathbf I\bigr)\mathbf B
  =
  (\delta\mathbf A)^{-1}
  \bigl(\exp(\delta\mathbf A)-\mathbf I\bigr)
  \delta\mathbf B .
\]
Mamba-2 keeps the exponential decay from this discretization, where
$a_k[h]=\exp(\Delta_k[h]A_h)$ is the per-head diagonal transition for head $h$.
For the input map, Mamba-2 uses the simplified selective-scan parameterization
\[
  \bar{\mathbf B}_k \approx \Delta_k \mathbf B_k,
\]
so the scan input is
\[
  x_k[h,p]
  =
  \Delta_k[h]x^{\mathrm{raw}}_k[h,p],
\]
and the dynamic factor $B_k$ is applied separately in the grouped state update~\cite{dao2024mamba2}.
Thus only the decay branch requires an online exponential, which \system{} evaluates in MPC, while the input branch requires one elementwise product between $\Delta_k$ and $x^{\mathrm{raw}}_k$.

These model details justify the compact packet used by the encrypted scan contract in \S\ref{sec:ks-he:scan}.
After MPC factor construction, \system{} sends $(x_k,a_k,B_k,C_k)$ to CKKS, expands state-shaped products only inside the current HEScan chunk, contracts the scanned state with $C_k$, and returns the $H\times P$ output to MPC for the skip and gate, after which the gated result passes through the hybrid NeRMS normalization before HEOutProj applies $W_{\mathrm{out}}$ in CKKS.

\section{Encrypted Kernel Construction}
\label{app:encrypted-kernels}

\subsection{CKKS Backend and Parameters}
\label{app:he_choice}

The main text fixes CKKS as the FHE backend and reports the stage-level parameter sets in Table~\ref{tab:ckks_config}, while this appendix records the primitive timings that justify the concrete GPU backend choice.
The hybrid split keeps the headline encrypted path within leveled CKKS because nonlinear factor construction runs in MPC, the deployed depth fits the coefficient-modulus budget audited in Appendix~\ref{app:bootstrap-boundary}, and longer-sequence refresh is evaluated separately in Appendix~\ref{app:longseq}.
Thus the relevant operating point is GPU CKKS without bootstrapping.

Table~\ref{tab:he_primitive_bench} compares PhantomFHE~\cite{yang2024phantom}, the GPU backend used by \system{}, with DESILO~\cite{desiloFHE}, a CKKS library with bootstrapping support.
All measurements use one A100 GPU and the listed leveled chains.
PhantomFHE is faster for every measured primitive, so the hybrid design that avoids bootstrapping in the headline path lets \system{} leverage the faster GPU backend.

\begin{table}[ht]
\centering
\caption{CKKS primitive latency in milliseconds.}
\label{tab:he_primitive_bench}
\label{tab-he-primitive-bench}
\scriptsize
\setlength{\tabcolsep}{4pt}
\renewcommand{\arraystretch}{1.12}
\begin{tabular}{@{}l|l|ccccc@{}}
\toprule
\rowcolor{tabheader}
Chain & Library & add & pt mul & ct mul & rot & conj \\
\midrule
\multirow{2}{*}{\makecell{3-prime\\$[60,40,60]$}}
& PhantomFHE & 0.06 & 0.06 & 0.74 & 0.35 & 0.05 \\
& DESILO     & 0.55 & 3.40 & 3.92 & 13.4 & 9.38 \\
\midrule
\multirow{2}{*}{\makecell{10-prime\\$[60,40^{\times 8},60]$}}
& PhantomFHE & 0.07 & 0.07 & 2.07 & 1.00 & 0.04 \\
& DESILO     & 0.54 & 3.54 & 3.95 & 13.2 & 8.59 \\
\bottomrule
\end{tabular}
\end{table}

\subsection{Projection Kernel Breakdown}
\label{app:he-shifts}
\label{app:he-kernels:packing}
\label{app:he-kernels:proj}

This subsection expands the HEInProj and HEOutProj kernels introduced in \S\ref{sec:ks-he:ctpt}.
It fixes the segment shifts used by the packed ciphertext layout, then gives the BSGS plaintext-ciphertext projection, offline diagonal encoding, online lane extraction, and key-switch sources used by the layer-level ledger.
The kernel follows the packed CKKS matrix multiplication form used by BLB~\cite{xu2025blb} and EncFormer~\cite{zhu2026encformer}.

\runinhead{Segment View.} Following the notation in \S\ref{sec:ks-he:packing}, we view the slots as
\[
  n=\ell_{\mathrm{seg}}N_{\mathrm{seg}},
  \qquad
  \operatorname{idx}(r,s)=s\ell_{\mathrm{seg}}+r,
\]
where $r\in\{0,\ldots,\ell_{\mathrm{seg}}-1\}$ indexes positions inside one segment and
$s\in\{0,\ldots,N_{\mathrm{seg}}-1\}$ indexes the segment.
In \system{}, a segment stores the corresponding axis of a state or factor packing. For a slot vector $\mathbf{x}\in\mathbb{C}^n$, define segment and intra-segment shifts by
\[
\begin{aligned}
  (\Phi^\delta\mathbf{x})_{\operatorname{idx}(r,s)}
  &=
  \mathbf{x}_{\operatorname{idx}(r,(s+\delta)\bmod N_{\mathrm{seg}})},\\
  (\Psi^t\mathbf{x})_{\operatorname{idx}(r,s)}
  &=
  \mathbf{x}_{\operatorname{idx}((r+t)\bmod \ell_{\mathrm{seg}},s)} .
\end{aligned}
\]
We use the same symbols for their homomorphic realizations on ciphertexts.
A whole-segment shift is one CKKS rotation
\[
  \Phi^\delta(\langle\mathbf{x}\rangle)
  =
  \rho(\langle\mathbf{x}\rangle;\delta\ell_{\mathrm{seg}}).
\]
An intra-segment shift is implemented with two rotations and public masks
\[
  \Psi^t(\langle\mathbf{x}\rangle)
  =
  \mathbf{M}^{\mathrm{stay}}_t
  \odot
  \rho(\langle\mathbf{x}\rangle;t)
  +
  \mathbf{M}^{\mathrm{wrap}}_t
  \odot
  \rho(\langle\mathbf{x}\rangle;t-\ell_{\mathrm{seg}}).
\]
Here $\mathbf{M}^{\mathrm{stay}}_t$ keeps entries that do not cross a segment boundary and
$\mathbf{M}^{\mathrm{wrap}}_t$ keeps entries that wrap within the same segment.
When only the first $C$ segments are active, $\Phi^\delta_C$ denotes the same rotation-and-mask construction restricted to the first $C\ell_{\mathrm{seg}}$ slots.
These local maps use rotations, plaintext masks, and additions only.

\runinhead{Packing and Online Cost.} The projection kernel evaluates plaintext-weight maps of the form
\[
  Y=XW,
  \qquad
  X\in\mathbb{R}^{T\times d_{\mathrm{in}}},
  \quad
  W\in\mathbb{R}^{d_{\mathrm{in}}\times d_{\mathrm{out}}},
\]
where $T\le \ell_{\mathrm{seg}}$ is the token block stored inside each segment.
Let $C$ be the number of active feature segments in one ciphertext.
Input block $g$ is packed as
\[
  \mathbf{x}^{(g)}_{\operatorname{idx}(r,c)}
  =
  X[r,gC+c],
  \qquad
  c=0,\ldots,C-1,
\]
with out-of-range entries set to zero.
Output block $b$ uses the same convention for columns $bC,\ldots,bC+C-1$ of $Y$.
If $d_{\mathrm{in}}$ or $d_{\mathrm{out}}$ exceeds $C$, the matrix is split across ciphertext blocks.

\system{} uses the standard complex BSGS plaintext-ciphertext matrix multiplication kernel for this packing, which pairs two real input blocks in the real and imaginary CKKS channels
\[
  \langle\widetilde{\mathbf{x}}^{(u)}\rangle
  =
  \langle
    \mathbf{x}^{(2u)}
    +
    \mathrm{i}\mathbf{x}^{(2u+1)}
  \rangle .
\]
Choose BSGS parameters $n_1n_2=C$.
For output block $b$ and giant step $p$, the kernel forms
\[
  \langle\widetilde{\mathbf{c}}^{(b)}_p\rangle
  =
  \sum_{u}
  \sum_{q=0}^{n_1-1}
  \Phi_C^q
  \bigl(
    \langle\widetilde{\mathbf{x}}^{(u)}\rangle
  \bigr)
  \odot
  \widetilde{\mathbf{D}}^{(b)}_{u,p,q},
\]
where $\widetilde{\mathbf{D}}^{(b)}_{u,p,q}$ is the pre-encoded plaintext diagonal containing the two real weight slices for paired input blocks.
The real output accumulator is extracted as
\[
  \langle\mathbf{c}^{(b)}_p\rangle
  =
  \frac{1}{2}
  \left(
    \langle\widetilde{\mathbf{c}}^{(b)}_p\rangle
    +
    \overline{
      \langle\widetilde{\mathbf{c}}^{(b)}_p\rangle
    }
  \right),
\]
and the giant steps are folded by segment shifts
\[
  \langle\mathbf{y}^{(b)}\rangle
  =
  \sum_{p=0}^{n_2-1}
  \Phi_C^{pn_1}
  \bigl(
    \langle\mathbf{c}^{(b)}_p\rangle
  \bigr).
\]
All plaintext diagonals are encoded offline and reused across requests.
Figure~\ref{fig:cpmm-mamba} shows a toy projection instance.
Key-switches arise from the rotations in $\Phi_C^q$ and $\Phi_C^{pn_1}$ and from conjugations used for real-lane extraction, while plaintext diagonal multiplications and additions do not key-switch.

\begin{figure}[t]
\centering
\includegraphics[width=\columnwidth,keepaspectratio,page=9,trim=502pt 237pt 577pt 155.04pt,clip]{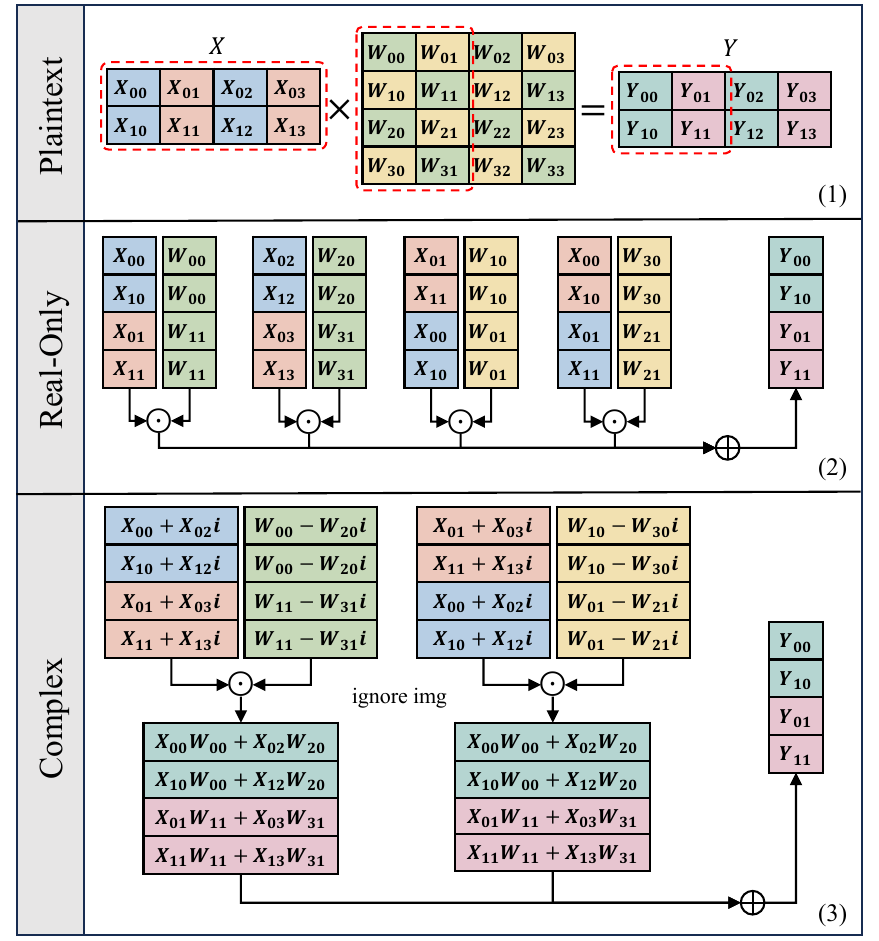}
\caption{Toy BSGS plaintext-ciphertext projection packing for the HEInProj kernel with $T{=}2$, $C{=}4$, and $n_1{=}n_2{=}2$. The toy multiplies $X\in\mathbb{R}^{2\times4}$ by $W\in\mathbb{R}^{4\times4}$ to produce $Y\in\mathbb{R}^{2\times4}$, while the complex path pairs two real feature blocks in the CKKS real and imaginary channels.}
\label{fig:cpmm-mamba}
\end{figure}

\runinhead{Use in \system{}.} HEInProj uses this method to produce $\mathbf{z}$, $\mathbf{xBC}$, and $\mathbf{dt}$.
A separately stored gate matrix $W_z$ is evaluated in the same fused projection call.
HEOutProj uses the same method for $W_{\mathrm{out}}$ after NeRMS.

\subsection{HEScan Kernel Packing and Correctness}
\label{app:hescan-details}

\phantomsection
\label{app:hescan-packing}
\runinhead{Exact State and Slot Maps.} Let $E=HP$ and $c_{\mathrm{state}}=s_{\mathrm{state}}/d_s$, let $\mathcal{E}_{\kappa}$ and $\lambda_{\kappa}$ be the chunk channel set and slot map of Eq.~\ref{eq:hescan-chunk-map}, and for an expanded channel $e=hP+p$ define $h(e)=\lfloor e/P\rfloor$ and $p(e)=e\bmod P$.
For a sequence block, the compact source packings are
\[
\begin{aligned}
  \lambda_a(k,h)
  &=
  kH+h,
  \\
  \lambda_x(k,e)
  &=
  kE+e,
  \\
  \lambda_B(k,g,i)
  &=
  kGd_s+gd_s+i,
  \\
  \lambda_C(k,g,i)
  &=
  kGd_s+gd_s+i.
\end{aligned}
\]
Writing $j=\lambda_{\kappa}(e,i)$, the public broadcasts satisfy
\[
\begin{aligned}
  [\mathrm{br}^{(\kappa)}(\mathbf{a}_k)]_j
  &=
  a_k[h(e)],
  \\
  [\mathrm{br}^{(\kappa)}(\mathbf{x}_k)]_j
  &=
  x_k[h(e),p(e)],
  \\
  [\mathrm{br}^{(\kappa)}(\mathbf{B}_k)]_j
  &=
  B_k[g(h(e)),i],
  \\
  [\mathrm{br}^{(\kappa)}(\mathbf{C}_k)]_j
  &=
  C_k[g(h(e)),i].
\end{aligned}
\]
The reduction map is
\[
  [\Sigma_{d_s}(\mathbf{v}^{(\kappa)})]_e
  =
  \sum_{i=0}^{d_s-1}
  v^{(\kappa)}_{\lambda_{\kappa}(e,i)},
  \qquad
  e\in\mathcal{E}_{\kappa},
\]
with inactive slots masked to zero.
In sequence packing, token $k$ is placed at $(e-\kappa c_{\mathrm{state}})L+k$ before each chunk emits its slice of the global $LHP$ output stream.

\phantomsection
\label{app:resident-hescan-protocol}

\phantomsection
\label{app:hescan-correctness}
\begin{proof}[\textnormal{\textbf{Correctness Proof}}]
For each state chunk $\kappa$, the tuple $\mathcal{T}_k^{(\kappa)}$ represents the affine map induced by token $k$.
Expanding right-after-left composition and applying multiplicative compatibility of $\mathrm{br}^{(\kappa)}$ gives Eq.~\ref{eq:compact_affine_compose}, so the operator is exact function composition and remains associative.

Within one block, the prefix network produces the same map as packet-induced sequential evaluation.
For longer sequences, $\mathcal{S}_{j}^{(\kappa)}$ is the exact map of block $j$, while $\mathcal{G}_{j}^{(\kappa)}$ composes all preceding block maps.
Equation~\ref{eq:resident-block-schedule} therefore gives the global prefix through local position $t$, and replay leaves its value unchanged. Finally, Eq.~\ref{eq:hescan-contract-output} applies the output factor to every exact prefix chunk and sums a partition of the full recurrent coordinates after alignment to the common output packing.
The resulting ciphertext contains the same $\mathbf{m}_k$ as Eq.~\ref{eq:fesc_contract}.
Only CKKS arithmetic and the specified conversion quantization separate the deployed result from this real-arithmetic value.
\end{proof}

\phantomsection
\label{app:fused-fhe-post-ssm}
\runinhead{Fused Post-SSM Path.} In the secure end-to-end pipeline, $\langle \mathbf{m}_k\rangle$ is converted to MPC shares before the skip and gate path.
For the standalone fused FHE component benchmark, the same algebra remains in CKKS.
\begin{align*}
  \langle \mathbf{y}_k\rangle
    &=
    \langle \mathbf{m}_k\rangle
    +
    D_{\mathrm{skip}}\odot\langle \mathbf{x}_{\mathrm{skip},k}\rangle,
    \\
  \langle\tilde{\mathbf{y}}_k\rangle
    &=
    \langle \mathbf{y}_k\rangle\odot\langle\mathbf{g}^{\mathrm{out}}_k\rangle.
\end{align*}

\subsection{Complex-Packed Reduction Accounting}
\label{app:complex-reduction-accounting}

When two ciphertexts encode real slot vectors
\[
  \mathbf{v}^{(0)},\mathbf{v}^{(1)}\in\mathbb{R}^n
\]
with the same $d$-axis packing, \system{} reduces both with one complex-packed rotate-add tree.
It forms
\[
  \langle \boldsymbol{\chi}\rangle
  =
  \langle \mathbf{v}^{(0)}\rangle
  +
  \mathrm{i}\,\langle \mathbf{v}^{(1)}\rangle
  =
  \langle \mathbf{v}^{(0)}+\mathrm{i}\mathbf{v}^{(1)}\rangle ,
\]
where multiplication by $\mathrm{i}$ is a plaintext multiplication.
Linearity of $\Sigma_d$ gives
\[
  \langle \boldsymbol{\Xi}\rangle
  =
  \Sigma_d\bigl(\langle \boldsymbol{\chi}\rangle\bigr)
  =
  \left\langle
    \Sigma_d(\mathbf{v}^{(0)})
    + \mathrm{i}\,\Sigma_d(\mathbf{v}^{(1)})
  \right\rangle .
\]
One conjugation then separates the two real reductions
\begin{align}
  \left\langle\Sigma_d(\mathbf{v}^{(0)})\right\rangle
    &= \tfrac12\bigl(\langle\boldsymbol{\Xi}\rangle+\overline{\langle\boldsymbol{\Xi}\rangle}\bigr),
    \label{eq:complex_reduce_unpack_re}\\
  \left\langle\Sigma_d(\mathbf{v}^{(1)})\right\rangle
    &= -\tfrac{\mathrm{i}}{2}\bigl(\langle\boldsymbol{\Xi}\rangle-\overline{\langle\boldsymbol{\Xi}\rangle}\bigr).
    \label{eq:complex_reduce_unpack_im}
\end{align}
The constants $\mathrm{i}$, $1/2$, and $-\mathrm{i}/2$ are plaintext constants, so the only additional key-switch beyond the rotate-add tree is the conjugation in Eqs.~\ref{eq:complex_reduce_unpack_re}--\ref{eq:complex_reduce_unpack_im}.

For \system{}, the contracted state chunks $\langle\mathbf{v}_k^{(\kappa)}\rangle$ are paired as
\[
  (2\ell,2\ell+1)
\]
before the state-axis reduction, and an odd final chunk falls back to the ordinary real $\Sigma_{d_s}$ tree.
With
\[
  K_s=\left\lceil \frac{HPd_s}{s_{\mathrm{state}}}\right\rceil,
\]
the complex-packed $d_s$-reduction contributes
\begin{align*}
  N_{\mathrm{rot}}^{\mathrm{red}}
    &= L\left\lceil\frac{K_s}{2}\right\rceil\lceil\log_2 d_s\rceil,
    &
  N_{\mathrm{conj}}^{\mathrm{red}}
    &= L\left\lfloor\frac{K_s}{2}\right\rfloor.
\end{align*}
When $K_s$ is even, this is
\[
  L\frac{K_s}{2}\bigl(\lceil\log_2 d_s\rceil+1\bigr)
\]
total reduction key-switches, versus
\[
  LK_s\lceil\log_2 d_s\rceil
\]
for the real-only path.

\subsection{Prefix-Network Selection}
\label{app:scan-comparison}

\system{} supports three parallel prefix network topologies for the SSM scan.
For a padded length
\[
  L' = 2^{\lceil\log_2 L\rceil},
\]
each time step $k$ is represented as an affine tuple $(p_k, s_k)$.
The associative operator $\bullet$ composes two consecutive steps by
\[
  (p_j, s_j)\bullet(p_k, s_k)
  =
  (p_j p_k,\; s_j p_k + s_k).
\]
All three networks compute the inclusive prefix
\[
  t_{1:k} = t_1 \bullet \cdots \bullet t_k
\]
for every $k$.

\runinhead{Kogge--Stone Scan.} Kogge--Stone scan~\cite{kogge1973parallel}
is the fully parallel prefix network used when minimum prefix depth is
prioritized. At each stage
\[
  s = 1, \ldots, \lceil\log_2 L'\rceil,
\]
every position $k \geq 2^{s-1}$ combines with the position $2^{s-1}$ steps to its left using values from the previous stage
\[
  t_k^{(s)} \;\leftarrow\; t_{k - 2^{s-1}}^{(s-1)} \;\bullet\; t_k^{(s-1)}, \qquad k \geq 2^{s-1}.
\]
The stride doubles each stage, so after $\lceil\log_2 L'\rceil$ stages every position $k$ accumulates the full prefix $t_{1:k}$.
At stage $s$, the active-position count is
\[
  L' - 2^{s-1}.
\]
The total number of compositions is therefore
\[
  \sum_{s=1}^{\log_2 L'}(L' - 2^{s-1})
  =
  L'\log_2 L' - (L'-1).
\]

\runinhead{Sklansky Scan.} Sklansky scan~\cite{sklansky1960conditional}
keeps logarithmic depth by broadcasting each lower-half prefix across an aligned
upper half. At stage $s$, positions are grouped into aligned blocks of $2^s$.
Every position $k$ in the upper half of block $b$,
\[
  k\in
  [b\cdot 2^s + 2^{s-1},\,(b+1)\cdot 2^s),
\]
combines with the last position of the lower half
\[
  t_k^{(s)} \;\leftarrow\; t_{b\cdot 2^s + 2^{s-1} - 1}^{(s-1)} \;\bullet\; t_k^{(s-1)}.
\]
Exactly $L'/2$ positions update per stage, yielding
\[
  \frac{L'}{2}\lceil\log_2 L'\rceil
\]
total compositions at the same depth $\lceil\log_2 L'\rceil$ as Kogge--Stone but with half the work.
Fan-out is explicit since one source, the last element of each lower half, supplies $2^{s-1}$ destinations simultaneously.

\runinhead{Brent--Kung Scan.} Brent--Kung scan~\cite{brent1982regular}
trades an extra logarithmic pass for fewer composition nodes and lower fan-out.
The network proceeds in two phases.
During the up-sweep,
\[
  s = 1, \ldots, \lceil\log_2 L'\rceil,
\]
every $2^s$-aligned position $k$ combines with position $k - 2^{s-1}$
\[
  t_k^{(s)} \;\leftarrow\; t_{k-2^{s-1}}^{(s-1)} \;\bullet\; t_k^{(s-1)}, \qquad 2^s \mid k.
\]
After the up-sweep, only positions at power-of-two multiples carry complete prefixes. The remaining ``gap'' positions hold only short partial prefixes from the up-sweep stage at which they were last touched.
The down-sweep uses stages
\[
  s = \lceil\log_2 L'\rceil - 1, \ldots, 1
\]
and propagates the missing prefixes by having the rightmost element of every aligned lower half pass its accumulated prefix to the adjacent gap position
\[
  t_{k + 2^{s-1}}^{(\mathrm{dn},s)} \;\leftarrow\; t_k^{(\mathrm{dn},s+1)} \;\bullet\; t_{k+2^{s-1}}^{(\mathrm{up},s)}, \qquad 2^s \mid k,\; 2^{s+1} \nmid k.
\]
The up-sweep performs $L'-1$ compositions, and the down-sweep performs
\[
  L' - 1 - \log_2 L'.
\]
The total composition count is
\[
  2(L'-1) - \log_2 L'
\]
at depth
\[
  2\lceil\log_2 L'\rceil - 1.
\]

Figure~\ref{fig:scan-nets-compare} illustrates the computations of the Kogge--Stone and Sklansky networks at $L=8$, with the Brent--Kung wiring shown in Figure~\ref{fig:scan-tree}.
A filled circle at position $k$ and stage $s$ marks an affine composition $\bullet$, with the diagonal line carrying the cross-input from position $j<k$ at stage $s{-}1$ and the vertical wire carrying the self-input from stage $s{-}1$.

\begin{figure}[t]
\centering
\includegraphics[width=0.6\columnwidth,keepaspectratio,page=6,trim=520pt 105pt 569pt 118.04pt,clip]{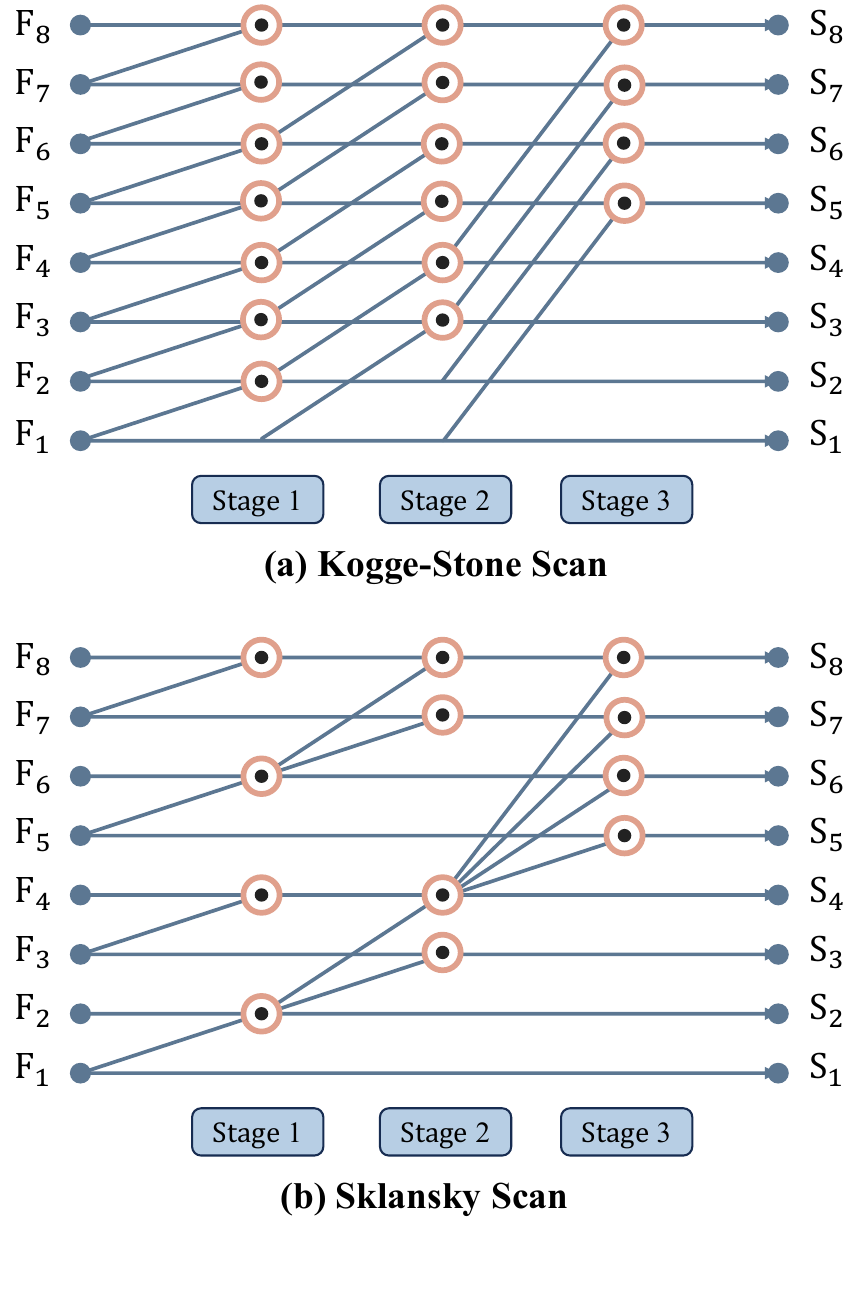}
\caption{Illustrative diagrams for the Kogge--Stone and Sklansky parallel prefix networks at $L=8$.}
\label{fig:scan-nets-compare}
\end{figure}

\subsection{Multiplicative Depth Budget}
\label{app:bootstrap-boundary}

The scan is evaluated without bootstrapping under the FHE-standard 128-bit coefficient-modulus budgets, which are $881$ bits at ring degree $N=32{,}768$ and $1{,}782$ bits at $N=65{,}536$~\cite{bossuat2021efficient}.
A chain of the form $[60,40^{\times k},60]$ consumes $120+40k$ bits, and Table~\ref{tab:ckks-level-audit} reports the deployed-level audit using $N=32{,}768$ for the short-scan rows and $N=65{,}536$ for the long-scan rows, where every deployed length leaves a positive coefficient-modulus margin.

\begin{table}[ht]
\centering
\caption{CKKS level and coefficient-modulus budget.}
\label{tab:ckks-level-audit}
\scriptsize
\setlength{\tabcolsep}{4pt}
\renewcommand{\arraystretch}{1.1}
\begin{adjustbox}{max width=\columnwidth}
\begin{tabular}{@{}c|c|c|c|c@{}}
\toprule
\rowcolor{tabheader}
$L$ & Ring $N$ & Data levels & Chain bits & Margin \\
\midrule
128  & $32{,}768$ & 17 & 800  & $+81$  \\
256  & $32{,}768$ & 19 & 880  & $+1$   \\
512  & $65{,}536$ & 21 & 960  & $+812$ \\
1024 & $65{,}536$ & 23 & 1040 & $+732$ \\
2048 & $65{,}536$ & 22 & 1000 & $+772$ \\
\bottomrule
\end{tabular}
\end{adjustbox}
\end{table}

The projection stages use independent shallow chains and remain at $N=32{,}768$. The larger long-scan ring is used only for the non-bootstrapped HEScan kernel whose depth exceeds the short-ring budget.

% ============================================================

\section{Hybrid Protocols and Security}
\label{app:hybrid-security}

\subsection{CKKS--MPC Boundary Protocols}
\label{app:conversion}

This appendix specifies the CKKS--MPC conversion boundary used by \system{}.
\system{} instantiates the complex-lane conversions of EncFormer~\cite{zhu2026encformer}, whose masking and share-modulus conversion follow BLB and SiRNN~\cite{xu2025blb,rathee2021sirnn}.
The conversion primitive, correctness argument, and semi-honest security are inherited from those protocols.
The \system{} contribution at this boundary is the compact factorized payload and its placement in the selective-SSM pipeline.
The algorithms below state the exact interfaces and parameters used by the deployed system.

\runinhead{Boundary Interface.} At a boundary, a CKKS ciphertext at scale $\Delta$ decrypts and decodes to a packed slot vector
\[
  \mathbf{u}+\mathrm{i}\mathbf{v}\in\mathbb{C}^{n},
\]
where the real lane $\mathbf{u}$ and imaginary lane $\mathbf{v}$ are fixed-point vectors whose centered representatives fit in $\mathbb{Z}_{2^\ell}$.
The complex conversion protocols are
\[
  \left\{
  \begin{aligned}
    \Pi_{\mathrm{C2M}}^{\mathbb{C}}\bigl(\langle \mathbf{u}+\mathrm{i}\mathbf{v}\rangle\bigr)
    &\rightarrow
    \bigl([[\mathbf{u}]]_{2^\ell},[[\mathbf{v}]]_{2^\ell}\bigr), \\
    \Pi_{\mathrm{M2C}}^{\mathbb{C}}\bigl([[\mathbf{u}]]_{2^\ell},[[\mathbf{v}]]_{2^\ell}\bigr)
    &\rightarrow
    \langle \mathbf{u}+\mathrm{i}\mathbf{v}\rangle .
  \end{aligned}
  \right.
\]
The real-only case is obtained by setting the imaginary lane to zero.

\runinhead{CKKS-to-MPC Conversion.} Let $q=Q_{\mathrm{conv}}$ be the CKKS boundary modulus and let $\hat{\mathbf{p}}=\mathrm{Encode}(\mathbf{u}+\mathrm{i}\mathbf{v},\Delta)$ be the plaintext-ring element encrypted in $\langle\mathbf{m}\rangle$.
Following the secure polynomial-ring masking approach of BLB~\cite{xu2025blb} and EncFormer~\cite{zhu2026encformer}, the server samples a uniform mask
\[
  \hat{\mathbf{r}}\leftarrow R_q
\]
directly in the plaintext ring, not by encoding a random slot vector.
The client decrypts only the masked plaintext-ring element
\[
  \hat{\mathbf{p}}+\mathbf{e}+\hat{\mathbf{r}}\pmod q,
\]
where $\mathbf{e}$ is the CKKS decryption error.
The server keeps $-\hat{\mathbf{r}}$ as its share.
The parties then convert the resulting additive shares from modulus $q$ to $\mathbb{Z}_{2^\ell}$ and decode the two complex lanes locally.

\begin{algorithm}[t]
\caption{Complex CKKS-to-MPC conversion instantiated from EncFormer~\cite{zhu2026encformer}, Algorithm~3}
\label{alg:complex-ckks-to-mpc}
\small
\begin{algorithmic}[1]\raggedright
\Require $P_1$ holds $\langle\mathbf{m}\rangle\in A^2_{N,q}$ at scale $\Delta$
\Ensure Each $P_b$ receives shares $[[\mathbf{u}]]^b_{2^\ell}$ and $[[\mathbf{v}]]^b_{2^\ell}$
\State $P_1$ samples $\hat{\mathbf{r}}\leftarrow R_q$ uniformly.
\State $P_1$ sends $\langle\mathbf{d}\rangle\gets \langle\mathbf{m}\rangle+\hat{\mathbf{r}}$ to $P_0$.
\State $P_1$ sets $[[\hat{\mathbf{t}}]]^q_1\gets-\hat{\mathbf{r}}$.
\State $P_0$ decrypts and sets $[[\hat{\mathbf{t}}]]^q_0\gets\mathrm{Dec}(\langle\mathbf{d}\rangle)$.
\State $P_0,P_1$ invoke $[[\hat{\mathbf{t}}]]_{2^\ell}\gets \Pi_{\mathrm{Field2Ring}}([[\hat{\mathbf{t}}]]_q)$.
\State Each $P_b$ decodes $\mathbf{z}_b\gets\mathrm{Decode}([[\hat{\mathbf{t}}]]^b_{2^\ell})$.
\State $P_b$ outputs $[[\mathbf{u}]]^b_{2^\ell}\gets\Re(\mathbf{z}_b)$ and $[[\mathbf{v}]]^b_{2^\ell}\gets\Im(\mathbf{z}_b)$.
\end{algorithmic}
\end{algorithm}

The two plaintext-ring shares reconstruct the masked value because
\[
  [[\hat{\mathbf{t}}]]^q_0+[[\hat{\mathbf{t}}]]^q_1
  =
  \hat{\mathbf{p}}+\mathbf{e}
  \pmod q .
\]
The ring conversion preserves the centered representative except with statistical error at most $2^{-\sigma_{\mathrm{stat}}}$.
Since CKKS decoding is linear before the final fixed-point rounding, local decoding of the two plaintext-ring shares gives additive shares of the decoded real and imaginary lanes.
Correctness therefore holds up to CKKS decoding error, fixed-point rounding, and the statistical error of the ring conversion.

\runinhead{MPC-to-CKKS Conversion.} For the reverse direction, the parties locally center-lift their fixed-point shares, encode them as plaintext-ring shares, convert those shares to the CKKS boundary modulus, and let the client encrypt one share.
The server adds its own plaintext share to obtain an encryption of the sum.

\begin{algorithm}[t]
\caption{Complex MPC-to-CKKS conversion instantiated from EncFormer~\cite{zhu2026encformer}, Algorithm~4}
\label{alg:complex-mpc-to-ckks}
\small
\begin{algorithmic}[1]\raggedright
\Require Shares $[[\mathbf{u}]]_{2^\ell}$ and $[[\mathbf{v}]]_{2^\ell}$ at scale $\Delta$
\Ensure $P_1$ receives $\langle\mathbf{m}\rangle\in A^2_{N,q}$
\State Each $P_b$ center-lifts its shares to $\tilde{\mathbf{u}}_b,\tilde{\mathbf{v}}_b$.
\State Each $P_b$ encodes $[[\hat{\mathbf{t}}]]^{2^\ell}_b \gets \mathrm{Encode}(\tilde{\mathbf{u}}_b+\mathrm{i}\tilde{\mathbf{v}}_b,\Delta)$.
\State $P_0,P_1$ invoke $[[\hat{\mathbf{t}}]]^q\gets \Pi_{\mathrm{Ring2Field}}([[\hat{\mathbf{t}}]]^{2^\ell})$.
\State $P_0$ encrypts and sends $\langle\mathbf{c}\rangle\gets \mathrm{Enc}_{pk}([[\hat{\mathbf{t}}]]^q_0)$ to $P_1$.
\State $P_1$ outputs $\langle\mathbf{m}\rangle\gets \langle\mathbf{c}\rangle+[[\hat{\mathbf{t}}]]^q_1$.
\end{algorithmic}
\end{algorithm}

The encoded plaintext shares satisfy
\[
  [[\hat{\mathbf{t}}]]^q_0+[[\hat{\mathbf{t}}]]^q_1
  \equiv
  \mathrm{Encode}(\mathbf{u}+\mathrm{i}\mathbf{v},\Delta)
  \pmod q
\]
except for the statistical error of $\Pi_{\mathrm{Ring2Field}}$ and the fixed public rounding used by the encoder.
After the server adds its plaintext share to the encrypted client share, the resulting ciphertext decrypts to the encoded sum up to CKKS encryption noise.

Correctness and semi-honest security follow from the complex-lane conversion analysis of EncFormer~\cite{zhu2026encformer} and the assumed security of the share-extension primitive.
In $\Pi_{\mathrm{C2M}}^{\mathbb{C}}$, the client decrypts a plaintext-ring value masked by fresh uniform server randomness.
In $\Pi_{\mathrm{M2C}}^{\mathbb{C}}$, the server receives only a CKKS encryption of the client's share.
With public tensor shapes and the fixed \system{} schedule, the instantiated conversions add no leakage beyond the public parameters, message sizes, schedule-induced timing, and abort status already included in the functionality of Theorem~\ref{thm:main-security}.

\phantomsection\label{app:ring-field}
\runinhead{Share-Modulus Conversion.} We record the SiRNN/BLB-style share-modulus instantiation used by \system{} for the inherited conversion boundary.
The conversion protocols use additive shares over two moduli.
For any modulus $M$, write $[[w]]_M$ for a two-party additive sharing with $w_0+w_1=w\bmod M$.
Let $\mathrm{cl}_M(w)$ be the centered representative in $(-M/2,M/2]$.
We use the standard share-extension primitive $\Pi_{\mathrm{Ext}}$ from SiRNN/BLB-style conversion protocols~\cite{rathee2021sirnn,xu2025blb}.

For field-to-ring conversion, $\Pi_{\mathrm{Field2Ring}}$ maps $[[w]]_q$ to $[[w]]_{2^\ell}$.
The parties invoke
\[
  \Pi^{q,2^{\ell'}}_{\mathrm{Ext}},
  \qquad
  \ell'\ge \max(\lceil\log_2 q\rceil,\ell),
\]
to obtain shares representing $\mathrm{cl}_q(w)$ in $\mathbb{Z}_{2^{\ell'}}$, except with error at most $2^{-\sigma_{\mathrm{stat}}}$.
They then locally reduce shares modulo $2^\ell$.

For ring-to-field conversion, $\Pi_{\mathrm{Ring2Field}}$ maps $[[w]]_{2^\ell}$ to $[[w]]_q$.
The parties first extend to $\mathbb{Z}_{2^{\ell+\sigma_{\mathrm{stat}}}}$.
Let the extended shares be $w'_0,w'_1$.
They locally output
\[
  [[w]]^{q}_0=w'_0\bmod q,
  \qquad
  [[w]]^{q}_1=(w'_1-2^{\ell+\sigma_{\mathrm{stat}}})\bmod q .
\]
Except with probability at most $2^{-\sigma_{\mathrm{stat}}}$, the two new shares reconstruct the centered value modulo $q$.

\phantomsection\label{app:modtrim}
\runinhead{Adopted Modulus Trimming.} Following EncFormer~\cite{zhu2026encformer}, \system{} modulus-switches a ciphertext to the shortest admissible boundary modulus before CKKS-to-MPC conversion.
A CKKS ciphertext at a boundary need not retain levels for future CKKS multiplications before it enters MPC.
When the current modulus is larger than needed, $P_1$ modulus-switches the ciphertext down to a boundary modulus $q=Q_{\mathrm{conv}}$ satisfying
\[
  \log_2 q \ge \ell+\sigma_{\mathrm{stat}}+1,
  \qquad
  q/2 > \Delta B_{\max}.
\]
The first condition gives statistical slack for share conversion.
The second prevents plaintext-ring wraparound for boundary magnitude $B_{\max}$.
This trimming changes only the transmitted ciphertext size.
It does not change the converted plaintext value, the MPC ring, or the conversion functionality.

\subsection{MPC Nonlinear Calibration and Cost}
\label{app:mpc-nonlinear-details}

Section~\ref{sec:mpc} defines the four MPC nonlinear protocols.
This appendix complements the task-level sweep by reporting the numerical approximations used by each protocol together with the online communicating primitives required by each candidate configuration.
All approximation domains, coefficients, thresholds, bucket boundaries, and iteration counts are fixed before inference and reused across all tasks and sequence lengths.

\runinhead{Approximation Calibration.} MoSiLU and MoSoft share the even-residual construction from \S\ref{sec:mpc-mosilu-mosoft}, and both use threshold $\tau=4$ with Remez minimax polynomials fitted over $[0,\tau^2]$.
Degree four attains maximum errors of $0.072$ for MoSiLU and $0.018$ for MoSoft, while degree six provides no further improvement under the fixed tail rule.
NeRMS uses $K=8$ public initializer buckets and one Newton iteration, with the bucket boundaries and initializers calibrated once from the variance range and then held fixed across datasets and sequence lengths.
MMExp exploits the Mamba decay domain $\Delta_k[h]A_h\leq0$, where $A_h=-\exp(A_{\log,h})$, and uses a degree-four minimax approximation of $\exp(z)$ over $z\in[-8,0]$.
Figure~\ref{fig:mpc_approx_accuracy} compares the selected approximations with their exact functions, and Figure~\ref{fig:mpc_approx_accuracy_err} reports the corresponding pointwise absolute errors and the numerical tradeoffs among the candidate degrees and Newton-step counts.

\begin{figure}[t]
\centering
\includegraphics[width=\linewidth]
{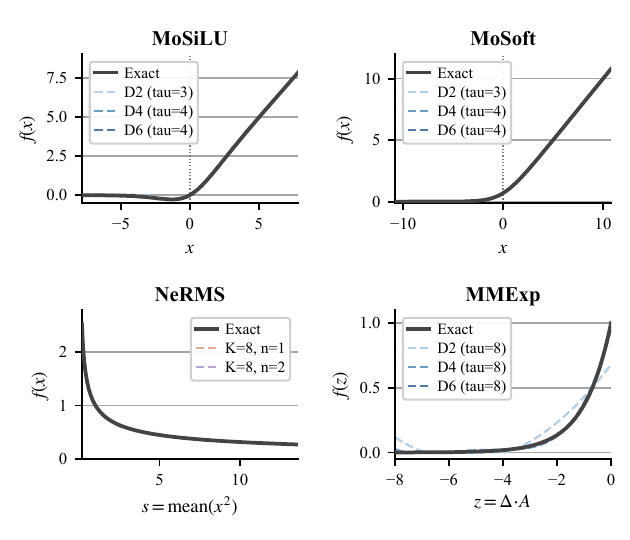}
\caption{Exact functions and calibrated approximations for the MPC nonlinear protocols.}
\label{fig:mpc_approx_accuracy}
\end{figure}

\begin{figure}[t]
\centering
\includegraphics[width=\linewidth]
{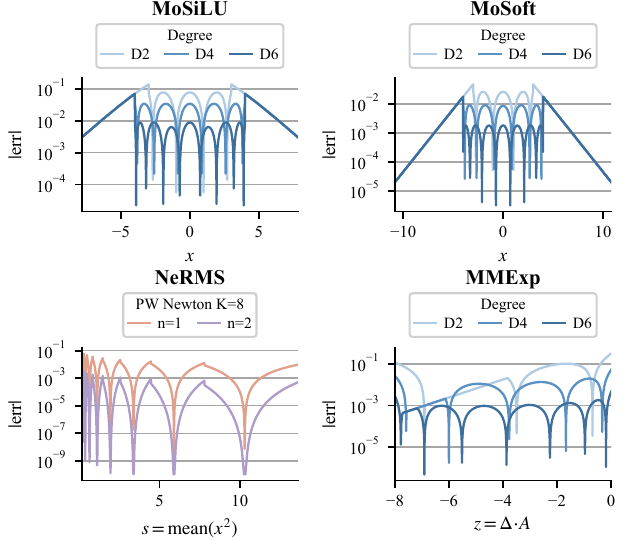}
\caption{Pointwise absolute error over the calibrated approximation domains.}
\label{fig:mpc_approx_accuracy_err}
\end{figure}

\runinhead{Online Primitive Cost.} Table~\ref{tab:mpc_deployed_counts} counts only the communicating primitives defined in Table~\ref{tab:mpc_ops}, since additions, reductions, and public-coefficient operations are local.
FP prod.\ denotes a secure fixed-point product, Cmp.\ a secure comparison, and Mux a secret bit-value selection.
Settings D$d$ use polynomial degree $d$, while N$t$ use $t$ Newton iterations.
MoSiLU and MoSoft have different public coefficients but the same online protocol graph.

\begin{table}[t]
\centering
\caption{Per-output online communicating primitive counts for the nonlinear configurations.}
\label{tab:mpc_deployed_counts}
\scriptsize
\setlength{\tabcolsep}{3.5pt}
\renewcommand{\arraystretch}{1.10}
\begin{tabular}{@{}l|l|l|ccc@{}}
\toprule
\rowcolor{tabheader}
Protocol & Setting & Unit & FP prod. & Cmp. & Mux \\
\midrule
MoSiLU \& MoSoft & D2 & element & $1$ & $2$ & $2$ \\
\rowcolor{oursrow}
MoSiLU \& MoSoft & D4 & element & $2$ & $2$ & $2$ \\
MoSiLU \& MoSoft & D6 & element & $3$ & $2$ & $2$ \\
\midrule
\rowcolor{oursrow}
NeRMS & N1 & token vector & $2$ & $7$ & $0$ \\
NeRMS & N2 & token vector & $5$ & $7$ & $0$ \\
NeRMS & N3 & token vector & $8$ & $7$ & $0$ \\
\midrule
MMExp & D2 & head & $3$ & $1$ & $1$ \\
\rowcolor{oursrow}
MMExp & D4 & head & $5$ & $1$ & $1$ \\
MMExp & D6 & head & $7$ & $1$ & $1$ \\
\bottomrule
\end{tabular}
\end{table}

\subsection{End-to-End \system{} Block Protocol}
\label{app:fesc-block-protocol}

Algorithm~\ref{alg:fesc-block-protocol} instantiates the block pipeline in Figure~\ref{fig:pipeline} as the ordered sequence of FHE kernels, boundary conversions, and MPC nonlinear protocols.

\begin{algorithm}[t]
\caption{\system{} Algorithm for one Block.}
\label{alg:fesc-block-protocol}
\small
\begin{algorithmic}[1]
\Require Encrypted block input $\langle\mathbf X\rangle$ and server-held parameters
\Ensure Encrypted block output $\langle\mathbf o\rangle$
\State $(\langle\mathbf z\rangle,\langle\mathbf{xBC}\rangle,\langle\mathbf{dt}\rangle)
       \gets\mathsf{HEInProj\ kernel}(\langle\mathbf X\rangle)$
\State $\langle\boldsymbol\eta\rangle\gets\mathsf{HEConv\ kernel}(\langle\mathbf{xBC}\rangle)$
\State $([[\mathbf z]],[[\mathbf{dt}]],[[\boldsymbol\eta]])
       \gets\Pi_{\mathsf{C2M}}^{\mathbb C}(\langle\mathbf z\rangle,\langle\mathbf{dt}\rangle,\langle\boldsymbol\eta\rangle)$
\State $[[\boldsymbol\xi]]\gets\Pi_{\mathsf{MoSiLU}}([[\boldsymbol\eta]])$
\State $([[\mathbf x^{\mathrm{raw}}]],[[\mathbf B]],[[\mathbf C]])\gets\mathsf{Split}([[\boldsymbol\xi]])$
\State $[[\boldsymbol\Delta]]\gets\Pi_{\mathsf{MoSoft}}([[\mathbf{dt}]]+[[\mathbf b_\Delta]])$
\State $[[\mathbf a]]\gets\Pi_{\mathsf{MMExp}}([[\boldsymbol\Delta]],[[\mathbf A]])$
\State $[[\mathbf x]]\gets\Pi_\times([[\boldsymbol\Delta]],[[\mathbf x^{\mathrm{raw}}]])$
\State $[[\mathbf g^{\mathrm{out}}]]\gets\Pi_{\mathsf{MoSiLU}}([[\mathbf z]])$
\State $(\langle\mathbf x\rangle,\langle\mathbf a\rangle,\langle\mathbf B\rangle,\langle\mathbf C\rangle)
       \gets\Pi_{\mathsf{M2C}}^{\mathbb C}([[\mathbf x]],[[\mathbf a]],[[\mathbf B]],[[\mathbf C]])$
\State $\langle\mathbf m\rangle\gets\mathsf{HEScan\ kernel}(\langle\mathbf x\rangle,\langle\mathbf a\rangle,\langle\mathbf B\rangle,\langle\mathbf C\rangle)$
\State $[[\mathbf m]]\gets\Pi_{\mathsf{C2M}}^{\mathbb C}(\langle\mathbf m\rangle)$
\State $[[\mathbf y]]\gets[[\mathbf m]]+\Pi_\times([[\mathbf D_{\mathrm{skip}}]],[[\mathbf x^{\mathrm{raw}}]])$
\State $[[\widetilde{\mathbf y}]]\gets\Pi_\times([[\mathbf y]],[[\mathbf g^{\mathrm{out}}]])$
\State $\langle\widetilde{\mathbf y}\rangle\gets\Pi_{\mathsf{M2C}}^{\mathbb C}([[\widetilde{\mathbf y}]])$
\State $\langle\mathbf u\rangle\gets\langle\widetilde{\mathbf y}\rangle\odot\langle\widetilde{\mathbf y}\rangle$
\State $[[\mathbf u]]\gets\Pi_{\mathsf{C2M}}^{\mathbb C}(\langle\mathbf u\rangle)$
\State $[[\mathbf s_{\mathrm{RMS}}]]\gets\Pi_{\mathsf{NeRMS}}([[\mathbf u]],HP,\epsilon)$
\State $\langle\mathbf s_{\mathrm{RMS}}\rangle\gets\Pi_{\mathsf{M2C}}^{\mathbb C}([[\mathbf s_{\mathrm{RMS}}]])$
\State $\langle\widehat{\mathbf y}\rangle\gets
       \boldsymbol\gamma_{\mathrm{RMS}}\odot
       (\langle\widetilde{\mathbf y}\rangle\odot\mathrm{br}(\langle\mathbf s_{\mathrm{RMS}}\rangle))$
\State $\langle\mathbf o\rangle\gets\mathsf{HEOutProj\ kernel}(\langle\widehat{\mathbf y}\rangle)$
\State \Return $\langle\mathbf o\rangle$
\end{algorithmic}
\end{algorithm}

\subsection{Security Analysis}
\label{app:security-analysis}

\runinhead{Functionality and Leakage.}
The ideal functionality receives the private input $x$, the private model $\theta$, and public metadata $\mu$, then returns only the deployed inference output to the client.
\[
  \mathcal{F}_{\mathrm{SelSSM}}(x,\theta;\mu)
  \rightarrow
  \left(
    f^{\mathrm{dep}}_{\theta}(x),
    \bot
  \right).
\]
The function $f^{\mathrm{dep}}_{\theta}$ includes the public approximation, fixed-point, and quantization rules defined in \S\ref{sec:prelim_model}.
Table~\ref{tab:threat-model} summarizes the private values, public metadata, and transcript leakage captured by this functionality. Malicious behavior, party collusion, implementation side channels, and data-dependent hardware leakage are outside the scope of this analysis.

\begin{table}[t]
\centering
\caption{Privacy and leakage of the deployed protocol.}
\label{tab:threat-model}
\scriptsize
\setlength{\tabcolsep}{3pt}
\renewcommand{\arraystretch}{1.10}
\begin{tabularx}{\columnwidth}{@{}>{\raggedright\arraybackslash}p{0.4\columnwidth}|X@{}}
\toprule
\rowcolor{tabheader}
Object & Status \\
\midrule
Client input and secret key
& Private to $P_0$ \\

Model weights
& Private to $P_1$ except through the authorized output \\

Selective factors and scan states
& Private ciphertexts or additive shares \\

CKKS public and evaluation keys
& Available to $P_1$ \\

Final output
& Revealed only to $P_0$ \\

Public metadata
& Architecture, dimensions, length, crypto parameters, packing, block size, and schedule \\

Observable transcript
& Public message count, dimensions, and abort status \\
\bottomrule
\end{tabularx}
\end{table}

The explicit protocol leakage is
\[
  \mathsf{Leak}
  =
  \left(
    \mu,\,
    \{|M_j|\}_{j=1}^{J},\,
    \mathsf{abort}
  \right),
\]
where $J$ is the number of composed protocol components and the number and dimensions of all messages are determined by the public schedule.
No token value, selective factor, recurrent state, comparison result, bucket index, clipping decision, or activation-region decision is opened.

\phantomsection
\label{app:approx-nonlinear-security}
\runinhead{Component Security.} CKKS evaluation blocks reveal no encrypted client value under CKKS semantic security.
The server-held plaintext weights used by these blocks are part of $P_1$'s private input and are not public metadata.

At a CKKS-to-MPC boundary, the client decrypts only
\[
  \hat v + e + \hat r
  \pmod{Q_{\mathrm{conv}}},
\]
where $\hat v$ is the encoded value, $e$ is the CKKS approximation error, and $\hat r$ is uniform over the conversion ring.
The decrypted value is statistically independent of $\hat v$.
At an MPC-to-CKKS boundary, the server receives only an encryption of the client's additive share and combines it with its own share under encryption.

MoSiLU, MoSoft, MMExp, and NeRMS are fixed compositions of $\Pi_{\times}$, $\Pi_{\mathrm{cmp}}$, $\Pi_{\mathrm{mux}}$, truncation, and the boundary protocols, with public coefficients, thresholds, bucket boundaries, iteration counts, tensor shapes, and schedules.
Comparison and selection bits remain shared under a fixed message schedule, while approximation and fixed-point errors are included in the deployed functionalities and do not introduce additional leakage.

\phantomsection
\label{app:pipeline-security}
\begin{proof}[\textnormal{\textbf{End-to-End Security}}]
Fix a corrupted party $P_b$ for $b\in\{0,1\}$.
We construct a simulator $\mathcal{S}_b$ from the corrupted party's input, authorized output, and public leakage.

For a corrupted server, CKKS semantic security hides the client input and all ciphertext intermediates.
The MPC transcripts are simulatable from the server's model input and the prescribed interfaces of MPC functionalities, while MPC-to-CKKS conversion reveals only encryptions of client shares.

For a corrupted client, the model weights remain server-local.
Every intermediate CKKS decryption occurs inside CKKS-to-MPC conversion after the ciphertext has been masked by fresh server randomness.
The resulting value is statistically independent of the underlying intermediate, while the remaining MPC transcripts are simulatable from the client input and prescribed outputs.

The protocol is a fixed sequential composition $\Pi_{\mathsf{Infer}}=\Pi_1\circ\cdots\circ\Pi_J$ of CKKS blocks, MPC blocks, and conversion boundaries.
Each component admits a simulator under the assumptions of Theorem~\ref{thm:main-security}.
Let $\mathsf{H}_j$ denote the execution in which the first $j$ component transcripts are simulated.
Then
\[
  \mathsf{H}_{j-1}\equiv_c\mathsf{H}_j
  \qquad
  \text{for }j=1,\ldots,J .
\]
The real execution is computationally indistinguishable from the ideal execution of $\mathcal{F}_{\mathrm{SelSSM}}$ with leakage $\mathsf{Leak}$.
\end{proof}

\section{Cost and Resource Accounting}
\label{app:cost-accounting}

\subsection{HEScan Kernel Work and Residency}
\label{app:hescan-cost}
\label{app:hescan-work-accounting}

The HEScan kernel work and residency analysis uses the Brent--Kung affine-composition count
\[
  C_{\mathrm{BK}}(n)
  =
  2n-\log_2 n-2
\]
for a power-of-two scan of size $n$.
For block $j$, let $\overline B_j=2^{\lceil\log_2 B_j\rceil}$ and $\overline K_{\mathrm{blk}}=2^{\lceil\log_2 K_{\mathrm{blk}}\rceil}$.
The two-pass resident schedule executes
\[
  C_{\mathrm{res}}
  =
  2\sum_{j=0}^{K_{\mathrm{blk}}-1}
    C_{\mathrm{BK}}(\overline B_j)
  +
  C_{\mathrm{BK}}(\overline K_{\mathrm{blk}})
  +
  \sum_{j=1}^{K_{\mathrm{blk}}-1} B_j .
\]
The terms account for block-summary construction and replay, the summary scan, and carry application.

We charge two ciphertext products for every composition and state chunk. Update construction and output contraction each add one product per token and chunk, giving
\[
  N_{\mathsf{ctmul}}
  =
  2K_sC_{\mathrm{res}}
  +
  2LK_s
  =
  O(LK_s).
\]

Let $R_{\mathrm{comp}}$, $R_{\mathrm{build}}$, and $R_{\mathrm{out}}$ denote the rotations required by affine composition, update construction, and contraction with state-axis reduction.
The aggregate key-switch count is
\[
  N_{\mathsf{ks}}
  =
  N_{\mathsf{ctmul}}
  +
  K_sC_{\mathrm{res}}R_{\mathrm{comp}}
  +
  LK_s
  \left(
    R_{\mathrm{build}}+R_{\mathrm{out}}
  \right)
  +
  N_{\mathsf{conj}}^{\mathrm{pair}} .
\]
For fixed model dimensions and slot maps, $N_{\mathsf{ks}}=O(LK_s)$, and Appendix~\ref{app:complex-reduction-accounting} gives the instantiated rotation and conjugation counts.

With $D_{\mathrm{BK}}(n)=2\lceil\log_2 n\rceil-1$, the multiplicative dependency depth satisfies
\[
  D_{\mathrm{scan}}
  \le
  D_{\mathrm{BK}}(B_{\mathrm{blk}})
  +
  D_{\mathrm{BK}}(\overline K_{\mathrm{blk}})
  +
  1
  =
  O(\log L).
\]

This protocol-level bound assumes the exclusive parallel-summary scan used by the resident HEScan schedule. Peak ciphertext residency satisfies
\[
  N_{\mathsf{live}}
  \le
  c_{\mathrm{blk}}B_{\mathrm{blk}}
  +
  c_{\mathrm{sum}}K_{\mathrm{blk}}
  +
  K_y
  +
  K_{\mathrm{fac}}
  +
  c_0 .
\]
The state-chunk count
\[
  K_s
  =
  \left\lceil
    \frac{Ed_s}{s_{\mathrm{state}}}
  \right\rceil
\]
multiplies total work but not peak residency because the HEScan kernel processes one state chunk at a time. Peak device memory stays within $1.11\times$ at any given length while $K_s$ varies by $8\times$, so residency follows the resident schedule rather than the model shape.

\subsection{Online Communication}
\label{app:comm-analysis}

The online communication is the sum of CKKS--MPC conversion traffic and MPC protocol traffic.
\[
  C_{\mathrm{online}}
  =
  C_{\mathrm{conv}}
  +
  C_{\mathrm{MPC}} .
\]
For boundary $b$ containing $V_b$ real fixed-point values, complex conversion uses
\[
  K_b
  =
  \left\lceil
    \frac{V_b}{2n_{\mathrm{ckks}}}
  \right\rceil
\]
ciphertexts.
If $S_{\mathrm{ct}}(Q_b)$ denotes the serialized ciphertext size at boundary modulus $Q_b$, then
\[
  C_{\mathrm{conv}}
  =
  \sum_b
  \left(
    K_bS_{\mathrm{ct}}(Q_b)
    +
    C_{\mathrm{ext}}(V_b)
  \right),
\]
where $C_{\mathrm{ext}}$ is the ring-to-field or field-to-ring share-conversion traffic.
Rounding is applied independently at each boundary.

The deployed inbound boundary contains
\[
  V_{\mathrm{in}}
  =
  L(H+HP+2Gd_s)
\]
real values from $(\mathbf x,\mathbf a,\mathbf B,\mathbf C)$, while the outbound boundary contains
\[
  V_{\mathrm{out}}
  =
  LHP
\]
values from the contracted output $\mathbf m$.
The state-shaped tensors $\mathbf p$, $\mathbf s$, and $\mathbf h$ never cross the boundary.

The MPC term follows from the deployed MoSiLU/MoSoft D4, NeRMS N1, and MMExp D4 primitive counts in Appendix~\ref{app:mpc-nonlinear-details}.
Consequently,
\[
  C_{\mathrm{online}}
  =
  O\!\left(
    L(E+Gd_s+H)
  \right),
\]
so the complete online path scales linearly with sequence length and compact factor dimensions, eliminating both token-pair work and dense-state conversion.

% ============================================================

\subsection{Offline Preprocessing}
\label{app:offline-budget}

The MPC protocols consume input-independent Beaver triples and comparison correlations generated before inference by standard two-party OT/VOLE/PCG preprocessing~\cite{boyle2020pcg,pang2024bolt,xu2025blb,huang2022cheetah,kei2025shaft}.
For one layer, the deployed MoSiLU/MoSoft D4, NeRMS N1, and MMExp D4 configuration consumes
\[
\begin{aligned}
  T_{\mathrm{layer}}(L)
  &=
  L\Bigl(
    4(E+2Gd_s)+4E+4H+6H
  \\
  &\qquad
    +E+2E+2
  \Bigr)
  \\
  &=
  L(11E+8Gd_s+10H+2)
\end{aligned}
\]
triples.
For $N_L$ layers,
\[
  T_{\mathrm{Beaver}}(L)
  =
  N_LT_{\mathrm{layer}}(L).
\]
With arithmetic width $\ell$, dealer-based generation transmits
\[
  C_{\mathrm{dealer}}
  =
  \frac{6\ell}{8}
  T_{\mathrm{Beaver}}
\]
bytes in aggregate.
The OT-extension estimate used in Table~\ref{tab:offline_triples} is
\[
  C_{\mathrm{OT}}
  \approx
  300T_{\mathrm{Beaver}}
\]
bytes~\cite{ishai2003extending,gilboa1999two,keller2020mpspdz}.
For \mambabase{} with $E{=}1{,}536$, $H{=}24$, $G{=}1$, $d_s{=}128$, and $N_L{=}12$, Table~\ref{tab:offline_triples} reports the resulting preprocessing volume.

\begin{table}[ht]
\centering
\caption{Offline Beaver-triple budget.}
\label{tab:offline_triples}
\scriptsize
\setlength{\tabcolsep}{5pt}
\renewcommand{\arraystretch}{1.1}
\begin{tabular}{@{}c|cccc@{}}
\toprule
\rowcolor{tabheader}
$L$ & Triples/layer & Total triples & Dealer (GB) & OT (GB) \\
\midrule
$128$    & $2.32$\,M & $27.9$\,M  & $0.92$  & $8.4$  \\
$512$    & $9.30$\,M & $111.6$\,M & $3.68$  & $33.5$ \\
$1{,}024$ & $18.60$\,M & $223.2$\,M & $7.37$  & $67.0$ \\
$2{,}048$ & $37.20$\,M & $446.3$\,M & $14.73$ & $133.9$ \\
\bottomrule
\end{tabular}
\end{table}

\section{Evaluation Methodology and Baselines}
\label{app:evaluation-methodology}

\subsection{Training and Accuracy Evaluation Setup}
\label{app:hyperparams}

The Mamba and BERT controls use the same dataset splits, target lengths, document-prefix policy, seed set, and validation-loss checkpoint selection, while retaining their native tokenizers and fixed architecture-specific training recipes.
For $L>512$, BERT-base copies its final pretrained position vector into the added positions and fine-tunes the extended embedding table at the target length.
The reported long-context accuracy therefore reflects full-context evaluation rather than truncation to $512$ tokens.

\subsection{Baseline Evidence and Provenance}
\label{app:baseline_sources}

\runinhead{Evidence Matrix.}
Table~\ref{tab:baseline-evidence-matrix} follows the system order of Table~\ref{tab:e2e_comparison} and records the evidence source, validated execution scope, paper use, and treatment of unavailable long-sequence results.
Artifact-derived values are used within the scope that we successfully validated.

\begin{table*}[t]
\centering
\caption{Evidence and provenance for baseline comparisons.}
\label{tab:baseline-evidence-matrix}
\scriptsize
\setlength{\tabcolsep}{4pt}
\renewcommand{\arraystretch}{1.10}
\newcommand{\baselinegroup}[2][0pt]{\raisebox{#1}{\makecell[c]{\rotatebox[origin=c]{90}{\scriptsize\itshape #2}}}}
\begin{tabularx}{\textwidth}{@{}c|>{\raggedright\arraybackslash}p{0.1\textwidth}|>{\raggedright\arraybackslash}p{0.11\textwidth}|>{\raggedright\arraybackslash}p{0.14\textwidth}|>{\raggedright\arraybackslash}p{0.15\textwidth}|>{\raggedright\arraybackslash}X@{}}
\toprule
\rowcolor{tabheader}
Type & System & Evidence & Validated scope & Paper use & Long-sequence treatment \\
\midrule
\baselinegroup[-0.70em]{FHE} &
AEGIS~\cite{gong2026aegis}
& Paper reported
& BERT-base at $L{=}128$, $512$, and $2048$
& Long-sequence reference
& Paper-reported long-sequence rows are used at the reported lengths. \\
\midrule
\multirow{3}{*}{\baselinegroup[-2.10em]{MPC}} &
MPCFormer~\cite{li2023mpcformer}
& Paper and artifact
& BERT-base at $L{=}128$, $256$, and $512$
& Latency baseline
& Reproduced rows are scaled to the paper network profile. \\
& SIGMA~\cite{gupta2024sigma}
& Artifact reproduced
& BERT-base at $L{=}128$
& LAN-normalized latency
& Reproduced short-context rows are used, with no validated longer executable path. \\
& SHAFT~\cite{kei2025shaft}
& Artifact reproduced
& BERT-base at $L{=}128$
& LAN-normalized latency
& Reproduced short-context rows are used, with no validated longer executable path. \\
\midrule
\multirow{5}{*}{\baselinegroup[-4.00em]{Hybrid}} &
BOLT~\cite{pang2024bolt}
& Paper reported
& BERT-base at $L{=}128$
& Latency baseline
& Paper-reported short-context rows are used, with no validated long-sequence executable path. \\
& BumbleBee~\cite{lu2025bumblebee}
& Artifact reproduced
& BERT-base at $L{=}128$
& Latency baseline
& Reproduced short-context rows are used, with no validated longer executable path. \\
& BLB~\cite{xu2025blb}
& Artifact reproduced
& BERT-base at $L{=}128$
& Latency and key-switch anchor
& Reproduced short-context rows are used, while longer probes exceed the host-memory threshold or expose no executable path. \\
& EncFormer~\cite{zhu2026encformer}
& Artifact reproduced
& BERT-base at $L{=}128$
& Latency and key-switch anchor
& Reproduced short-context encrypted rows are used, with longer key-switch rows generated by its CKKS simulator. \\
& \cellcolor{oursrow}\system{}
& \cellcolor{oursrow}This work
& \cellcolor{oursrow}SSM at $L{=}128$--$2048$
& \cellcolor{oursrow}Deployed system
& \cellcolor{oursrow}Native encrypted rows are measured at every reported length. \\
\bottomrule
\end{tabularx}
\end{table*}

\runinhead{Artifact and Latency Provenance.}
AEGIS~\cite{gong2026aegis} reports four-A100-40GB BERT-base SST-2 end-to-end latencies directly.
Its one-GPU-equivalent entries at $L{=}128$ and $L{=}512$ multiply the reported four-GPU values by the midpoint $3.844\times$ of the paper's $3.823$--$3.865\times$ four-GPU speedup range.
At $L{=}2048$, AEGIS reports single-GPU memory infeasibility, so Table~\ref{tab:e2e_comparison} marks the one-GPU cell as OOM.
MPCFormer's~\cite{li2023mpcformer} released profiler was built and executed through $L{=}512$, but it times same-node passes on synthetic inputs. Table~\ref{tab:e2e_comparison} therefore anchors on its directly measured $L{=}128$ latency and scales the reproduced $L{=}256$ and $L{=}512$ costs by the $9.45\times$ ratio between that measurement and the profiler's $L{=}128$ cost normalized to the same network profile.
SIGMA~\cite{gupta2024sigma} is reproduced from its artifact and normalized to the paper network profile using its measured online traffic and offline FSS material, while SHAFT~\cite{kei2025shaft} is reproduced from its released short-context artifact and converted to the same network profile using its compute, communication, and round ledger.
BOLT~\cite{pang2024bolt} was built locally, but its two-party inference loop did not complete after initialization, so the paper-reported value is retained.
BumbleBee~\cite{lu2025bumblebee} completed its released $L{=}128$ configuration, and its artifact exposes no validated executable path beyond that length.
BLB~\cite{xu2025blb} was built and executed from its Zenodo release, which provides longer-input options, but the $L{=}256$ and $L{=}512$ runs exceed the memory threshold defined in the paper.
EncFormer~\cite{zhu2026encformer} runs its native encrypted path only through $L{=}128$.
We therefore anchor its key-switch counts on the reproduced artifact at that length and take longer sequences from its simulator ledger, rather than using EncFormer as a long-sequence latency baseline.
CipherPrune~\cite{zhang2025cipherprune} changes the securely executed sequence through token pruning and is therefore discussed as an orthogonal long-context strategy.
The two strategies are summarized in Figure~\ref{fig:related_strategies}, where token pruning shortens the securely executed sequence and multi-GPU placement distributes encrypted attention across devices, so both reduce cost without changing the quadratic scaling.

\begin{figure}[t]
\centering
\includegraphics[width=0.9\columnwidth,keepaspectratio,page=11,trim=630pt 441pt 374pt 232pt,clip]{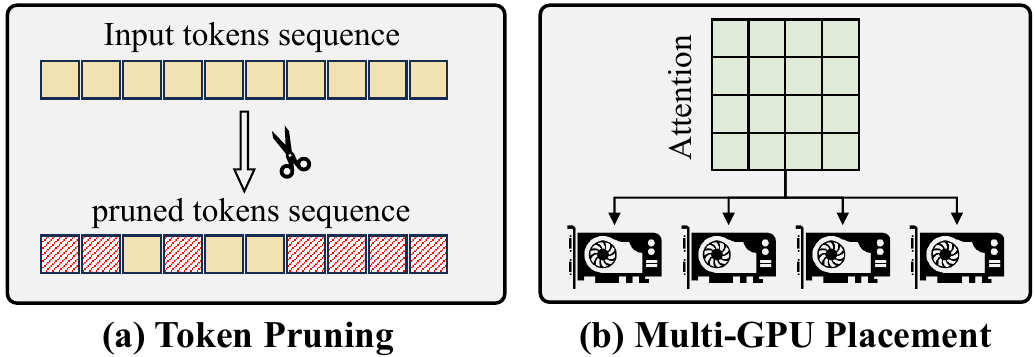}
\caption{Orthogonal long-context strategies, where (a) token pruning shortens the secure sequence and (b) multi-GPU placement distributes encrypted attention.}
\label{fig:related_strategies}
\end{figure}

\subsection{Hardware Memory Limits}
\label{app:oom-threshold}

We classify reproduced configurations against the physical memory limits of the evaluation hardware. The device reference is one A100-40GB, a widely deployed accelerator used in AWS \texttt{p4d.24xlarge} instances, and the reproduction host provides $256$\,GB of system memory~\cite{aws_ec2_p4}. Configurations exceeding $40$\,GB of live GPU allocation or $256$\,GB of peak resident host allocation are classified as OOM. This host limit is permissive relative to a \texttt{p4d.24xlarge}, which provides $1{,}152$\,GiB across eight A100s, or $144$\,GiB per GPU, making our limit approximately $1.78\times$ larger than the equal-share reference~\cite{aws_ec2_p4}. We reserve OOM for configurations that fail from memory exhaustion or whose measured live allocation exceeds the corresponding physical limit. 

\subsection{Baseline Key-Switch Accounting}
\label{app:baseline-ks-accounting}

Figure~\ref{fig:ks_scaling} and Table~\ref{tab:ks_full_layer} compare one-layer key-switch counts under a common accounting setting, with one relinearization for each ciphertext multiplication. The deployed long-context HEScan-kernel runtime rows instead use $n_{\mathrm{ckks}}{=}32{,}768$ complex slots. The HEInProj and HEOutProj kernels retain the $16{,}384$-slot configurations in Table~\ref{tab:ckks_config}.

\begin{table}[t]
\centering
\caption{One-layer key-switch counts at $16{,}384$ slots.}
\label{tab:ks_full_layer}
\scriptsize
\setlength{\tabcolsep}{5pt}
\renewcommand{\arraystretch}{1.08}
\begin{tabular}{@{}c|ccc|cc@{}}
\toprule
\rowcolor{tabheader}
& \multicolumn{3}{c|}{Key switches}
& \multicolumn{2}{c}{Ratio to \system{}} \\
\rowcolor{tabheader}
$L$ & \system{} & EncFormer & BLB & EncFormer & BLB \\
\midrule
$128$ & $13.9\mathrm{K}$ & $3.7\mathrm{K}$ & $4.9\mathrm{K}$ & $0.27\times$ & $0.35\times$ \\
$256$ & $38.8\mathrm{K}$ & $21.5\mathrm{K}$ & $22.5\mathrm{K}$ & $0.55\times$ & $0.58\times$ \\
$512$ & $79.2\mathrm{K}$ & $99.2\mathrm{K}$ & $136.0\mathrm{K}$ & $1.25\times$ & $1.72\times$ \\
$1{,}024$ & $160.0\mathrm{K}$ & $552.7\mathrm{K}$ & $948.3\mathrm{K}$ & $3.46\times$ & $5.93\times$ \\
$2{,}048$ & $321.5\mathrm{K}$ & $3.0\mathrm{M}$ & $7.1\mathrm{M}$ & $9.45\times$ & $21.97\times$ \\
\bottomrule
\end{tabular}
\end{table}

At short context lengths, the wider Mamba factor path gives \system{} a higher fixed operation cost, and the ordering reverses at $L{=}512$ as token-pair attention begins to dominate.
At $L{=}2{,}048$, EncFormer and BLB require $9.45\times$ and $21.97\times$ more key-switches, respectively, reflecting quadratic attention growth against the linear scan-contract path.

\subsection{MPCMamba Comparison and Protocol Cost}
\label{app:mpcmamba-protocols}
\label{app:protocol-cost-baseline}

MPCMamba~\cite{yu2025mpcmamba} is the only prior secure Mamba inference system, but it evaluates Vision Mamba under MPC on short image-patch sequences using six RTX~4090 GPUs across two servers.
\system{} instead evaluates encrypted SSM on long-document text using a hybrid CKKS--MPC backend at sequence lengths from $128$ to $8192$.
MPCMamba provides no public artifact, and its workload, protocol family, and hardware differ from \system{}.
We therefore use MPCMamba as a same-family protocol reference rather than a direct end-to-end latency baseline, with the evaluation scope summarized in Table~\ref{tab:mpcmamba-setting}.

\begin{table}[t]
\centering
\caption{Evaluation scope of MPCMamba and \system{}.}
\label{tab:mpcmamba-setting}
\scriptsize
\setlength{\tabcolsep}{4pt}
\renewcommand{\arraystretch}{1.12}
\begin{tabularx}{\columnwidth}{@{}l|X|X@{}}
\toprule
\rowcolor{tabheader}
Axis & MPCMamba~\cite{yu2025mpcmamba} & \system{} \\
\midrule
Workload & Vision classification & Long-document text classification \\
Model & Vision Mamba & Mamba-2 \\
Sequence regime & Image-patch sequences & $128$ to $8192$ text tokens \\
Private inference & MPC only & Hybrid CKKS--MPC \\
Nonlinear protocols & RMSNorm, SiLU, softplus, and exponential & NeRMS, MoSiLU, MoSoft, and MMExp \\
Backend & CrypTen with SEMI-2K & Phantom CKKS with SCI-style MPC \\
Hardware & Two servers with six RTX~4090 GPUs & One to four GPUs in the main evaluation \\
Reported metrics & Accuracy and latency & Accuracy, latency, communication, and memory \\
Artifact & Not publicly available & Paper artifact provided \\
\bottomrule
\end{tabularx}
\end{table}

We reconstruct MPCMamba's nonlinear protocol costs from the published algorithms and parameters. The comparison reports protocol-level operation counts under the shared cost model, complementing the end-to-end latency measurements. MPCMamba uses $n_{\exp}=8$ limit-approximation squarings and $t=3$ Newton or Householder iterations.
Its resulting per-output operation counts are
\[
\begin{aligned}
  C_{\mathrm{SiLU}}
  &= n_{\exp}+2t+4 = 18,
  \\
  C_{\mathrm{Softplus}}
  &= n_{\exp}+t(n_{\exp}+2)+n_{\exp} = 46,
  \\
  C_{\mathrm{RMSNorm}}
  &= 2+4t+n_{\exp} = 22,
  \\
  C_{\exp(\Delta A)}
  &= n_{\exp} = 8.
\end{aligned}
\]

\system{} uses the deployed D4, N1, and D4 configurations reported in Appendix~\ref{app:mpc-nonlinear-details}.
MoSiLU and MoSoft each require two products, two comparisons, and two muxes.
NeRMS requires two products and seven comparisons.
MMExp requires five products, one comparison, and one mux.

\section{Deployment Scaling}
\label{app:deployment-scaling}

This section validates \system{} deployment across runtime stages, GPU counts, model shapes, six accelerator classes, and sequence lengths through $L{=}8192$.

\subsection{End-to-End Runtime Breakdown}
\label{app:runtime-breakdown}

Table~\ref{tab:runtime-breakdown} decomposes the measured $L{=}2{,}048$ $12$-layer latency, where the left panel gives a per-layer stage breakdown on one A100-40GB GPU and the pie panel visualizes the same shares.
The scan compositions dominate at ${\sim}74\%$ of single-GPU per-layer time, with the remaining time split across projections, MPC nonlinear work, CKKS--MPC conversion, and LAN communication.
One-time offline weight encoding for the BSGS projections costs roughly $43$\,s per layer at this length, and is amortized across inferences and excluded from the per-layer online breakdown below.

\begin{table}[t]
\centering
\caption{Runtime breakdown for deployed \mambabase{} at $L{=}2{,}048$.}
\label{tab:runtime-breakdown}
\begin{minipage}[c]{0.58\columnwidth}
\centering
\scriptsize
\setlength{\tabcolsep}{3.5pt}
\renewcommand{\arraystretch}{1.1}
\begin{tabular}{@{}l|c@{}}
\toprule
\rowcolor{tabheader}
Stage, 1 GPU & Time (s/layer) \\
\midrule
MPC RMSNorm                              & $5.3$   \\
HEInProj kernel                          & $15.2$  \\
MPC factored activation                  & $9.8$   \\
HEScan kernel                            & $284.2$ \\
HEOutProj kernel                         & $17.4$  \\
CKKS--MPC conversion                      & $1.5$   \\
LAN transfer and protocol layers         & $52.9$  \\
\midrule
\textbf{Online Total, 1 GPU} & $\mathbf{386.3}$ \\

\bottomrule
\end{tabular}
\end{minipage}\hspace{0.01\columnwidth}
\begin{minipage}[c]{0.39\columnwidth}
\centering
\includegraphics[width=\linewidth]{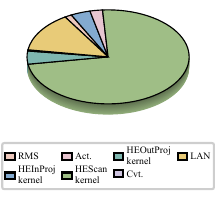}
\end{minipage}
\end{table}

\subsection{Multi-GPU HEScan Kernel Execution}
\label{app:multi-gpu-hescan}

Parallel prefix evaluation separates the associative recurrence operator from
the hardware partition that evaluates it~\cite{blelloch1990prefix,harris2007scan,gu2024mamba}.
After the compact FHE--MPC boundary, the HEScan kernel admits a direct device partition
because distinct state chunks cover disjoint expanded channels. \system{} assigns
complete state chunks to GPUs and keeps every sequence prefix, block summary,
carry, and contraction associated with a chunk on the same device. 

\runinhead{State-Chunk Partitioning.} Let the $K_s$ state chunks be those
defined in \S\ref{sec:ks-he:hescan-packing}. For $N_{\mathrm{gpu}}$ GPUs,
device $g$ receives the contiguous chunk set
\begin{equation}
  \mathcal{Q}_g
  =
  \left\{
    \kappa\ \middle|
    \left\lfloor
      \frac{gK_s}{N_{\mathrm{gpu}}}
    \right\rfloor
    \le
    \kappa
    <
    \left\lfloor
      \frac{(g+1)K_s}{N_{\mathrm{gpu}}}
    \right\rfloor
  \right\}.
  \label{eq:multi-gpu-chunk-partition}
\end{equation}
The sets $\mathcal{Q}_0,\ldots,\mathcal{Q}_{N_{\mathrm{gpu}}-1}$ are disjoint,
cover all state chunks, and differ in size by at most one chunk. Contiguous
assignment preserves the global expanded-channel order and simplifies output
assembly.

Each GPU receives a device-local copy of the common CKKS context, public
encryption key, evaluation keys, and compact factor ciphertexts, while only the
state-packing ciphertexts for chunks in $\mathcal{Q}_g$ are constructed on
device $g$.

\begin{algorithm}[t]
\caption{Multi-GPU HEScan kernel execution}
\label{alg:multi-gpu-hescan}
\small
\begin{algorithmic}[1]
\Require Encrypted factor streams $\langle\mathbf x_k\rangle,\langle\mathbf a_k\rangle,\langle\mathbf B_k\rangle,\langle\mathbf C_k\rangle$
\Require Chunks $K_s$, block size $B_{\mathrm{blk}}$, GPUs $N_{\mathrm{gpu}}$
\Ensure Contracted outputs $\langle\mathbf m_1\rangle,\ldots,\langle\mathbf m_L\rangle$
\State Partition chunks into $\mathcal{Q}_0,\ldots,\mathcal{Q}_{N_{\mathrm{gpu}}-1}$ by Eq.~\ref{eq:multi-gpu-chunk-partition}
\State Replicate the CKKS context, keys, and compact factor ciphertexts
\ForAll{GPUs $g=0,\ldots,N_{\mathrm{gpu}}-1$}
  \State Initialize shard output $\langle\mathbf M_g\rangle\gets0$
  \For{each state chunk $\kappa\in\mathcal{Q}_g$}
    \For{blocks $j=0,\ldots,K_{\mathrm{blk}}-1$}
      \State Build block maps $\{\mathcal{T}_{j,t}^{(\kappa)}\}_t$ with identity padding
      \State $(\{\mathcal{P}_{j,t}^{(\kappa)}\}_t,\mathcal{S}_{j}^{(\kappa)})
             \gets\mathsf{BKScan}(\{\mathcal{T}_{j,t}^{(\kappa)}\}_t)$
      \State Store $\mathcal{S}_{j}^{(\kappa)}$ and release block-local prefixes
    \EndFor
    \State $\{\mathcal{G}_{j}^{(\kappa)}\}_j
           \gets\mathsf{ExclusiveScan}(\{\mathcal{S}_{j}^{(\kappa)}\}_j)$
    \For{blocks $j=0,\ldots,K_{\mathrm{blk}}-1$}
      \State Replay block $j$ to rebuild $\{\mathcal{P}_{j,t}^{(\kappa)}\}_{t<B_j}$
      \State $\{\widehat{\mathcal{P}}_{j,t}^{(\kappa)}\}_{t<B_j}
             \gets\{\mathcal{P}_{j,t}^{(\kappa)}
             \circ\mathcal{G}_{j}^{(\kappa)}\}_{t<B_j}$
      \State $\langle\mathbf M_g\rangle\gets
             \langle\mathbf M_g\rangle+
             \mathsf{Contract}_{d_s}(\widehat{\mathcal{P}}_{j}^{(\kappa)},
             \mathrm{br}^{(\kappa)}(\langle\mathbf C_j\rangle))$
    \EndFor
  \EndFor
\EndFor
\State \Return $\sum_{g=0}^{N_{\mathrm{gpu}}-1}\langle\mathbf M_g\rangle$
\end{algorithmic}
\end{algorithm}

\runinhead{Cross-Device Communication.} Each state chunk stays on one GPU, so
all affine composition, carry propagation, and contraction are device local.
Multi-GPU execution replicates the compact factor ciphertexts once and collects
the contracted shard outputs at the end, with no encrypted collective during
the prefix scan. Since shards occupy disjoint expanded-channel coordinates, their
ciphertext sum reconstructs the single-GPU packed output, or they can be added
as MPC shares before the post-SSM stage. This chunk partition avoids the
token-pair matrix split and prefix-node communication used by multi-GPU
encrypted attention placement~\cite{gong2026aegis}.

\runinhead{Correctness.} For a fixed chunk $\kappa$,
Algorithm~\ref{alg:multi-gpu-hescan} executes the same two-pass resident
schedule defined in Eq.~\ref{eq:resident-block-schedule}. The first pass
computes exact block summaries, the exclusive summary scan computes the exact
incoming carry for every block, and the replay pass composes that carry with
each local prefix before contraction. Distributing the outer loop over
$\kappa$ therefore changes only the device that evaluates each independent
chunk.

Each
device consequently produces a disjoint subset of the terms in
Eq.~\ref{eq:hescan-contract-output}, and summing the shard outputs gives the same
contracted ciphertext as the single-GPU HEScan kernel path. The block size and GPU
count affect work placement and peak residency but not the recurrence being
evaluated.

\runinhead{Scaling Analysis.} Let $f$ denote the fraction of one-GPU
execution time spent in the HEScan kernel. When the non-scan stages remain unsharded and
scan work is balanced across $N_{\mathrm{gpu}}$ devices, the idealized
execution time and speedup are
\begin{equation}
  T_{N_{\mathrm{gpu}}}^{\mathrm{ideal}}
  =
  T_{\mathrm{other}}
  +
  \frac{T_{\mathrm{scan}}}{N_{\mathrm{gpu}}},
  \qquad
  S_{N_{\mathrm{gpu}}}^{\mathrm{ideal}}
  =
  \frac{1}{1-f+f/N_{\mathrm{gpu}}}.
  \label{eq:hescan-amdahl}
\end{equation}
At $L{=}2048$, the HEScan kernel accounts for $73.6\%$ of one-GPU per-layer latency in Table~\ref{tab:runtime-breakdown}, and four-GPU execution achieves a measured $2.15\times$ speedup, closely matching the $2.23\times$ Amdahl prediction from Eq.~\ref{eq:hescan-amdahl}.

\subsection{Model-Shape Scaling}
\label{app:dimension-robustness}

\system{} supports model shapes beyond the deployed Mamba-base configuration, and we validate this breadth by sweeping hidden width $d_{\mathrm{model}}$, head count $H$, and state dimension $d_s$ while retaining the same PhantomFHE CKKS backend, SCI/EzPC-style MPC backend, and blocked-scan execution path.
With active state-packing capacity $s_{\mathrm{state}}{=}16{,}384$, each shape induces the encrypted state-chunk count
\[
  K_s
  =
  \left\lceil
    \frac{HPd_s}{s_{\mathrm{state}}}
  \right\rceil .
\]
\phantomsection
\label{app:ks-scaling-model}
The scan key-switch count follows directly from the Brent--Kung network, where
each node uses two composition key-switches per state chunk, giving
\[
  2K_s\mathrm{BK}(L)
\]
composition key-switches, where
\[
  \mathrm{BK}(L)
  =
  2(L'-1)-\log_2 L',
  \qquad
  L'
  =
  2^{\lceil\log_2 L\rceil}.
\]
Latency therefore scales with the number of resident state chunks, while larger active packings trade larger ciphertexts for fewer chunks, and Table~\ref{tab:shape_scaling_backend} reports both the concrete configurations and an independent shape-sweep latency measurement over all evaluated sequence lengths.

\begin{table}[t]
\centering
\caption{Shape sweep and measured latency across \system{} configurations.}
\label{tab:shape_scaling_backend}
\begin{minipage}[c]{0.5\columnwidth}
\centering
\scriptsize
\setlength{\tabcolsep}{1.5pt}
\renewcommand{\arraystretch}{1.08}
\begin{tabular}{@{}l|cccc@{}}
\toprule
\rowcolor{tabheader}
Config & $d_{\mathrm{model}}$ & $H$ & $d_s$ & Layers \\
\midrule
\system{}-small  & $384$  & $12$ & $64$  & $12$ \\
\system{}-slim   & $768$  & $24$ & $64$  & $12$ \\
\cellcolor{oursrow}\system{}-base & \cellcolor{oursrow}$768$ & \cellcolor{oursrow}$24$ & \cellcolor{oursrow}$128$ & \cellcolor{oursrow}$12$ \\
\system{}-wide   & $1024$ & $32$ & $128$ & $12$ \\
\system{}-state+ & $768$  & $24$ & $256$ & $12$ \\
\system{}-130M   & $768$  & $24$ & $128$ & $24$ \\
\system{}-370M   & $1024$ & $32$ & $128$ & $48$ \\
\bottomrule
\end{tabular}
\end{minipage}\hspace{0.02\columnwidth}
\begin{minipage}[c]{0.46\columnwidth}
\centering
\includegraphics[width=\linewidth]{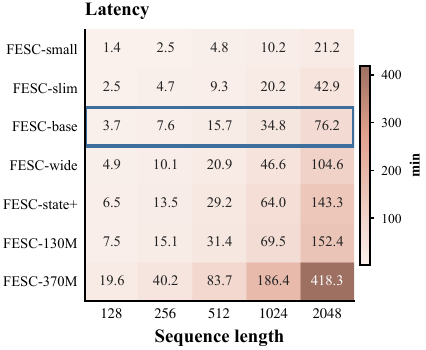}
\end{minipage}
\end{table}

\subsection{Longer-Sequence Deployment}
\label{app:longseq}

The main evaluation establishes \system{} through $L{=}2048$, and this section pushes the system to $L{=}8192$ while preserving the factorized boundary, block-local Brent--Kung scan, and resident execution schedule within the defined OOM budget.

These configurations serialize the block carries instead of scanning the
summaries in parallel, so for resident block size $B_{\mathrm{blk}}$ and
$K_{\mathrm{blk}}=\lceil L/B_{\mathrm{blk}}\rceil$ the scan consumes
\[
  D_{\mathrm{scan}}
  =
  2\left\lceil\log_2 B_{\mathrm{blk}}\right\rceil
  +
  K_{\mathrm{blk}}
\]
multiplicative levels.
The first term is the block-local Brent--Kung depth.
The second term is the carry depth across blocks.
At $L{=}2048$ and $B_{\mathrm{blk}}{=}1024$, this expression gives
$22$ levels and matches the $1000$-bit deployed chain in
Table~\ref{tab:ckks-level-audit}.
The longer-sequence configurations use $N{=}65{,}536$, which provides
$41$ levels within a $1782$-bit modulus budget.

Feasibility requires both the level count and ciphertext residency to remain
within the device limits.
At $L{=}4096$, $B_{\mathrm{blk}}{=}512$ requires $26$ levels and reaches
a measured peak of $37.8$\,GB.
The configuration completes a measured $12$-layer run without bootstrapping at serialized-carry depth $26$, with $748.9$\,s per layer.

At $L{=}8192$, the unrefreshed depth and memory constraints no longer overlap.
With $B_{\mathrm{blk}}{=}512$, the scan requires only $34$ levels, but the
measured $40.4$\,GB footprint exceeds the usable allocator envelope.
Reducing the block to $B_{\mathrm{blk}}{=}256$ lowers the peak to
$27.6$\,GB, but increases the requirement to $48$ levels, beyond the
available chain.
Since residency increases with block size while carry depth decreases, these
two configurations bracket the unrefreshed feasibility boundary.

A secure carry refresh every $R$ blocks bounds the scan depth by
\[
  D_{\mathrm{refresh}}
  =
  2\left\lceil\log_2 B_{\mathrm{blk}}\right\rceil
  +
  R .
\]
With $B_{\mathrm{blk}}{=}256$ and $R{=}8$, the HEScan kernel requires $24$ levels
and $1160$ chain bits.
The measured peak falls to $28.6$\,GB, and the complete $L{=}8192$ scan runs
in $1748.9$\,s per layer.
The refresh converts the carry to shares with $\Pi_{\mathrm{C2M}}^{\mathbb{C}}$
and back with $\Pi_{\mathrm{M2C}}^{\mathbb{C}}$, returning it at the top of the
modulus chain without bootstrapping and without revealing it to either party.
The resulting execution remains bootstrapping-free, and Table~\ref{tab:longseq} collects the depth, residency, and latency of every longer-sequence configuration.

\begin{table}[t]
\centering
\caption{Longer-sequence execution on one A100-40GB GPU.}
\label{tab:longseq}
\scriptsize
\setlength{\tabcolsep}{2.4pt}
\renewcommand{\arraystretch}{1.1}
\newcolumntype{Y}{>{\centering\arraybackslash}X}
\newcommand{\lsz}{\phantom{0}}
\newcommand{\lsspan}[1]{%
  \multicolumn{3}{c@{}}{#1}}
% The last three columns are equal-width X columns. With plain c columns the
% wide \lsspan rows push all of their overhang into the final column, which
% leaves "12L min" far wider than "s/layer" and "GB".
\begin{tabularx}{\columnwidth}{@{}r|r|c|c|YYY@{}}
\toprule
\rowcolor{tabheader}
$L$ & $B_{\mathrm{blk}}$ & Ref. & $D$ & s/layer & GB & $12$L min \\
\midrule
$2048$ & $1024$ & none
& $22$ & \lsz$386.3$ & $32.7$ & \lsz$77.3$ \\
$4096$ & $512$ & none
& $26$ & \lsz$748.9$ & $37.8$ & $149.8$ \\
\midrule
$8192$ & $512$ & none
& $34$ & \lsspan{allocator failure at a measured peak of $40.4$\,GB} \\
$8192$ & $256$ & none
& $48$ & \lsspan{exceeds the available $41$ levels} \\
$8192$ & $256$ & every $8$ blocks
& $24$ & $1748.9$ & $28.6$ & $349.8$ \\
\bottomrule
\end{tabularx}
\let\lsz\relax
\let\lsspan\relax
\end{table}

\subsection{Cross-GPU Portability}
\label{app:gpu-memory-envelope}

Figure~\ref{fig:gpu_envelope} tests whether the resident HEScan schedule preserves long-context deployment across GPU generations and memory capacities by executing the same native backend build on T4, Titan, A100-40GB, A100-80GB, H100, and H200 devices.
\system{} reaches $L{=}2048$ on every tested class, demonstrating portable execution reach while latency tracks device capability.

\begin{figure}[t]
\centering
\includegraphics[width=\columnwidth]{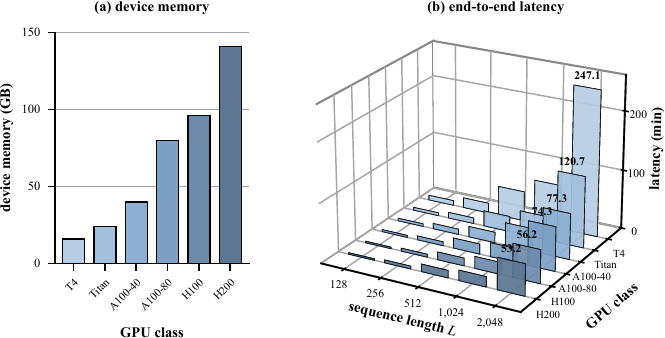}
\caption{\system{} execution across six GPU classes. Panel a reports device memory, and Panel b reports end-to-end latency from $L{=}128$ to $L{=}2{,}048$.}
\label{fig:gpu_envelope}
\end{figure}

\section{Cross-Family SSM Implementations}
\label{app:ssm-implementations}
\label{app:cross-ssm-family}

This appendix shows how invariant and selective SSM families map to the same scan-contract interface exposed by the HEScan kernel in the deployed Mamba-2 pipeline.

\subsection{Common Scan-Contract Interface}

Let $r$ denote an SSM family and let $e\in\{0,\ldots,E^{(r)}-1\}$ index its expanded scan-channel dimension.
For the grouped Mamba-2 family, $e=(h,p)$ and $E^{(r)}=HP$.
Let $i\in\{0,\ldots,d_s^{(r)}-1\}$ index the SSM state coordinate, and let $g_r(e)$ denote the group assigned to channel $e$ by family $r$.
Each family adapter first emits a compact factor packet
\[
  \mathsf{Pack}^{(r)}_k
  =
  \bigl(
    x^{(r)}_k,\;
    a^{(r)}_k,\;
    B^{(r)}_k,\;
    C^{(r)}_k,\;
    \ldots
  \bigr),
  \qquad
  k=1,\ldots,L .
\]
The packet contains recurrence factors rather than a materialized dense transition, and the family adapter maps it to the affine packet consumed by HEScan,
\[
  \mathsf{Aff}^{(r)}_k
  =
  \bigl(
    a^{(r)}_k,\;
    s^{(r)}_k,\;
    \widehat{C}^{(r)}_k
  \bigr).
\]
\system{} interprets the family-specific decay through a broadcast map $\operatorname{br}^{(r)}_a$,
\[
  \alpha^{(r)}_k[e,i]
  =
  \operatorname{br}^{(r)}_a
    \bigl(a^{(r)}_k\bigr)[e,i],
\]
and scans the additive state update already represented by $s^{(r)}_k$,
\[
  h^{(r)}_k[e,i]
  =
  \alpha^{(r)}_k[e,i]\,
  h^{(r)}_{k-1}[e,i]
  +
  s^{(r)}_k[e,i],
  \qquad
  h^{(r)}_0=0.
\]
The contracted scan output is
\[
  m^{(r)}_k[e]
  =
  \sum_{i=0}^{d_s^{(r)}-1}
    h^{(r)}_k[e,i]\,
    \widehat{C}^{(r)}_k[e,i].
\]
For S4D, S5, LRU, Mamba-1/S6, and Mamba-2, the affine packet is the one-endpoint implementation
\[
  \begin{aligned}
  s^{(r)}_k[e,i]
  &=
  x^{(r)}_k[e]\,
  B^{(r)}_k[g_r(e),i],\\
  \widehat{C}^{(r)}_k[e,i]
  &=
  C^{(r)}_k[g_r(e),i].
  \end{aligned}
\]
Mamba-3 uses the same affine packet after its exponential--trapezoidal factor builder forms a two-endpoint state input.
Thus the scan-contract interface is shared, while the construction of $(a_k,s_k,\widehat{C}_k)$ remains family-specific, as summarized in Table~\ref{tab:ssm-family-packet-summary}.

\begin{table*}[t]
\centering
\caption{Scan-packet factors by family at $L{=}128$ for $d_{\mathrm{model}}{=}768$, $H{=}24$, $P{=}64$, $E{=}1536$.}
\label{tab:ssm-family-packet-summary}
\label{tab:cross_ssm_family}
\scriptsize
\setlength{\tabcolsep}{6pt}
\renewcommand{\arraystretch}{1.12}
\begin{tabular}{@{}l|l|lll|rrr@{}}
\toprule
\rowcolor{tabheader}
Family & Regime & Decay & State input & Contraction & $d_s$ & $K_s$ & Affine relin. \\
\midrule
S4D~\cite{gu2022s4d} & invariant & Fixed diagonal & $x_k\bar B$ & Fixed $\bar C$ & $64$ & $6$ & $0$ \\
S5~\cite{smith2023s5} & invariant & Fixed shared transition & $x_k\bar B$ & Fixed $\bar C$ & $128$ & $12$ & $0$ \\
LRU~\cite{orvieto2023lru} & invariant & Fixed complex recurrence & $x_k\bar B$ & Fixed $\bar C$ & $128$ & $12$ & $0$ \\
\midrule
Mamba-1/S6~\cite{gu2024mamba} & selective & Token-dependent $\exp(\Delta_kA)$ & $x_kB_k$ & Token-dependent $C_k$ & $16$ & $2$ & $988$ \\
\rowcolor{oursrow}
Mamba-2~\cite{dao2024mamba2} & selective & Headwise token-dependent & $x_kB_k$ & Grouped dynamic $C_k$ & $128$ & $12$ & $5{,}928$ \\
Mamba-3~\cite{lahoti2026mamba3} & selective & Trapezoidal complex & $\beta_kv_{k-1}+\gamma_kv_k$ & Dynamic or RoPE-adjusted $C_k$ & --- & --- & --- \\
\bottomrule
\end{tabular}
\end{table*}

\subsection{Invariant Families}

\runinhead{S4D.}
S4D instantiates the invariant packet with a diagonal state-transition spectrum and static input/output factors
\[
  \begin{aligned}
  a^{\mathrm{S4D}}_k[i]=\bar a^{\mathrm{S4D}}[i],
  \quad
  B^{\mathrm{S4D}}_k[g,i]=\bar B^{\mathrm{S4D}}[g,i],\\
  C^{\mathrm{S4D}}_k[g,i]=\bar C^{\mathrm{S4D}}[g,i].
  \end{aligned}
\]
Only the scan input varies with the token
\[
  x^{\mathrm{S4D}}_k[e]
  =
  x^{\mathrm{scan}}_k[e].
\]
The resulting \system{} state update is
\[
  h^{\mathrm{S4D}}_k[e,i]
  =
  \bar a^{\mathrm{S4D}}[i]\,
  h^{\mathrm{S4D}}_{k-1}[e,i]
  +
  x^{\mathrm{scan}}_k[e]\,
  \bar B^{\mathrm{S4D}}[g_r(e),i],
\]
with contracted output
\[
  m^{\mathrm{S4D}}_k[e]
  =
  \sum_i
    h^{\mathrm{S4D}}_k[e,i]\,
    \bar C^{\mathrm{S4D}}[g_r(e),i].
\]
This is the diagonal invariant-decay case.

\runinhead{S5.}
The S5-style artifact row uses the same invariant scan law as S4D, but with richer static MIMO-style $B,C$ factors.
In the diagonal-basis view used by the \system{} adapter, the recurrence factors are
\[
  \begin{aligned}
  a^{\mathrm{S5}}_k[i]=\bar a^{\mathrm{S5}}[i],
  \quad
  B^{\mathrm{S5}}_k[g,i]=\bar B^{\mathrm{S5}}[g,i],\\
  C^{\mathrm{S5}}_k[g,i]=\bar C^{\mathrm{S5}}[g,i],
  \end{aligned}
\]
and
\[
  x^{\mathrm{S5}}_k[e]
  =
  x^{\mathrm{scan}}_k[e].
\]
The update and contraction are
\[
  h^{\mathrm{S5}}_k[e,i]
  =
  \bar a^{\mathrm{S5}}[i]\,
  h^{\mathrm{S5}}_{k-1}[e,i]
  +
  x^{\mathrm{scan}}_k[e]\,
  \bar B^{\mathrm{S5}}[g_r(e),i],
\]
and
\[
  m^{\mathrm{S5}}_k[e]
  =
  \sum_i
    h^{\mathrm{S5}}_k[e,i]\,
    \bar C^{\mathrm{S5}}[g_r(e),i].
\]
It should be read as an invariant scan row with richer static $B,C$ mixing, not as a general dense-transition benchmark.

\runinhead{LRU.}
LRU is also an invariant affine recurrence.
In complex form, let
\[
  \bar\lambda^{\mathrm{LRU}}[i]
  =
  \rho_i e^{\mathrm{i}\theta_i},
  \qquad
  |\rho_i|<1,
\]
or use the equivalent real block form for complex-conjugate pairs.
The artifact uses the same invariant factorization
\[
  \begin{aligned}
  a^{\mathrm{LRU}}_k[i]=\bar\lambda^{\mathrm{LRU}}[i],
  \quad
  B^{\mathrm{LRU}}_k[g,i]=\bar B^{\mathrm{LRU}}[g,i],\\
  C^{\mathrm{LRU}}_k[g,i]=\bar C^{\mathrm{LRU}}[g,i],
  \end{aligned}
\]
while
\[
  x^{\mathrm{LRU}}_k[e]
  =
  x^{\mathrm{scan}}_k[e].
\]
Thus
\[
  h^{\mathrm{LRU}}_k[e,i]
  =
  \bar\lambda^{\mathrm{LRU}}[i]\,
  h^{\mathrm{LRU}}_{k-1}[e,i]
  +
  x^{\mathrm{scan}}_k[e]\,
  \bar B^{\mathrm{LRU}}[g_r(e),i],
\]
and
\[
  m^{\mathrm{LRU}}_k[e]
  =
  \sum_i
    h^{\mathrm{LRU}}_k[e,i]\,
    \bar C^{\mathrm{LRU}}[g_r(e),i].
\]
LRU differs from S4D/S5 in its lightweight linear-recurrent parameterization, but cryptographically it remains in the invariant-factor regime.

\subsection{Selective Families}

\runinhead{Mamba-1/S6.}
Mamba-1/S6 instantiates the same packet with token-dependent selective factors.
Let $e$ denote the selective-scan channel index.
The timestep and decay are
\[
  \Delta^{\mathrm{S6}}_k[e]
  =
  \operatorname{softplus}
  \bigl(
    d^{\mathrm{S6}}_k[e]+b^{\mathrm{S6}}_{\Delta}[e]
  \bigr),
\]
and
\[
  a^{\mathrm{S6}}_k[e,i]
  =
  \exp
  \bigl(
    \Delta^{\mathrm{S6}}_k[e]\,
    A^{\mathrm{S6}}[e,i]
  \bigr).
\]
The scan input is step-scaled
\[
  x^{\mathrm{S6}}_k[e]
  =
  \Delta^{\mathrm{S6}}_k[e]\,
  x^{\mathrm{raw,S6}}_k[e].
\]
The dynamic input and output factors are produced by the family-specific factor projection and represented through the same group map $g_{\mathrm{S6}}$
\[
  B^{\mathrm{S6}}_k[g_{\mathrm{S6}}(e),i],
  \qquad
  C^{\mathrm{S6}}_k[g_{\mathrm{S6}}(e),i].
\]
The \system{} recurrence is
\[
  h^{\mathrm{S6}}_k[e,i]
  =
  a^{\mathrm{S6}}_k[e,i]\,
  h^{\mathrm{S6}}_{k-1}[e,i]
  +
  x^{\mathrm{S6}}_k[e]\,
  B^{\mathrm{S6}}_k[g_{\mathrm{S6}}(e),i],
\]
with contraction
\[
  m^{\mathrm{S6}}_k[e]
  =
  \sum_i
    h^{\mathrm{S6}}_k[e,i]\,
    C^{\mathrm{S6}}_k[g_{\mathrm{S6}}(e),i].
\]
Unlike the invariant families, $a_k$, $B_k$, and $C_k$ are private token-dependent factors.
Unlike Mamba-2, whose decay is compact over heads and broadcast over $p,i$, the S6 row allows decay to vary over the scan channel and state coordinate. The common $\operatorname{br}^{(r)}_a$ map covers both decay structures.

\runinhead{Mamba-2.}
Mamba-2 uses the grouped selective scan form.
Let $e=(h,p)$, where $h$ is the SSM head and $p$ is the per-head channel.
Let $g(h)\in\{1,\ldots,G\}$ be the group assigned to head $h$.
The timestep, step-scaled input, and decay are
\[
  \Delta^{\mathrm{M2}}_k[h]
  =
  \operatorname{softplus}
  \bigl(
    d^{\mathrm{M2}}_k[h]+b^{\mathrm{M2}}_{\Delta}[h]
  \bigr),
\]
\[
  x^{\mathrm{M2}}_k[h,p]
  =
  \Delta^{\mathrm{M2}}_k[h]\,
  x^{\mathrm{raw,M2}}_k[h,p],
\]
and
\[
  a^{\mathrm{M2}}_k[h]
  =
  \exp
  \bigl(
    \Delta^{\mathrm{M2}}_k[h]\,
    A^{\mathrm{M2}}[h]
  \bigr).
\]
The grouped factors are
\[
  B^{\mathrm{M2}}_k[g,i],
  \qquad
  C^{\mathrm{M2}}_k[g,i],
  \qquad
  g=1,\ldots,G .
\]
The broadcast decay is $\alpha^{\mathrm{M2}}_k[h,p,i]=a^{\mathrm{M2}}_k[h]$, so the recurrence is
\[
  h^{\mathrm{M2}}_k[h,p,i]
  =
  a^{\mathrm{M2}}_k[h]\,
  h^{\mathrm{M2}}_{k-1}[h,p,i]
  +
  x^{\mathrm{M2}}_k[h,p]\,
  B^{\mathrm{M2}}_k[g(h),i],
\]
and the contracted scan output is
\[
  m^{\mathrm{M2}}_k[h,p]
  =
  \sum_i
    h^{\mathrm{M2}}_k[h,p,i]\,
    C^{\mathrm{M2}}_k[g(h),i].
\]
Mamba-2 differs from Mamba-1/S6 by grouped SSD-style $B_k,C_k$ sharing while retaining private token-dependent selective decay.

\runinhead{Mamba-3.}
Mamba-3 remains affine and prefix-computable, but its factor builder no longer emits the one-endpoint state input $x_kB_k$~\cite{lahoti2026mamba3,mamba3code}.
Its private builder produces a positive timestep $\Delta^{\mathrm{M3}}_k$, a sigmoid gate $\lambda^{\mathrm{M3}}_k$, and a negative transition parameter $A^{\mathrm{M3}}_k$, from which the exponential--trapezoidal coefficients follow,
\[
  \begin{aligned}
  \rho^{\mathrm{M3}}_k
  &=
  \exp\!\left(\Delta^{\mathrm{M3}}_kA^{\mathrm{M3}}_k\right),
  &\qquad
  \gamma^{\mathrm{M3}}_k
  &=
  \lambda^{\mathrm{M3}}_k\Delta^{\mathrm{M3}}_k,\\
  \beta^{\mathrm{M3}}_k
  &=
  \bigl(1-\lambda^{\mathrm{M3}}_k\bigr)\Delta^{\mathrm{M3}}_k\rho^{\mathrm{M3}}_k .
  \end{aligned}
\]
Writing $v^{\mathrm{M3}}_k[h,p,j]=B^{\mathrm{M3}}_k[h,j]\,x^{\mathrm{M3}}_k[h,p]$, the affine state input combines two endpoints,
\[
  s^{\mathrm{M3}}_k
  =
  \beta^{\mathrm{M3}}_k v^{\mathrm{M3}}_{k-1}
  +
  \gamma^{\mathrm{M3}}_k v^{\mathrm{M3}}_k ,
\]
so HEScan runs the same affine recurrence under multiplier $\rho^{\mathrm{M3}}_k$ and contracts with $C^{\mathrm{M3}}_k$, and the extra secure work falls on constructing $s_k$ rather than on the scan algebra.

The complex transition admits two CKKS representations, a paired-real rotation and the direct complex multiplier $q^{\mathrm{M3}}_k[h,j]=\rho^{\mathrm{M3}}_k[h]\bigl(\cos\phi^{\mathrm{M3}}_k[h,j]+\mathrm{i}\sin\phi^{\mathrm{M3}}_k[h,j]\bigr)$.
Both use the same secure factors and differ only in how the transition is packed, and Appendix~\ref{app:cross-family-results} selects the faster direct-complex path.

The MIMO variant raises the state-input and state-output rank from one to $R$, summing $R$ rank-indexed copies of the recurrence and contracting with $R$ output factors.
These factors stay private in \system{}, so the rank-$R$ products become additional ciphertext or MPC operations rather than ordinary low-precision matrix multiplications, and Table~\ref{tab:mamba3_implementation_l128} reports the resulting native secure implementation at $L{=}128$.

\subsection{Shared HEScan Kernel Backend}

The family adapters above differ in how they construct the affine packet
\[
  \mathsf{Aff}^{(r)}_k
  =
  \bigl(
    a^{(r)}_k,\;
    s^{(r)}_k,\;
    \widehat{C}^{(r)}_k
  \bigr).
\]
For S4D, S5, and LRU, the recurrence factors are invariant server-held plaintexts.
Mamba-1/S6 and Mamba-2 produce private one-endpoint selective factors, while Mamba-3 produces the two-endpoint affine input $s_k$ described in the preceding subsection.
Once this packet is available, all families use the same resident prefix schedule and state-contraction backend.

\runinhead{Affine Prefix Composition.}
For state chunk $\kappa$, token $k$ is represented by
\[
  \mathcal{T}^{(r,\kappa)}_k
  =
  \bigl(
    a^{(r)}_k,\;
    s^{(r,\kappa)}_k
  \bigr).
\]
For adjacent scan items $i<j$, HEScan applies
\[
  \mathcal{T}^{(r,\kappa)}_i
  \bullet
  \mathcal{T}^{(r,\kappa)}_j
  =
  \bigl(
    a^{(r)}_j a^{(r)}_i,\;
    \operatorname{br}^{(r,\kappa)}_a(a^{(r)}_j)
    \odot
    s^{(r,\kappa)}_i
    +
    s^{(r,\kappa)}_j
  \bigr).
\]
The Brent--Kung resident schedule computes
\[
  \mathcal{P}^{(r,\kappa)}_k
  =
  \mathcal{T}^{(r,\kappa)}_1
  \bullet\cdots\bullet
  \mathcal{T}^{(r,\kappa)}_k,
  \qquad
  h^{(r,\kappa)}_k
  =
  \mathcal{P}^{(r,\kappa)}_k(0),
\]
and contracts each prefix before releasing the chunk workspace,
\[
  \langle m^{(r)}_k\rangle
  =
  \sum_{\kappa}
  \Sigma_{d_s^{(r)}}\!\left(
    \langle h^{(r,\kappa)}_k\rangle
    \odot
    \widehat{C}^{(r,\kappa)}_k
  \right).
\]

\runinhead{Encrypted Operand Regimes.}
The common backend receives either
\[
  \bigl(
    \bar a^{(r)},\;
    \langle s^{(r)}_k\rangle,\;
    \bar C^{(r)}
  \bigr)
\]
for invariant families, or
\[
  \bigl(
    \langle a^{(r)}_k\rangle,\;
    \langle s^{(r)}_k\rangle,\;
    \langle\widehat{C}^{(r)}_k\rangle
  \bigr)
\]
for selective families.
Invariant decay therefore introduces no ciphertext relinearization during affine composition.
Selective families pay a ciphertext composition cost according to the size and structure of their private transition factors.

\subsection{Cross-Family HEScan Cost and Latency}
\label{app:cross-family-results}

Figure~\ref{fig:cross_ssm_family} summarizes how invariant and selective SSM adapters feed the common scan-contract backend. Their costs differ because the representations of $a_k$, $s_k$, and $\widehat{C}_k$ differ, with invariant factors remaining plaintext, Mamba-2 retaining compact headwise decay, Mamba-1/S6 using a state-expanded private transition, and Mamba-3 requiring a two-endpoint secure factor builder. Table~\ref{tab:cross_ssm_family} reports the canonical state shape and affine-composition cost of the one-endpoint families at $L{=}128$.
Table~\ref{tab:cross_ssm_family_e2e_latency} then fixes $d_{\mathrm{model}}{=}768$, $H{=}24$, $P{=}64$, and $d_s{=}128$, so that $E{=}1536$, $K_s{=}12$, slot packing, residency, contraction, and output packing are held constant.
The resulting differences therefore reflect each family's factor representation and encrypted affine composition rather than model dimensions.

The invariant families have nearly identical latency because their recurrence factors remain plaintext and only the encrypted scan input varies across tokens.
Mamba-1/S6 is substantially slower because its private transition expands across both channel and state coordinates, which motivates the grouped Mamba-2 factorization used by the \system{} pipeline.

\begin{figure*}[t]
\centering
\includegraphics[
  width=\textwidth,
  keepaspectratio,
  page=10,
  trim=42pt 234pt 129pt 206pt,
  clip
]{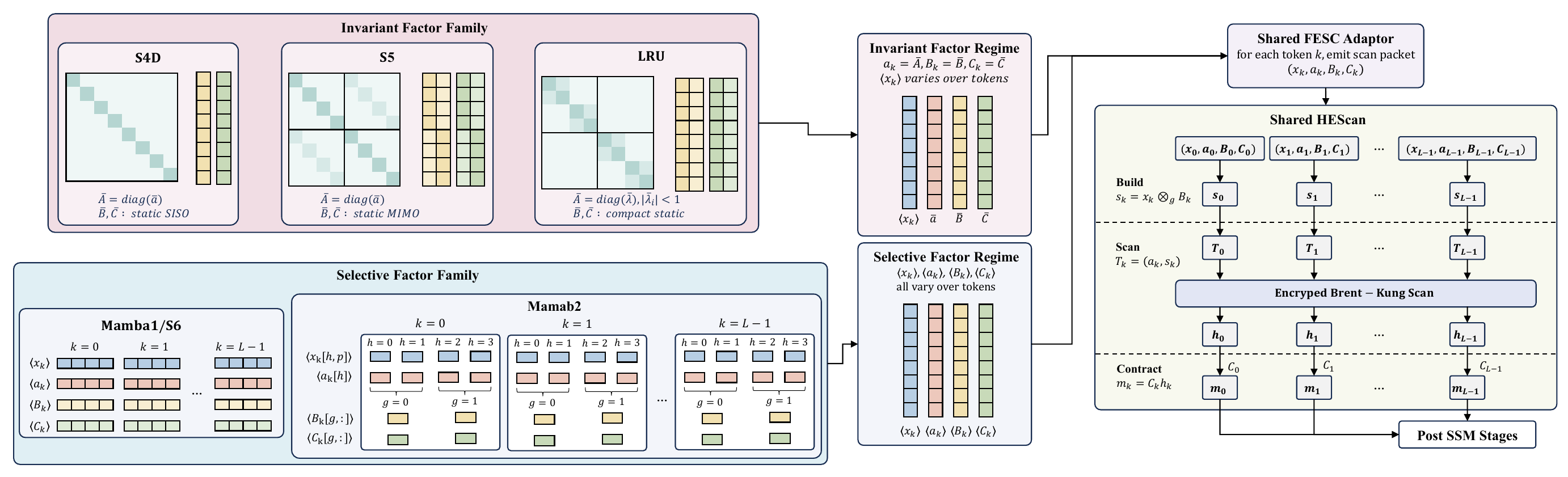}
\caption{Invariant and selective SSM families use the same \system{} scan-contract backend.}
\label{fig:cross_ssm_family}
\end{figure*}

\begin{table}[t]
\centering
\caption{Matched-shape latency for SSM families, in seconds.}
\label{tab:cross_ssm_family_e2e_latency}
\scriptsize
\setlength{\tabcolsep}{3.5pt}
\begin{tabular}{l|ccccc}
\toprule
\rowcolor{tabheader}
Family & $L{=}128$ & $L{=}256$ & $L{=}512$ & $L{=}1{,}024$ & $L{=}2{,}048$ \\
\midrule
S4D          & $6.7$   & $12.2$    & $26.9$      & $65.2$      & $132.6$ \\
S5           & $6.6$   & $12.0$    & $26.5$      & $65.1$      & $132.6$ \\
LRU          & $6.5$   & $11.9$    & $26.5$      & $64.8$      & $132.5$ \\
\midrule
Mamba-1/S6   & $722.6$ & $1{,}430.6$ & $2{,}865.5$ & $5{,}771.6$ & $12{,}623.6$ \\
\rowcolor{oursrow}
Mamba-2      & $18.5$  & $36.9$    & $77.4$      & $173.1$     & $381.5$ \\
\bottomrule
\end{tabular}
\end{table}

Mamba-3 is reported in Table~\ref{tab:mamba3_implementation_l128} because its native secure-layer scope includes the exponential--trapezoidal factor builder and, for MIMO, rank-$R$ state-input and output factors. For Mamba-3, we implement both paired-real and direct-complex transition representations and select the faster direct-complex path for the native secure-layer evaluation.
The table compares Mamba-2 and Mamba-3 under the same single-layer component harness, while Table~\ref{tab:e2e_comparison} reports deployed end-to-end model latency.
At $d_s{=}128$, Table~\ref{tab:mamba3_implementation_l128} reports $179.1$\,s for SISO and $405.8$\,s for MIMO-R4, compared with $48.3$\,s for the matched Mamba-2 control.
Profiling attributes the additional cost primarily to secure two-endpoint and rank-$R$ factor construction rather than to the resident prefix schedule.
The result therefore demonstrates compatibility with the shared HEScan backend while confirming that compact Mamba-2 remains the faster encrypted deployment at this scale.
Reducing Mamba-3 factor-construction overhead is left to future work.

\begin{table}[t]
\centering
\caption{Native single-layer latency for the Mamba-3 compatibility study at $L{=}128$.}
\label{tab:mamba3_implementation_l128}
\scriptsize
\setlength{\tabcolsep}{4pt}
\renewcommand{\arraystretch}{1.08}
\begin{tabular}{l|c|c|c|c|c}
\toprule
\rowcolor{tabheader}
Model & Rank & $d_s$ & Layer time (s) & Rel. err. & Ratio \\
\midrule
\rowcolor{oursrow}
Mamba-2 & $1$ & $64$  & $34.2$  & $0.010$ & $1.00\times$ \\
Mamba-3 SISO    & $1$ & $64$  & $130.3$ & $0.059$ & $3.81\times$ \\
Mamba-3 MIMO-R4 & $4$ & $64$  & $274.6$ & $0.060$ & $8.03\times$ \\
\midrule
\rowcolor{oursrow}
Mamba-2 & $1$ & $128$ & $48.3$  & $0.004$ & $1.00\times$ \\
Mamba-3 SISO    & $1$ & $128$ & $179.1$ & $0.059$ & $3.71\times$ \\
Mamba-3 MIMO-R4 & $4$ & $128$ & $405.8$ & $0.060$ & $8.39\times$ \\
\bottomrule
\end{tabular}
\end{table}

% ============================================================

\section{Training and Accuracy}
\label{app:training-accuracy}

\subsection{Approximation-Aware Fine-Tuning Procedure}
\label{app:cotrain}

The AAT procedure described in \S\ref{sec:workload-setup} starts from the same task-specific $12$-layer student used for exact evaluation. Before fine-tuning, it replaces every exact Mamba nonlinearity with its deployment polynomial and calibrates the bucket boundaries once on 500 training samples, after which the boundaries remain frozen. Table~\ref{tab:cotrain-details} summarizes the configuration.

\begin{table}[ht]
\centering
\caption{Model construction and approximation configuration.}
\label{tab:cotrain-details}
\scriptsize
\setlength{\tabcolsep}{4pt}
\renewcommand{\arraystretch}{1.1}
\begin{tabular}{@{}l|l@{}}
\toprule
\rowcolor{tabheader}
Setting & Value \\
\midrule
Teacher initialization & \texttt{state-spaces/mamba2-130m}  \\
Teacher adaptation & One task epoch \\
Student distillation & One task epoch, $T{=}2$, $\lambda{=}0.5$ \\
Exact/AAT initialization & Task-specific $12$-layer student \\
Optimizer & AdamW \\
Learning rate & $1 \times 10^{-5}$ \\
Weight decay & 0.01 \\
LR schedule & Linear warmup (6\%) then linear decay \\
Precision & AMP (fp16 optimizer, fp32 master weights) \\
Activation swap & Before fine-tuning, not during \\
Bucket calibration & Single pass, 500 train samples, frozen \\
\bottomrule
\end{tabular}
\end{table}

With the approximations active from the first gradient update, the model can adapt to their error throughout fine-tuning. Post-swap evaluation instead applies them to an already fine-tuned exact checkpoint without further optimization, providing a lower-bound deployment scenario without learned compensation.

% ============================================================

\subsection{Encrypted Accuracy Audit Scope}
\label{app:encrypted-audit-scope}

The encrypted audit runs $2{,}670$ stratified long-document queries per seed end-to-end on the deployed native encrypted backend. The full test splits contain $1{,}400$ SCOTUS, $2{,}500$ arXiv, and $5{,}000$ Patent examples, so evaluating every example at the three deployed lengths $L\in\{512,1024,2048\}$ would require $26{,}700$ encrypted queries per seed.
Table~\ref{tab:encrypted_audit_runtime} converts the measured latency rows into native encrypted wall-clock time and motivates the stratified $10\%$ subset as the accuracy-evaluation unit.
On this paired subset, Encrypted \system{} stays within $0.48$ percentage points of Plain \system{} in every row and within $0.15$ percentage points in aggregate.

\begin{table}[t]
\centering
\caption{Serial encrypted-evaluation runtime implied by Table~\ref{tab:e2e_comparison}.}
\label{tab:encrypted_audit_runtime}
\scriptsize
\setlength{\tabcolsep}{2.4pt}
\renewcommand{\arraystretch}{1.08}
\begin{adjustbox}{max width=\columnwidth}
\begin{tabular}{@{}l|l|c|cc|cc@{}}
\toprule
\rowcolor{tabheader}
& & & \multicolumn{2}{c|}{1 GPU} & \multicolumn{2}{c}{4 GPUs} \\
\rowcolor{tabheader}
System & Scope & Queries/seed & 1 seed & 3 seeds & 1 seed & 3 seeds \\
\midrule
\multirow{2}{*}{AEGIS$^\dagger$~\cite{gong2026aegis}} & Full test & $26{,}700$ & $2{,}865.4$\,d ($7.84$\,y) & $23.5$\,y & $745.6$\,d ($2.04$\,y) & $6.12$\,y \\
 & 10\% audit & $2{,}670$ & $286.5$\,d & $859.6$\,d ($2.35$\,y) & $74.6$\,d & $223.7$\,d \\
\midrule
\multirow{2}{*}{\system{}} & Full test & $26{,}700$ & $823.3$\,d ($2.26$\,y) & $6.76$\,y & $385.7$\,d ($1.06$\,y) & $3.17$\,y \\
 & \cellcolor{oursrow}10\% audit & \cellcolor{oursrow}$2{,}670$ & \cellcolor{oursrow}$82.3$\,d & \cellcolor{oursrow}$247.0$\,d & \cellcolor{oursrow}$38.6$\,d & \cellcolor{oursrow}$115.7$\,d \\
\bottomrule
\end{tabular}
\end{adjustbox}
\vspace{1pt}
\parbox{\linewidth}{\scriptsize $^{\dagger}$AEGIS one-GPU entries are serial-equivalent values computed from paper-reported multi-GPU rows and the interpolation below, not native one-GPU runs.}
\end{table}

The AEGIS row gives an illustrative cross-system estimate of the cost scale for encrypted long-sequence evaluation. Applying its reported per-query latencies from Table~\ref{tab:gpu_scaling}, with the missing $L{=}1024$ point geometrically interpolated from the $L{=}512$ and $L{=}2048$ rows, to the same query count would put a single-seed full-test sweep at $7.84$ years on one GPU and $2.04$ years on four, about $3.5\times$ and $1.9\times$ longer than \system{}.

We also audit encrypted-versus-plaintext residual-stream agreement at all three deployed lengths. Figure~\ref{fig:layerwise-3task-L512} shows that the cumulative residual error grows monotonically with depth, with the growth rate stabilizing across all $12$ layers.

\begin{figure}[ht]
\centering
\includegraphics[width=0.9\linewidth]{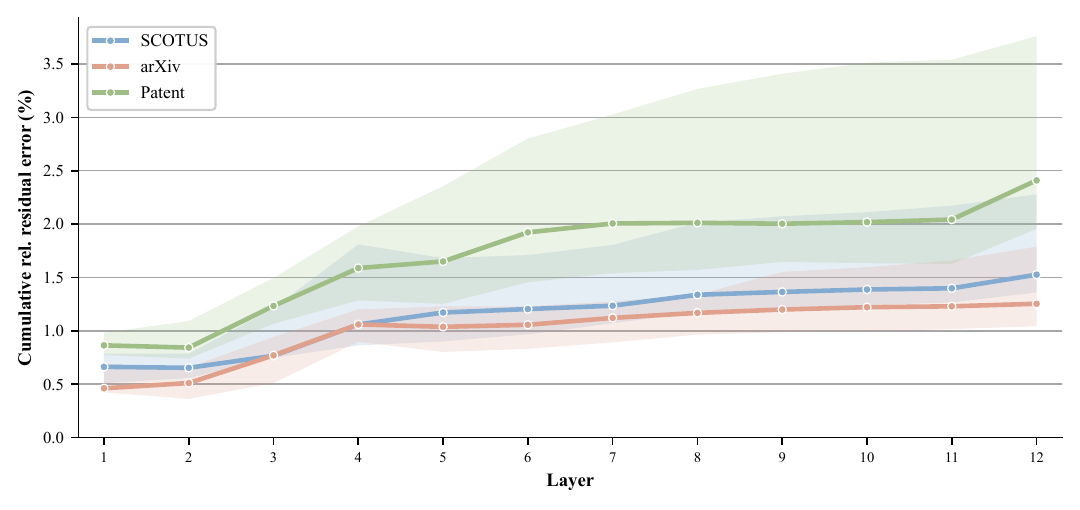}
\caption{Layerwise residual-stream error on the encrypted audit subset.}
\label{fig:layerwise-3task-L512}
\end{figure}

\subsection{Complete Accuracy Tables}
\label{app:full-results}

This subsection gives the complete values behind the main evaluation figures: Table~\ref{tab:exact_baselines} reports the three-seed unapproximated accuracy in Figure~\ref{fig:exact_accuracy}, Table~\ref{tab:fesc_network_comm} reports the per-length communication and latency in Figure~\ref{fig:network_cost}, and Table~\ref{tab:scan_residency} reports the scan residency and key-switch counts in Figure~\ref{fig:scan_residency}.
For Table~\ref{tab:scan_residency}, the peak column records the logged live ciphertext inventory, while the result column applies the one-A100 memory envelope and CKKS chain-depth constraints.

The approximation sweeps use paired controls: Exact and AAT start from the same task-specific distilled \mambabase{} checkpoint and receive one additional exact or approximation-aware epoch, while BERT uses \texttt{bert-base-uncased} trained directly for two epochs.
Table~\ref{tab:degree_selection} varies MoSiLU and MoSoft with RMSNorm and the exponential exact; Table~\ref{tab:allapprox} varies the NeRMS Newton count with both fixed at D4; and Table~\ref{tab:sweep_summary} varies the MMExp degree with both fixed.

\begin{table}[ht]
\centering
\caption{Three-seed mean task accuracy [\%].}
\label{tab:exact_baselines}
\scriptsize
\setlength{\tabcolsep}{4pt}
\renewcommand{\arraystretch}{1.1}
\begin{adjustbox}{max width=\columnwidth}
\begin{tabular}{@{}l|l|ccccc@{}}
\toprule
\rowcolor{tabheader}
Dataset & Model & $L{=}128$ & $L{=}256$ & $L{=}512$ & $L{=}1024$ & $L{=}2048$ \\
\midrule
\multirow{2}{*}{SCOTUS} & BERT-base & $48.78$ & $58.38$ & $64.07$ & $64.36$ & $66.36$ \\
                         & \mambabase{} & $44.65$ & $58.86$ & $68.36$ & $70.50$ & $72.50$ \\
\midrule
\multirow{2}{*}{arXiv}  & BERT-base & $86.00$ & $86.40$ & $86.56$ & $87.76$ & $88.36$ \\
                         & \mambabase{} & $85.48$ & $87.60$ & $88.84$ & $89.96$ & $90.28$ \\
\midrule
\multirow{2}{*}{Patent} & BERT-base & $60.47$ & $64.77$ & $65.78$ & $66.98$ & $67.72$ \\
                         & \mambabase{} & $61.22$ & $64.48$ & $65.88$ & $67.44$ & $68.12$ \\
\bottomrule
\end{tabular}
\end{adjustbox}
\end{table}

\begin{table}[ht]
\centering
\caption{\system{} communication and end-to-end latency.}
\label{tab:fesc_network_comm}
\scriptsize
\setlength{\tabcolsep}{2.8pt}
\renewcommand{\arraystretch}{1.12}
\begin{adjustbox}{max width=\columnwidth}
\begin{tabular}{@{}c|c|cccc@{}}
\toprule
\rowcolor{tabheader}
 & & \multicolumn{4}{c}{Latency (min)} \\
\rowcolor{tabheader}
$L$ & Online (GB) & LAN & WAN1 & WAN2 & WAN3 \\
\midrule
$128$ & $4.7$ & $4.4$ & $5.5$ & $10.5$ & $12.1$ \\
$256$ & $9.3$ & $9.0$ & $11.1$ & $21.0$ & $22.6$ \\
$512$ & $18.5$ & $17.8$ & $21.9$ & $41.7$ & $43.3$ \\
$1{,}024$ & $36.9$ & $38.1$ & $46.1$ & $85.8$ & $87.4$ \\
$2{,}048$ & $73.8$ & $77.3$ & $93.2$ & $172.5$ & $174.1$ \\
\bottomrule
\end{tabular}
\end{adjustbox}
\end{table}

\begin{table}[ht]
\centering
\caption{Resident HEScan kernel schedules at $L{=}2048$.}
\label{tab:scan_residency}
\scriptsize
\setlength{\tabcolsep}{2.6pt}
\renewcommand{\arraystretch}{1.12}
\begin{adjustbox}{max width=\columnwidth}
\begin{tabular}{@{}l|cc|ccc|l@{}}
\toprule
\rowcolor{tabheader}
Scan schedule & $B_{\mathrm{blk}}$ & $K_{\mathrm{blk}}$ & Peak (GB) & HEScan KSw & Time (s/layer) & Result \\
\midrule
Two-pass & $64$ & $32$ & --- & $370\mathrm{K}$ & --- & chain exhausted \\
Two-pass & $128$ & $16$ & --- & $372\mathrm{K}$ & --- & chain exhausted \\
Two-pass & $256$ & $8$ & $12.8$ & $372\mathrm{K}$ & $302.1$ & feasible, slower \\
Two-pass & $512$ & $4$ & $19.5$ & $370\mathrm{K}$ & $296.3$ & feasible, slower \\
\cellcolor{oursrow}Two-pass & \cellcolor{oursrow}$1024$ & \cellcolor{oursrow}$2$ & \cellcolor{oursrow}$32.7$ & \cellcolor{oursrow}$372\mathrm{K}$ & \cellcolor{oursrow}$284.2$ & \cellcolor{oursrow}default \\
Full-length & $2048$ & $1$ & --- & $283\mathrm{K}$ & --- & OOM \\
\bottomrule
\end{tabular}
\end{adjustbox}
\end{table}

\phantomsection
\label{app:nonlinear-sweeps}
\begin{table*}[t]
\centering
\caption{MoSiLU/MoSoft degree sweep with three-seed deviations.}
\label{tab:degree_selection}
\scriptsize
\setlength{\tabcolsep}{2.2pt}
\renewcommand{\arraystretch}{1.1}
\begin{adjustbox}{max width=\textwidth}
\begin{tabular}{@{}l|c|cc|ccc|ccc@{}}
\toprule
\rowcolor{tabheader}
& & \multicolumn{2}{c|}{Baselines} & \multicolumn{3}{c|}{Post-Swap} & \multicolumn{3}{c}{AAT} \\
\rowcolor{tabheader}
Dataset & $L$ & BERT-base & Mamba-base & D2 & D4 & D6 & D2 & D4 & D6 \\
\midrule
\multirow{3}{*}{SCOTUS} & 512 & $64.07 \pm 2.20$ & $68.36 \pm 2.51$ & $68.00 \pm 0.37$ & $68.29 \pm 0.18$ & $68.50 \pm 0.15$ & $67.93 \pm 0.25$ & $68.43 \pm 0.25$ & $68.57 \pm 0.29$ \\
 & 1024 & $64.36 \pm 0.70$ & $70.50 \pm 2.27$ & $68.93 \pm 0.61$ & $70.43 \pm 0.46$ & $70.29 \pm 0.53$ & $69.14 \pm 0.82$ & $70.29 \pm 0.08$ & $70.21 \pm 0.27$ \\
 & 2048 & $66.36 \pm 2.07$ & $72.50 \pm 1.22$ & $70.86 \pm 0.12$ & $72.64 \pm 0.39$ & $72.71 \pm 0.35$ & $71.00 \pm 0.14$ & $72.71 \pm 1.39$ & $72.50 \pm 1.20$ \\
\midrule
\multirow{3}{*}{arXiv} & 512 & $86.56 \pm 1.31$ & $88.84 \pm 0.29$ & $88.92 \pm 0.42$ & $88.88 \pm 0.46$ & $88.88 \pm 0.50$ & $89.28 \pm 0.26$ & $88.96 \pm 0.22$ & $89.00 \pm 0.33$ \\
 & 1024 & $87.76 \pm 0.82$ & $89.96 \pm 0.42$ & $89.60 \pm 0.24$ & $89.96 \pm 0.34$ & $89.88 \pm 0.34$ & $89.92 \pm 0.16$ & $89.96 \pm 0.04$ & $89.88 \pm 0.12$ \\
 & 2048 & $88.36 \pm 0.77$ & $90.28 \pm 0.31$ & $89.96 \pm 0.09$ & $90.28 \pm 0.16$ & $90.28 \pm 0.20$ & $90.24 \pm 0.43$ & $90.36 \pm 0.17$ & $90.48 \pm 0.45$ \\
\midrule
\multirow{3}{*}{Patent} & 512 & $65.78 \pm 0.63$ & $65.88 \pm 0.49$ & $64.64 \pm 0.13$ & $65.84 \pm 0.42$ & $65.80 \pm 0.37$ & $65.34 \pm 0.43$ & $65.90 \pm 0.14$ & $65.90 \pm 0.18$ \\
 & 1024 & $66.98 \pm 0.67$ & $67.44 \pm 0.54$ & $66.20 \pm 0.34$ & $67.50 \pm 0.08$ & $67.44 \pm 0.13$ & $67.54 \pm 0.17$ & $67.60 \pm 0.12$ & $67.46 \pm 0.14$ \\
 & 2048 & $67.72 \pm 0.95$ & $68.12 \pm 0.50$ & $67.84 \pm 0.34$ & $68.14 \pm 0.29$ & $68.22 \pm 0.31$ & $68.14 \pm 0.05$ & $68.12 \pm 0.10$ & $68.18 \pm 0.17$ \\
\midrule
\multicolumn{2}{@{}l|}{MPC online ops} & --- & --- & $5$ & $6$ & $7$ & $5$ & $6$ & $7$ \\
\multicolumn{2}{@{}l|}{MPC online layers} & --- & --- & $2$ & $3$ & $4$ & $2$ & $3$ & $4$ \\
\midrule
\multicolumn{2}{@{}l|}{Avg $\Delta$ vs BERT-base} & --- & $+2.66$ & $+1.89$ & $+2.67$ & $+2.67$ & $+2.29$ & $\mathbf{+2.71}$ & $+2.69$ \\
\multicolumn{2}{@{}l|}{Avg $\Delta$ vs Mamba-base} & $-2.66$ & --- & $-0.77$ & $+0.01$ & $+0.01$ & $-0.37$ & $\mathbf{+0.05}$ & $+0.03$ \\
\bottomrule
\end{tabular}
\end{adjustbox}
\end{table*}

\begin{table*}[t]
\centering
\caption{NeRMS iteration sweep with three-seed deviations.}
\label{tab:allapprox}
\scriptsize
\setlength{\tabcolsep}{2.2pt}
\renewcommand{\arraystretch}{1.1}
\begin{adjustbox}{max width=\textwidth}
\begin{tabular}{@{}l|c|cc|ccc|ccc@{}}
\toprule
\rowcolor{tabheader}
& & \multicolumn{2}{c|}{Baselines} & \multicolumn{3}{c|}{Post-Swap} & \multicolumn{3}{c}{AAT} \\
\rowcolor{tabheader}
Dataset & $L$ & BERT-base & Mamba-base & N1 & N2 & N3 & N1 & N2 & N3 \\
\midrule
\multirow{3}{*}{SCOTUS} & 512 & $64.07 \pm 2.20$ & $68.36 \pm 2.51$ & $65.86 \pm 0.11$ & $65.86 \pm 0.46$ & $66.14 \pm 0.45$ & $68.29 \pm 0.62$ & $67.57 \pm 0.85$ & $67.14 \pm 0.80$ \\
 & 1024 & $64.36 \pm 0.70$ & $70.50 \pm 2.27$ & $66.64 \pm 0.29$ & $66.71 \pm 0.19$ & $66.64 \pm 0.33$ & $68.71 \pm 0.52$ & $68.71 \pm 0.41$ & $68.43 \pm 0.59$ \\
 & 2048 & $66.36 \pm 2.07$ & $72.50 \pm 1.22$ & $67.21 \pm 0.62$ & $67.50 \pm 0.37$ & $67.79 \pm 0.44$ & $70.07 \pm 0.95$ & $69.71 \pm 0.70$ & $69.36 \pm 0.83$ \\
\midrule
\multirow{3}{*}{arXiv} & 512 & $86.56 \pm 1.31$ & $88.84 \pm 0.29$ & $88.60 \pm 0.63$ & $88.76 \pm 0.72$ & $88.80 \pm 0.61$ & $88.84 \pm 1.09$ & $88.20 \pm 1.21$ & $88.16 \pm 1.03$ \\
 & 1024 & $87.76 \pm 0.82$ & $89.96 \pm 0.42$ & $89.28 \pm 0.24$ & $89.44 \pm 0.29$ & $89.64 \pm 0.18$ & $89.80 \pm 0.46$ & $89.80 \pm 0.54$ & $89.36 \pm 0.36$ \\
 & 2048 & $88.36 \pm 0.77$ & $90.28 \pm 0.31$ & $90.12 \pm 0.17$ & $90.08 \pm 0.08$ & $90.12 \pm 0.10$ & $90.04 \pm 0.39$ & $89.32 \pm 0.29$ & $8.48^{\dagger} \pm 1.31$ \\
\midrule
\multirow{3}{*}{Patent} & 512 & $65.78 \pm 0.63$ & $65.88 \pm 0.49$ & $65.18 \pm 0.17$ & $65.36 \pm 0.32$ & $65.36 \pm 0.30$ & $66.05 \pm 0.52$ & $66.19 \pm 0.65$ & $65.93 \pm 0.60$ \\
 & 1024 & $66.98 \pm 0.67$ & $67.44 \pm 0.54$ & $66.78 \pm 0.23$ & $66.78 \pm 0.24$ & $66.88 \pm 0.22$ & $67.63 \pm 0.46$ & $67.73 \pm 0.47$ & $67.51 \pm 0.42$ \\
 & 2048 & $67.72 \pm 0.95$ & $68.12 \pm 0.50$ & $67.62 \pm 0.06$ & $67.58 \pm 0.02$ & $67.58 \pm 0.09$ & $68.25 \pm 0.20$ & $68.11 \pm 0.15$ & $68.13 \pm 0.23$ \\
\midrule
\multicolumn{2}{@{}l|}{MPC online ops} & --- & --- & $9$ & $12$ & $15$ & $9$ & $12$ & $15$ \\
\multicolumn{2}{@{}l|}{MPC online layers} & --- & --- & $3$ & $6$ & $9$ & $3$ & $6$ & $9$ \\
\midrule
\multicolumn{2}{@{}l|}{Avg $\Delta$ vs BERT-base} & --- & $+2.66$ & $+1.04$ & $+1.12$ & $+1.22$ & $\mathbf{+2.19}$ & $+1.93$ & $-7.27$ \\
\multicolumn{2}{@{}l|}{Avg $\Delta$ vs Mamba-base} & $-2.66$ & --- & $-1.62$ & $-1.53$ & $-1.44$ & $\mathbf{-0.47}$ & $-0.73$ & $-9.93$ \\
\midrule
\multicolumn{2}{@{}l|}{Avg $\Delta$ vs BERT-base (excl.\ $\dagger$)} & --- & $+2.66$ & $+1.04$ & $+1.12$ & $+1.22$ & $\mathbf{+2.19}$ & $+1.93$ & $+1.80$ \\
\multicolumn{2}{@{}l|}{Avg $\Delta$ vs Mamba-base (excl.\ $\dagger$)} & $-2.66$ & --- & $-1.62$ & $-1.53$ & $-1.44$ & $\mathbf{-0.47}$ & $-0.73$ & $-0.95$ \\
\bottomrule
\end{tabular}
\end{adjustbox}
\par\vspace{1pt}
\noindent{\scriptsize\emph{Note.} $\dagger$ marks the collapsed arXiv L2048 N3 AAT cell.}
\end{table*}

\begin{table*}[t]
\centering
\caption{MMExp-degree sweep with three-seed deviations.}
\label{tab:sweep_summary}
\scriptsize
\setlength{\tabcolsep}{2.2pt}
\renewcommand{\arraystretch}{1.1}
\begin{adjustbox}{max width=\textwidth}
\begin{tabular}{@{}l|c|cc|cccc|cccc@{}}
\toprule
\rowcolor{tabheader}
& & \multicolumn{2}{c|}{Baselines} & \multicolumn{4}{c|}{Post-Swap} & \multicolumn{4}{c}{AAT} \\
\rowcolor{tabheader}
Dataset & $L$ & BERT-base & Mamba-base & D2 & D4 & D6 & D8 & D2 & D4 & D6 & D8 \\
\midrule
\multirow{3}{*}{SCOTUS}
 & 512  & $64.07 \pm 2.20$ & $68.36 \pm 2.51$ & $64.69 \pm 0.21$ & $65.79 \pm 0.69$ & $65.86 \pm 1.62$ & $65.86 \pm 0.27$ & $68.07 \pm 0.26$ & $67.71 \pm 2.91$ & $68.36 \pm 1.86$ & $68.43 \pm 2.37$ \\
 & 1024 & $64.36 \pm 0.70$ & $70.50 \pm 2.27$ & $65.33 \pm 0.42$ & $66.07 \pm 0.87$ & $66.64 \pm 2.88$ & $66.64 \pm 2.56$ & $69.52 \pm 0.36$ & $69.57 \pm 2.23$ & $69.93 \pm 0.48$ & $70.07 \pm 1.73$ \\
 & 2048 & $66.36 \pm 2.07$ & $72.50 \pm 1.22$ & $64.86 \pm 0.57$ & $67.21 \pm 0.72$ & $67.29 \pm 1.17$ & $67.21 \pm 0.89$ & $70.60 \pm 1.16$ & $70.93 \pm 1.93$ & $70.93 \pm 0.03$ & $70.93 \pm 1.51$ \\
\midrule
\multirow{3}{*}{arXiv}
 & 512  & $86.56 \pm 1.31$ & $88.84 \pm 0.29$ & $87.31 \pm 0.58$ & $88.60 \pm 0.67$ & $88.60 \pm 1.02$ & $88.60 \pm 1.07$ & $88.93 \pm 0.14$ & $89.08 \pm 0.36$ & $89.40 \pm 0.44$ & $88.52 \pm 1.84$ \\
 & 1024 & $87.76 \pm 0.82$ & $89.96 \pm 0.42$ & $88.23 \pm 0.36$ & $89.56 \pm 0.45$ & $89.32 \pm 0.44$ & $89.32 \pm 2.09$ & $62.65^{\dagger} \pm 46.92$ & $89.36 \pm 0.03$ & $87.72 \pm 0.27$ & $88.20 \pm 2.96$ \\
 & 2048 & $88.36 \pm 0.77$ & $90.28 \pm 0.31$ & $88.53 \pm 0.24$ & $90.08 \pm 0.30$ & $90.16 \pm 0.48$ & $90.12 \pm 1.40$ & $62.81^{\dagger} \pm 47.06$ & $89.80 \pm 1.49$ & $89.68 \pm 0.28$ & $89.96 \pm 1.28$ \\
\midrule
\multirow{3}{*}{Patent}
 & 512  & $65.78 \pm 0.63$ & $65.88 \pm 0.49$ & $64.15 \pm 0.34$ & $65.12 \pm 0.49$ & $65.10 \pm 0.50$ & $65.12 \pm 1.45$ & $66.01 \pm 0.30$ & $66.06 \pm 0.82$ & $66.08 \pm 0.21$ & $66.14 \pm 1.69$ \\
 & 1024 & $66.98 \pm 0.67$ & $67.44 \pm 0.54$ & $65.63 \pm 0.09$ & $66.84 \pm 1.03$ & $66.82 \pm 0.81$ & $66.80 \pm 1.91$ & $67.16 \pm 0.27$ & $67.00 \pm 0.26$ & $67.44 \pm 0.20$ & $66.86 \pm 2.47$ \\
 & 2048 & $67.72 \pm 0.95$ & $68.12 \pm 0.50$ & $66.59 \pm 0.18$ & $67.64 \pm 0.12$ & $67.62 \pm 0.48$ & $67.66 \pm 2.91$ & $49.35^{\dagger} \pm 29.68$ & $67.60 \pm 0.15$ & $68.20 \pm 2.10$ & $67.70 \pm 1.96$ \\
\midrule
\multicolumn{2}{@{}l|}{MPC online ops} & --- & --- & $5$ & $7$ & $9$ & $11$ & $5$ & $7$ & $9$ & $11$ \\
\multicolumn{2}{@{}l|}{MPC online layers} & --- & --- & $3$ & $4$ & $5$ & $6$ & $3$ & $4$ & $5$ & $6$ \\
\midrule
\multicolumn{2}{@{}l|}{Avg $\Delta$ vs BERT-base} & --- & $+2.66$ & $-0.29$ & $+1.00$ & $+1.05$ & $+1.04$ & $-5.87$ & $+2.13$ & $\mathbf{+2.20}$ & $+2.10$ \\
\multicolumn{2}{@{}l|}{Avg $\Delta$ vs Mamba-base} & $-2.66$ & --- & $-2.95$ & $-1.66$ & $-1.61$ & $-1.62$ & $-8.53$ & $-0.53$ & $\mathbf{-0.46}$ & $-0.56$ \\
\midrule
\multicolumn{2}{@{}l|}{Avg $\Delta$ vs BERT-base (excl.\ $\dagger$)} & --- & $+2.66$ & $-0.29$ & $+1.00$ & $+1.05$ & $+1.04$ & $\mathbf{+2.70}$ & $+2.13$ & $+2.20$ & $+2.10$ \\
\multicolumn{2}{@{}l|}{Avg $\Delta$ vs Mamba-base (excl.\ $\dagger$)} & $-2.66$ & --- & $-2.95$ & $-1.66$ & $-1.61$ & $-1.62$ & $-0.54$ & $-0.53$ & $\mathbf{-0.46}$ & $-0.56$ \\
\bottomrule
\end{tabular}
\end{adjustbox}
\par\vspace{1pt}
\noindent{\scriptsize\emph{Note.} $\dagger$ marks collapsed D2 AAT cells. The excluding-dagger rows summarize stable cells.}
\end{table*}

\end{document}